\documentclass[12pt,letterpaper]{article}

\usepackage{amsmath,amssymb,amsfonts,amsthm,bbm}
\usepackage{mathtools}   

\usepackage{enumitem}   

\usepackage{authblk} 

\usepackage[titletoc,title]{appendix}  

\usepackage{kpfonts}
\usepackage[T1]{fontenc}

\usepackage[table,xcdraw,dvipsnames]{xcolor}   
\usepackage{hyperref}
\newcommand\myshade{85}
\colorlet{mylinkcolor}{YellowOrange}
\colorlet{mycitecolor}{Aquamarine}
\colorlet{myurlcolor}{violet}

\hypersetup{
  linkcolor  = mylinkcolor!\myshade!black,
  citecolor  = mycitecolor!\myshade!black,
  urlcolor   = myurlcolor!\myshade!black,
  colorlinks = true,
}

\usepackage{caption}
\usepackage{subcaption}

\usepackage{multirow}
\usepackage{graphicx}

\usepackage{algorithm}
\usepackage{algpseudocode}

\renewcommand{\hat}{\widehat}
\renewcommand{\tilde}{\widetilde}

\newcommand{\bfm}[1]{\ensuremath{\boldsymbol{#1}}} 

\def\bzero{\bfm 0}

\def\bbone{\mathbbm{1}} 

\def\ba{\bfm a}   \def\bA{\bfm A}  
   \def\bB{\bfm B}  
   \def\bC{\bfm C}  
   \def\bD{\bfm D}  
\def\be{\bfm e}   \def\bE{\bfm E}  \def\EE{\mathbb{E}}
\def\bff{\bfm f}  \def\bF{\bfm F}  
   \def\bG{\bfm G}  
   \def\bH{\bfm H}  
   \def\bI{\bfm I}

   \def\bM{\bfm M}  \def\MM{\mathbb{M}}

   \def\bP{\bfm P}  
   \def\bQ{\bfm Q}  
   \def\bR{\bfm R}  \def\RR{\mathbb{R}}
\def\bs{\bfm s}     
     
\def\bu{\bfm u}   \def\bU{\bfm U}  
\def\bv{\bfm v}   \def\bV{\bfm V}  
\def\bw{\bfm w}     
\def\bx{\bfm x}   \def\bX{\bfm X}  
\def\by{\bfm y}   \def\bY{\bfm Y}  
\def\bz{\bfm z}   \def\bZ{\bfm Z}

\def\calE{{\cal  E}} 
\def\calF{{\cal  F}}

\def\calN{{\cal  N}}

\newcommand{\bfsym}[1]{\ensuremath{\boldsymbol{#1}}}

 \def\balpha{\bfsym \alpha}
 \def\bbeta{\bfsym \beta}
 \def\bgamma{\bfsym \gamma}             \def\bGamma{\bfsym \Gamma}
 \def\bdelta{\bfsym {\delta}}           \def\bDelta {\bfsym {\Delta}}
 \def\bfeta{\bfsym {\eta}}              
 \def\bmu{\bfsym {\mu}}                 
 
 \def\btheta{\bfsym {\theta}}           \def\bTheta {\bfsym {\Theta}}
 \def\beps{\bfsym \varepsilon}          
 \def\bepsilon{\bfsym \varepsilon}
              \def\bSigma{\bfsym \Sigma}
 \def\blambda {\bfsym {\lambda}}        \def\bLambda {\bfsym {\Lambda}}

          \def\bPhi{\bfsym {\Phi}}

\def\bLam {\bLambda}
\def\one {\mathbbm{1}}
  \def\mE{\mathcal{E}}
  \def\ve{\varepsilon}
  \def\mF{\mathcal{F}}

\def\mC{\mathcal C}
\def\mV{\mathcal V}

\providecommand{\abs}[1]{\left\lvert#1\right\rvert}
\providecommand{\norm}[1]{\left\lVert#1\right\rVert}

\providecommand{\paran}[1]{\left( #1 \right)}
\providecommand{\brackets}[1]{\left[ #1 \right]}
\providecommand{\braces}[1]{\left\{ #1 \right\}}

\providecommand{\defeq}{\triangleq}

\usepackage{mathtools}
\DeclarePairedDelimiterX{\infdivx}[2]{(}{)}{%
  #1 \; \delimsize\| \; #2%
}

\DeclareMathOperator{\cov}{cov}

\DeclareMathOperator{\diag}{diag}
\newcommand{\E}[1]{{\mathbb{E}} \left[ #1 \right]}

\DeclareMathOperator{\Var}{Var}
\DeclareMathOperator{\var}{var}

\DeclareMathOperator{\Tr}{Tr}

\newtheorem{definition}{Definition}[section]
\newtheorem{assumption}[definition]{Assumption}
\newtheorem{lemma}[definition]{Lemma}

\newtheorem{theorem}[definition]{Theorem}
\newtheorem{corollary}[definition]{Corollary}

\theoremstyle{definition}

\newcommand{\bigO}[1]{ \mathcal{O} \left( #1 \right) }

\newcommand{\smlo}[1]{{\rm o} \left( #1 \right) }

\newcommand{\Op}[1]{{\mathcal{O}_p} \left( #1 \right) }
\newcommand{\op}[1]{{\rm o_p} \left( #1 \right) }

\definecolor{royalpurple}{rgb}{0.47, 0.32, 0.66}

\def\beq{\begin{equation}}
\def\eeq{\end{equation}}

\def\bet{\begin{theorem}}
\def\eet{\end{theorem}}

\def\bel{\begin{lemma}}
\def\eel{\end{lemma}}

\def\eps{\varepsilon}
\def\blam{\blambda}

\def\vec{\mbox{vec}}
\def\wt{\widetilde}

\def\wh{\widehat}

\def\zero{\mathbf{0}}
\def\tr{\mbox{tr}}

\usepackage[framemethod=TikZ]{mdframed} 

\usepackage{fancybox}

\usepackage[top=1in, bottom=1in, left=1in, right=1in]{geometry}

\usepackage{setspace}
\usepackage{etoolbox}
\BeforeBeginEnvironment{equation*}{\begin{singlespace}\vspace*{-\baselineskip}}
	\AfterEndEnvironment{equation*}{\end{singlespace}\noindent\ignorespaces}
\BeforeBeginEnvironment{equation}{\begin{singlespace}\vspace*{-\baselineskip}}
	\AfterEndEnvironment{equation}{\end{singlespace}\noindent\ignorespaces}
\BeforeBeginEnvironment{align}{\begin{singlespace}\vspace*{-\baselineskip}}
	\AfterEndEnvironment{align}{\end{singlespace}\noindent\ignorespaces}
\BeforeBeginEnvironment{align*}{\begin{singlespace}\vspace*{-\baselineskip}}
	\AfterEndEnvironment{align*}{\end{singlespace}\noindent\ignorespaces}
\BeforeBeginEnvironment{eqnarray}{\begin{singlespace}\vspace*{-\baselineskip}}
	\AfterEndEnvironment{eqnarray}{\end{singlespace}\noindent\ignorespaces}
\BeforeBeginEnvironment{eqnarray*}{\begin{singlespace}\vspace*{-\baselineskip}}
	\AfterEndEnvironment{eqnarray*}{\end{singlespace}\noindent\ignorespaces}

\usepackage[authoryear]{natbib}  
\newcommand{\mybibsty}{chicago}
\newcommand{\mybib}{main}

\usepackage{graphicx}
\graphicspath{ {figs/} }

\begin{document}
%
%

\newcommand{\blind}{0}

\if0\blind
{
	\title{\bf Community Network Auto-Regression for High-Dimensional Time Series}
	\author[1]{Elynn Y. Chen \thanks{Elynn Chen's research is supported in part by NSF Grant DMS-1803241.}}
	\author[2]{Jianqing Fan \thanks{Jianqing Fan's research research is supported in part by NSF Grants DMS-1662139 and DMS-DMS-1712591 and NIH grants 5R01-GM072611-12. }}
	\author[3]{Xuening Zhu \thanks{Xuening Zhu is the corresponding author. Her research is supported in part by the National Natural Science Foundation of China (nos. 11901105, 71991472, U1811461), the Shanghai Sailing Program for Youth Science and Technology Excellence (19YF1402700), and the Fudan-Xinzailing Joint Research Centre for Big Data, School of Data Science, Fudan University.
Email: \url{xueningzhu@fudan.edu.cn}.}}
    \affil[1]{University of California, Berkeley}
	\affil[2]{Princeton University}
    \affil[3]{Fudan University}
	\date{\today}
	\maketitle
} \fi

\if1\blind
{
	\bigskip
	\bigskip
	\bigskip
	\title{\bf ...}
	\date{\vspace{-5ex}}
	\maketitle
	\medskip
} \fi

\begin{abstract}
\singlespacing
Modeling responses on the nodes of a large-scale network is an important task that arises commonly in practice.
This paper proposes a community network vector autoregressive (CNAR) model, which utilizes the network structure to characterize the dependence and intra-community homogeneity of the high dimensional time series.
The CNAR model greatly increases the flexibility and generality of the network vector autoregressive \cite[NAR]{zhu2017network} model by allowing heterogeneous network effects across different network communities.
In addition, the non-community-related latent factors are included to account for unknown cross-sectional dependence.
The number of network communities can diverge as the network expands, which leads to estimating a diverging number of model parameters.
We obtain a set of stationary conditions and develop an efficient two-step weighted least-squares estimator.
The consistency and asymptotic normality properties of the estimators are established.
The theoretical results show that the two-step estimator improves the one-step estimator by an order of magnitude when the error admits a factor structure.
The advantages of the CNAR model are further illustrated on a variety of synthetic and real datasets.

\vspace{2ex}

\noindent {\bf KEY WORDS: } Network autoregression; Community structure; Common latent factors; High dimensional time series; VAR model.

\end{abstract}

%
%

\section{Introduction}

Consider a large scale network with $N$ nodes, which are indexed from $i=1,\cdots, N$.
To characterize the network relationship among the nodes,
we use an adjacency matrix $\bA = (a_{ij}) \in \RR^{N\times N}$,
where $a_{ij} = 1$ if the $i$-th node is connected the $j$-th node,
otherwise $a_{ij} = 0$.
For the $i$-th node at time $t$, we collect a continuous type response $y_{it}\in\RR$.
At time $t$, all the responses $y_{it}$ constitute
a high-dimensional vector
$\by_t = (y_{1t}, \cdots, y_{Nt})^\top\in\RR^N$.
We aim to model the temporal dynamics of $\by_t$.  A simple VAR(1) model has $N^2$ autoregressive coefficient matrix and the community structure can be used to reduce the dimensionality.

A popular method of modeling the dynamics of $\by_t$ with underlying network structure is the network autoregression (NAR, \cite{zhu2017network}) model:
\begin{equation} \label{eqn:nar}
\by_t = \beta_1 \tilde\bA \by_{t-1} + \beta_2 \by_{t-1}  + \bZ_t \bgamma + \bepsilon_t,
\end{equation}
where $\tilde\bA = (\tilde a_{ij})$ is the row-normalized adjacency matrix,
with $\tilde a_{ij} = a_{ij}/n_i$ and $n_i = \sum_{j}a_{ij}$,
$\bZ_t\in \RR^p$ is a vector of auxiliary covariates, and
$\beta_1, \beta_2\in\RR$ and $\bgamma\in\RR^p$ are unknown parameters.
The model directly embeds the network adjacency matrix $\bA$ and provides easy interpretations.
Particularly, the network effect $\beta_1$ reflects the influence of connected nodes through their averages at time $t-1$, and the momentum effect $\beta_2$ quantifies the autoregressive effects from the same node.
The NAR model and its variants have been applied to a wide range of fields,
such as social behavior studies \citep{sojourner2013identification,liu2017peer,zhu2018multivariate},
financial risk management \citep{hardle2016tenet,zou2017covariance},
spatial data modeling \citep{lee2009spatial,shi2017spatial},
among others.

Despite its simple form and easy interpretation, the NAR model \eqref{eqn:nar} has only two parameters and may suffer from the risk of model misspecification.
To address this issue, some flexible extensions to the NAR model have been considered in the literature.
For instance,  \cite{dou2016generalized} and
\cite{zhu2019portal}
implement node-specific network effect $\beta_1$ to characterize different inferential powers of different nodes;
\cite{sun2016functional,wang2016nonparametric,sun2018estimation} investigate the nonlinear and nonparametric extensions;
and \cite{sojourner2013identification,liu2017peer,zhu2018multivariate} consider the multivariate responses.
However, previous researches do not fully address two important issues that arise commonly in real applications, namely, heterogeneous network effect and unknown cross-sectional dependence.

First, to characterize the heterogeneous network effect, we consider a community structure among network nodes.
In literature, the community structure is commonly modelled in social networks \citep{rohe2011spectral,zhao2012consistency}.
Individuals in different communities tend to behave quite differently.
Statistically, although the general dependence structure for responses is widely studied,
however, the network structure information is not fully used.
Figure \ref{adj_mat} shows that, in the Chinese A share stock market, the stocks held by the same shareholders may share denser relationships than other stocks.
This motivates us to incorporate the network structure in modeling the dynamics of the individuals.
Second, we consider a dynamic factor structure to allow for a more flexible form of cross-sectional dependence structure
other than those that can be explained by the network structure or observable covariates.
Notably, factor structure \citep{anderson1956statistical} arises commonly in a wide range of applications including economics \citep{stock2011dynamic,bai2008large} and biology \citep{desai2012cross}.
See Chapters 10 and 11 of \cite{fan2020statistical} for the applications of factor models and the references therein.
Previous literature related to NAR makes use of only the adjacency matrix $\bA$, making it difficult to incorporate heterogeneity directly in the model.
In addition, the growing dimensionality $N$ makes it more challenging to deal with the non-diagonal noise covariance matrix $\bSigma_\eps$.

\begin{figure}[htpb!]
    \centering
    \includegraphics[width=0.6\textwidth]{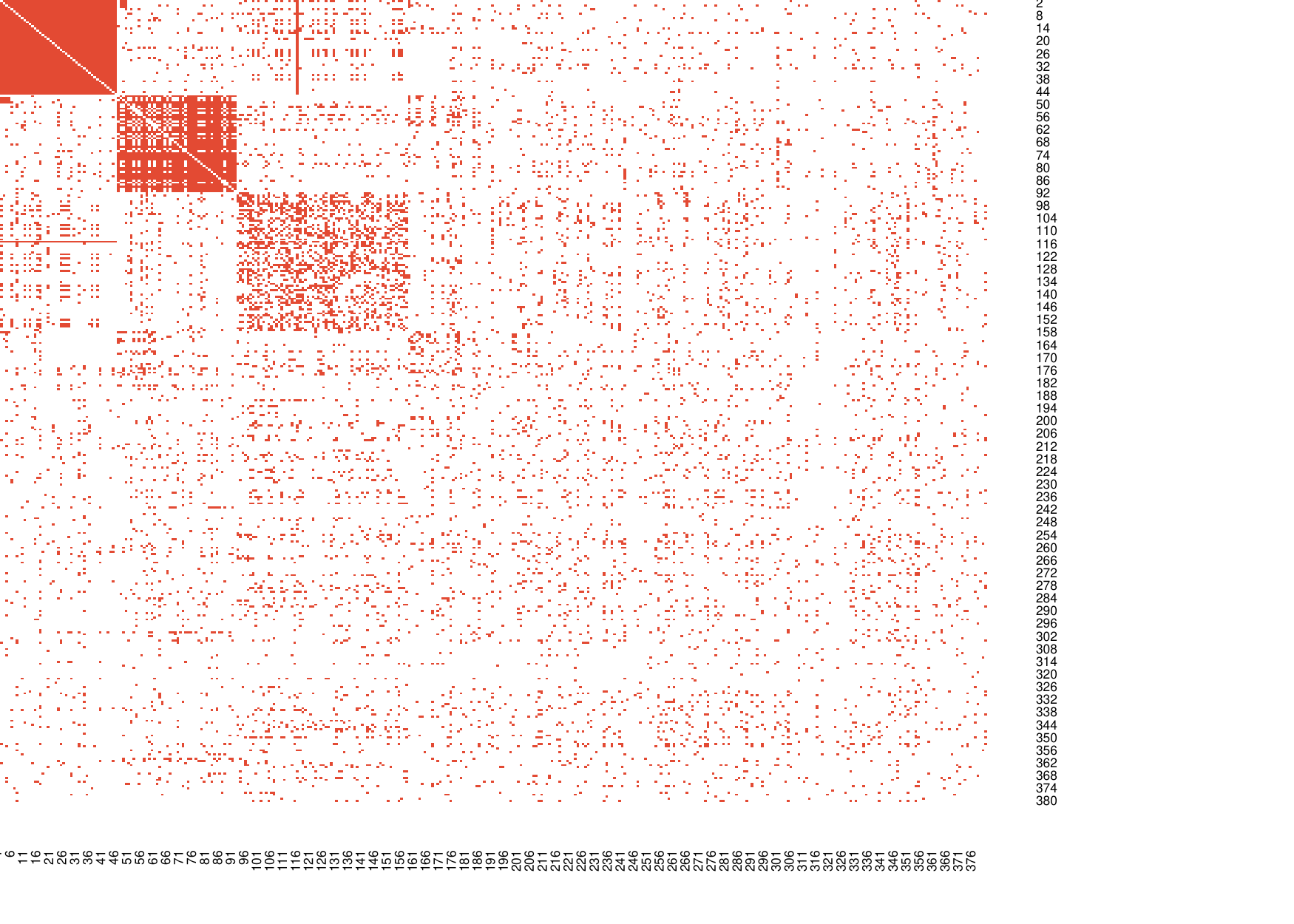}
    \caption{\small Visualization of the adjacency matrix among stocks in Chinese A share market with top 500 market values, where the $(i,j)$th edge is red when the stocks $i$ and $j$ have a common top shareholder.
    }\label{adj_mat}
\end{figure}

In this article, we propose a novel community-augmented network autoregression (CNAR) model that allows for both the heterogeneous network effects across network (i.e., the network coefficient $\beta_1$ are different across different communities) and the unknown dependences in responses $\by_t$.
Specifically, we exploit the network structure deeper through the Stochastic Block Model (SBM), which is widely employed as a canonical model to study clustering and community detection \citep{lei2015consistency,abbe2017entrywise} .
Heterogeneous network effects are enabled by using different network effects for nodes within different communities.
The community structure can be known or unknown as priori.
When it is unknown, we devise an equivalent reformulation of our original CNAR model to circumvent an unnecessary clustering step in our setting, which often incurs extra errors when the ultimate goal is instead to recover the community membership matrix.

At the same time, we allow for a latent cross-sectional dependence of $\by_t$ by modeling the noise $\beps_t$ with a linear factor structure.
Under such setting, the $\ell_2$ norm of the error covariance matrix $\norm{\bSigma_{\eps}}^2 \asymp \Op{N}$ in the genuine presence of the prevasive factors, which renders the ordinary least square estimators converge only at the rate $\Op{1/\sqrt{T}}$ despite of $NT$ observations.
We further propose an efficient two-step estimation procedure and discover surprisingly that it improves the convergence rate to $\Op{1/\sqrt{NT}}$.
The convergence rates and the asymptotic normality of the proposed estimators are presented under the generality that the number of communities $K$ is also growing with $N$.
Furthermore, the convergence of the high dimensional covariance matrix is also established under our modeling framework.
Extensive Monte Carlo simulations is carried out to evaluate the finite sample performance of estimators.
Numerical performances of CNAR 1st-step, 2nd-step estimators and NAR estimators under different simulation models
are compared and evaluated.
Lastly, for illustration propose, we analyze the stock return data from Chinese A stock market.
Empirical analysis on both synthetic and real data set confirms the advantage of the proposed modeling approach.

\subsection{Notation and organization}

Let lowercase letter $x$, boldface letter $\bx$, and  capital letter $\bX$ represent scalar, vector, and matrix, respectively.
We let $C, c, C_0, c_0, \ldots$ denote generic constants, where the uppercase and lowercase letters represent large and small constants, respectively.
The actual values of these generic constants may vary in specific scenarios.
For any matrix $\bX$,  we use $\bx_{i \cdot}$, $\bx_{\cdot j}$, and $x_{ij}$ to denote its $i$-th row, $j$-th column, and $ij$-th entry, respectively.
All vectors are column vectors and row vectors are written as $\bx^\top$ for any vector $\bx$.
Notation $\bbone$ represents a vector or a matrix of all ones of proper size.
In addition, $\be_k$ denotes a unit vector with the $k$th element being 1 and others being 0.

Let $\sigma_i(\bX)$ be the $i$-th largest singular value of $\bX$.
For an arbitrary square matrix $\bX\in\RR^{n\times n}$, let $\lambda_i(\bX)$ be the $i$th eigenvalue of $\bX$ satisfying $|\lambda_1(\bX)|\ge |\lambda_2(\bX)|\ge
\cdots \ge |\lambda_n(\bX)|$.
Specifically, denote $\rho(\bX) = |\lambda_{1}(\bX)|$ as the spectral radius for a square matrix $\bX$. For a symmetric $\bX$, we use $\lambda_{\min}(\bX)$ and $\lambda_{\max}(\bX)$ to denote its smallest and largest eigenvalues {for convenience.}
The following matrix norms are used: maximum norm $\norm{\bX}_{max} \defeq \underset{ij}{\max} \; \lvert x_{ij} \rvert$, $\ell_1$-norm $\norm{\bX}_1 \defeq \underset{j}{\max} \sum_{i} \lvert x_{ij} \rvert$, $\ell_{\infty}$-norm $\norm{\bX}_{\infty} \defeq \underset{i}{\max} \; \sum_{j} \lvert x_{ij} \rvert$,
$\ell_2$-norm $\norm{\bX}$ or $\norm{\bX}_2 \defeq \sigma_1(\bX)$, $(2,\ell_\infty)$-norm $\norm{\bX}_{(2,\ell_\infty)} \defeq \underset{\norm{\ba}_2=1}{\max} \norm{\bX\ba}_\infty = \underset{i}{\max} \norm{\bx_{i \cdot}}_2$,
Frobenius norm $\|\bX\|_F = \Tr(\bX^\top\bX)$,
$\Sigma$-norm $\|\bX\|_\Sigma = N^{-1/2}\|\Sigma^{-1/2}\bX\Sigma^{-1/2}\|_F$.
Here we use $\Tr \paran{\bX}$ to denote the trace of a square matrix $\bX$.


The rest of this paper is organized as follows.
In Section \ref{sec:cnar}, we introduce the model and notations.
In Section \ref{sec:est}, a two step estimation method is presented.
In Section \ref{sec:theory}, we develop the theoretical properties of the estimation.
In Section \ref{sec:simu}, we study the finite sample performance of our estimation via simulation and
Section \ref{sec:appl} provides an empirical study with a stock return dataset.
Section \ref{sec:summ} concludes.
All proofs and technique lemmas are relegate to the supplementary appendix; see Appendix \ref{sec:proof},
\ref{sec:poet-of-hat-eps}.


\section{Network Autoregression with a Community Structure}\label{sec:cnar}

\subsection{Model and notations}\label{sec:model}

We present the model formulation as follows.
First, we assume that the network adjacency matrix $\bA$
is generated by a Stochastic Block Model (SBM) with $N$ nodes and $K$ communities \footnote{This assumption is made largely for a clear model interpretation and theoretical derivation. As shown by its reformulation \eqref{eqn:cnar1} and the empirical examples, the CNAR performances well under more general network structures.}.
The SBM is parameterized by a pair of matrices $\paran{\bTheta, \bQ}$, where $\bTheta \in \MM^{N \times K}$ is the \textit{membership matrix} and $\bQ \in \RR^{K \times K}$ is a symmetric \textit{connectivity matrix}.
For each node $i$, let $k_i \in [K]$ be its community label, such that the $i$-th row of $\bTheta$ is 1 in column $k_i$ and $0$ elsewhere.
Given $\paran{\bTheta, \bQ}$, the adjacency matrix $\bA = \paran{a_{ij}}_{i,j \in [N]}$ is generated as
\begin{equation*}
   a_{ij} = \left \{
   \begin{array}{lll}
   \text{independent } Bernoulli \paran{q_{k_ik_j}}, & \text{if } i < j, \\
   a_{ji}, & \text{if } i > j,\\
   0, & \text{if } i=j.
   \end{array} \right .
\end{equation*}
To model heterogeneity across nodes in the autoregressive patterns,
it is natural to take the above community structure into consideration.
We propose the following community augmented network autoregression (CNAR) model,
\begin{equation} \label{eqn:cnar}
\by_t = \bTheta\bB\bTheta^\top \by_{t-1} + \beta_2 \by_{t-1}
+  \bZ_{t-1}\bgamma +  \bepsilon_t.
\end{equation}
where
$\bB\in\RR^{K\times K}$,
$\beta_2\in\RR$, and
$\bgamma\in\RR^{p}$
are unknown parameters.

The CNAR model (\ref{eqn:cnar}) characterizes the response $\by_t$ by a linear form of the following four components.
The first one is the network component, which characterizes the network effects among different communities.
It shows how the community total (equivalently, average) responses at time $t-1$ impact on the response $\by_t$.
The corresponding parameter $\bB\in\RR^{K\times K}$ is then referred to as {\it community effect.}
The $ij$-th element of $\bB$ represents the effect of $j$-th community to the $i$-th community.
It removes the restrictions of the component $\beta_1 \tilde\bA \by_{t-1}$ in the NAR model \eqref{eqn:nar} to much broader linear dependence (see Remark 1 below) and uses the true community profiles rather than the realized links, which are random under the stochastic block model.
In addition, the asymmetric community effect is allowed because $\bB$ is not restricted to a symmetric matrix.

Next, the {\it momentum effect} $\beta_2$ quantifies how the local node is driven by its historical behavior.
This can be generalized to a diagonal matrix, allowing different nodes to have different momentum parameters.
Lastly, $\bZ_{t-1}\bgamma$ includes covariates before time $t$, which is assumed to be independent with $\{\bepsilon_t\}$.
We further characterize the unknown cross-sectional dependence by assuming a latent factor structure on the noise term:
    \beq \label{eqn:fac_model}
    \bepsilon_t = \bLambda \bff_t + \be_t,
    \eeq
where $\bff_{t}\in\RR^M$ is the unknown latent factors, $\bLambda\in\RR^{N\times M}$ is the corresponding loading matrix, and the innovation term $\be_t\in\RR^N$ follows a multivariate normal distribution $N(\zero,\sigma^2\bI_N)$.
Under such setting, the network nodes are correlated with each other not only through the network relationship, but also through the common driven factors.
A strong factor will have impact on most network nodes through the corresponding factor loadings.
As a result, the factor structure in \eqref{eqn:fac_model} characterizes the cross-sectional dependence at a macro level.
In the following, we discuss the relationship of the proposed model to the existing literature.

{\bf Remark 1.} (Relation to the NAR model).
Note that the unknown parameter matrix $\bB$ specifies the autoregression network effects of $K$ communities.
If $\bTheta \bB \bTheta^\top = \beta_1\wt\bA$, then the CNAR model (\ref{eqn:cnar}) will reduce to the NAR model \eqref{eqn:nar}.
In general,  the above relationship does not hold.
From the aspect of model designing, the CNAR model is built based on the population parameter of $\E{\bA}$,
while the NAR model is constructed conditional on a realized adjacency matrix $\bA$.


{\bf Remark 2.} (Relation to the vector-autoregression and factor models).   CNAR model is a parsimoneous augmented vector autoregression (VAR) model with autoregression coefficient matrix modeled as  $\bTheta\bB\bTheta^\top \by_{t-1} + \beta_2 \bI_N$, augmented by the exogeneous variable $\bZ_{t-1}$ and cross-sectional dependent errors.
{The CNAR model (\ref{eqn:cnar}) could also be comprehended as a special form of factor models.}
Let $\bfeta_t = \bTheta^\top \by_{t}\in\RR^K$
be a projection of $\by_t$ using the membership projection
$\bTheta$, which is the vector of the total response in each community Then we have
\beq
\by_t = \wt\bB\bfeta_{t-1} + \bLambda\bff_{t} +\beta_2\by_{t-1}+\bZ_{t-1}\bgamma + \be_t,\nonumber
\eeq
where  $\wt\bB = \bTheta\bB$.
When the true $\bTheta$ or a consistent estimator $\hat\bTheta$ is available, $\bfeta_{t-1}$ can be treated as a known common factor, which embeds the community structure explicitly in the definition.
As a complement, $\bff_{t}$ consists of unknown factors characterizing
other sources of cross-sectional dependence.

{\bf Remark 3.} (Relation to the {community-leveled} vector autoregression).
By multiplying $\bTheta^\top$ in both sides of the CNAR model, we have
\beq
\bfeta_t =\bD_1\bB\bfeta_{t-1} +\beta_2\bfeta_{t-1}+\bTheta^\top\bZ_{t-1}\bgamma+\bTheta^\top(\bLambda \bff_{t-1} +\be_t),\nonumber
\eeq
where $\bD_1 = \bTheta^\top \bTheta$ is a diagonal matrix, consisting the number of members in each community.
Note that $\bfeta_t$ is a $K$-dimensional vector.
As a result, the CNAR model can be transformed to a {community-leveled} vector autoregression model, where the autoregression coefficients are specified by $\bD_1\bB$ and $\beta_2$.

\subsection{A reformulation of the CNAR model} \label{sec:reform}

Different from the NAR model \eqref{eqn:nar}, the CNAR model \eqref{eqn:cnar} does not directly use the network adjacency matrix $\bA$ or the row-normalized version $\wt \bA$.
Instead, it incorporates the network community structure on a population level.
In practice, we only observe a sampled adjacency matrix $\bA$ and the membership matrix $\bTheta$ is unknown.
Estimation of the network autoregressive coefficients $\bB$ necessitates an estimate of $\bTheta$ at first.
A common method to estimate $\bTheta$ in the SBM is through spectral clustering.
This method typically involves an additional step of clustering, for example by $k$-means, after conducting the eigenvalue decomposition of $\bA$.
In the following, we show that this extra clustering step can be circumvented through an alternative formulation of the CNAR model \eqref{eqn:cnar}.

Recall that, in a SBM, the heuristic of spectral clustering is to relate the eigenvectors of $\bA$ to those of $\bP\,\defeq\,\bTheta \bQ\bTheta^\top$ using the fact that $\E{\bA}=\bP-\diag\paran{\bP}$.
Let $\bP = \bU \bD \bU^\top$ be the eigen-decomposition of $\bP$ with $\bU^\top\bU = \bI_K$ and $\bD \in \RR^{K\times K}$ is a diagonal matrix. It is easy to see that $\bU$ and $\bTheta$ share the same column spaces.
Namely, there exists a rotation matrix $\bR\in\RR^{K\times K}$, such that
\beq
\bTheta = \bU\bR.\nonumber
\eeq
Using this fact, we let $\bB_1 = \bR\bB\bR^\top$ and reparameterize the CNAR model \eqref{eqn:cnar} as
\begin{align} \label{eqn:cnar1}
\by_t 
= \bU \bB_1 \bU^\top\by_{t-1} + \beta_2 \by_{t-1}
+ \bZ_{t-1}\bgamma + \bLambda \bff_{t-1} + \bepsilon_t.
\end{align}
The transformed network effects $\bB_1$ is equivalent to $\bB$ up to a rotation.
But the network autoregressive coefficient stays the same, that is $\bU \bB_1 \bU^\top = \bTheta \bB \bTheta^\top$.
Matrix $\bU$ is referred to as the spectral representation of $\bTheta$.
It can be estimated with higher accuracy than the membership matrix $\bTheta$, without resorting to the $k$-means algorithm to determine $\bTheta$.
Consequently, we estimate the pair $\paran{\bU, \bB_1}$ instead of $\paran{\bTheta, \bB}$.

\subsection{Model stationarity}

In this section, we investigate the stationarity of the time series $\{\by_t\}$.
Denote by $\bG = \bU\bB_1\bU^\top+\beta_2\bI_N$ and $\wt \bepsilon_t = \bZ_{t-1}\bgamma + \bepsilon_{t}$, we rewrite the CNAR model \eqref{eqn:cnar1} as
\beq
\by_t = \bG\by_{t-1} + \wt \bepsilon_t.\nonumber
\eeq
Thus, the CNAR model can be seen as a special case of the vector autoregression (VAR) model with autoregressive matrix given by $\bG$.
In our situation, the autoregressive matrix $\bG$ is structured by embedding
the community structure and low-rank common factor configuration.
We establish the stationarity condition for the CNAR model (\ref{eqn:cnar}) as follows.

\bet\label{stat1}
Assume that $E\|\bz_{t, i\cdot}\|_1<\infty$, $E\|\bff_t\|_1<\infty$, $\|\bLambda\|_\infty < \infty$, and  $\wt \bepsilon_t$ is strictly stationary.
If $\rho(\bB_1) + \abs{\beta_2} <1$,
then there exists a unique strictly stationary solution with finite first order moment to the CNAR model (\ref{eqn:cnar}). The solution takes the form,
$
\by_t =  \sum_{j = 0}^\infty \bG^j\wt\bepsilon_{t-j}.
$
\eet
Theorem \ref{stat1} establishes a sufficient condition (i.e., $\rho(\bB_1) + |\beta_2| <1$) for the strict stationarity of the model (\ref{eqn:cnar1}).
The proof is given in Appendix \ref{stat_proof}.

\section{Estimation} \label{sec:est}

According to our analysis in Section \ref{sec:reform}, Model~\eqref{eqn:cnar} is equivalent to Model~\eqref{eqn:cnar1}.
It suffices to estimate the spectral representation $\bU$ instead of the membership matrix $\bTheta$.
To perform the estimation, we firstly estimate $\hat\bU$ by the first $K$ eigenvectors of the adjacency matrix $\bA$ corresponding to the $K$ largest absolute eigenvalues.
The spectral representation is a major step in spectral clustering used in community detection for SBM \citep{lei2015consistency,zhao2012consistency,rohe2011spectral}.
As a result, a knowledge of $\hat\bU$ is sufficient for our purpose of estimation and prediction.
For completeness, we include a brief summary of the estimation properties of $\wh \bU$ in Lemma \ref{thm:SBM-eigvec-L2-bound} based on \cite{lei2015consistency}.
Next, we proceed to estimate the autoregression coefficients
by a two-step estimator.
In the first step, we plug-in the estimated $\hat\bU$ matrix, estimate the unknown autoregression coefficients and the regression residual $\hat\beps_t$ by the least squares method.
In the second step, we estimate the covariance of the noise term $\beps_t$ by using the POET method \citep{fan2013large} to the residual $\hat\beps_t$, and then use the estimated covariance matrix to further improve the estimation efficiency.
The details of the last two steps are presented in the next two sections.

\subsection{First step estimation of auto-regression coefficients}

First the CNAR Model~\eqref{eqn:cnar} could be re-written as
\beq
\by_t = \paran{(\by_{t-1}^\top \bU)\otimes \bU}\vec(\bB_1)
+ \beta_2 \by_{t-1}
+  \bZ_{t-1}\bgamma
+ \beps_t.  \nonumber
\eeq
Write $\bX_{t-1} = ((\by_{t-1}^\top \bU)\otimes \bU, \by_{t-1}, \bZ_{t-1})\in\RR^{N\times (K^2+p+1)}$ and
let $\btheta = (\vec(\bB_1)^\top, \beta_2, \bgamma^\top)^\top \in \RR^{K^2+p+1}$.
Then we have the following relationship,
\beq \label{eqn:ls-form}
\by_t = \bX_{t-1}\btheta + \beps_t.
\eeq
By using a standard least squares estimation,
$\btheta$ can be estimated by
\beq
\wt\btheta = \big(\sum_{t}\bX_{t-1}^\top\bX_{t-1}\big)^{-1}
\big(\sum_{t}\bX_{t-1}^\top\by_{t}\big).\nonumber
\eeq
Note that we assume $\beps_t$ is uncorrelated over $1\le t\le T$ but with cross-sectional dependence.
That implies that the above least squares estimation will yield a consistent estimation of $\btheta$.
In practice, the true value of $\bU$ is unknown.
Instead, we plug-in corresponding estimators into $\bX_{t-1}$ and obtain
$\wh \bX_{t-1} = ((\by_{t-1}^\top \wh\bU)\otimes \wh\bU, \by_{t-1}, \bZ_{t-1})\in\RR^{N\times (K^2+p+1)}$.
Accordingly, we obtain the first step estimator as
\beq \label{lse}
\wh\btheta^{(1)} = \big(\sum_{t}\wh\bX_{t-1}^\top\wh\bX_{t-1}\big)^{-1}
\big(\sum_{t}\wh \bX_{t-1}^\top\by_{t}\big).
\eeq

\subsection{Second-step weighted least squares}
\label{sec:2nd}

Once the first step estimator $\wh\btheta^{(1)}$ is given, we are then able to plug-in the estimator
to obtain the residuals as follows:
\begin{equation}\label{eqn:resid1}
\wh\beps_t = \by_t - \wh\bX_{t-1}\wh\btheta^{(1)}.
\end{equation}
Note that the true noise $\beps_t$ can be decomposed into a low dimensional factor structure $\bLambda\bff_t$ and an independent
error $\be_t$ in \eqref{eqn:fac_model}.
Therefore the covariance matrix $\bSigma_\eps$  of $\beps_t$ can be expressed as
\begin{align*}
\bSigma_\eps = \bLambda\bSigma_f\bLambda^\top + \bSigma_e,
\end{align*}
where $\Sigma_f\in\RR^{M\times M}$ and $\bSigma_e\in\RR^{N\times N}$ are the covariance matrices of the latent common factor and the idiosyncratic component, respectively.

Estimating a low-rank plus sparse covariance matrix $\bSigma_\eps$ with either observed or unobserved factors has been considered in \cite{fan2008high,fan2011high,fan2013large} and \cite{fan2019learning}.
However, the task here is slightly different from the standard problem:
while the $\braces{\beps_t}_{t=1}^T$ is known in previous studies, direct observation of $\braces{\beps_t}_{t=1}^T$ is absent in our setting.
We need to use the estimated $\braces{\hat\beps_t}_{t=1}^T$ which are the residuals from the first step regression of the CNAR model.
Moreover, we consider a time series setting where the number of communities $K$, and, as a result,  the number of coefficients $K^2+p+1$ are allowed to grow with $N$.

Specifically, we apply a simplified version of the POET method to the first-step CNAR residual $\hat\beps_t$.
$\bff_t$ and $\bLambda$ are both unknown and are estimated by PCA as in \cite{fan2013large}.
Let $\hat\bLambda$ and $\hat\bff_t$ be the estimated loading matrix and factor,
$\hat\be_t = \hat\beps_t - \hat\bLambda\hat\bff_t$ be the estimated idiosyncratic component.
The covariance of $\braces{\hat\beps_t}_{t=1}^T$ is given by
\begin{equation} \label{eqn:cov_eps}
\hat\bSigma_\ve = \hat\bLam\hat\bSigma_f\hat\bLam^\top + \hat\bSigma_e,
\end{equation}
where $\hat\bSigma_f = (1/T) \sum_{t=1}^T \hat \bff_t \hat \bff_t^\top$ and $\hat\bSigma_e = \diag\paran{(1/T) \sum_{t=1}^T\hat\be_t\hat\be_t^\top}$.
Using $\hat\bSigma_\eps^{-1}$, we construct a weighted least squares type objective function as follows,
\begin{align}
Q(\btheta) = \sum_t(\by_t - \hat\bX_{t-1}\btheta)^\top\wh\bSigma_\ve^{-1}(\by_t - \hat\bX_{t-1}\btheta), \nonumber
\end{align}
which yields our second step estimator $\wh\btheta^{(2)}$ as
\beq
\wh\btheta^{(2)} = \Big(\sum_t \hat\bX_{t-1}^\top\wh\bSigma_\ve^{-1}\hat\bX_{t-1}\Big)^{-1}\Big(\sum_t \hat\bX_{t-1}^\top\wh\bSigma_\ve^{-1}\by_{t}\Big).\label{theta_est2}
\eeq
One may wonder if \eqref{theta_est2} is computationally expensive since it requires to compute the inverse of
an $N\times N$-dimensional matrix $\wh\bSigma_\ve$.
This issue can be easily solved since $\bSigma_\ve$ takes a special decomposition as in \eqref{eqn:cov_eps}.
Applying the Sherman-Morrison-Woodbury formula and using the identification condition that $\hat\Sigma_f^{-1}=\bI_K$ we could obtain
\begin{align}
\hat\bSigma_\eps^{-1} = \hat\bSigma_e^{-1}
- \hat\bSigma_e^{-1}\hat\bLam
\paran{\bI_K+\hat\bLam^\top\hat\bSigma_e^{-1}\hat\bLam}^{-1}
\hat\bLam^\top\hat\bSigma_e^{-1}.
\end{align}
The inverse of a sparse matrix $\hat\bSigma_e$ can be computed efficiently since it is diagonal.

\section{Theoretical Properties} \label{sec:theory}

We present the theoretical properties of the CNAR model estimation in this section.
First, the asymptotic properties of the first step estimator are presented in Section
\ref{sec:first}.
Next, the properties of the covariance estimation as well as the second step estimation are stated in Section \ref{sec:second}.

\subsection{Asymptotic properties of the first-step estimation}\label{sec:first}

To investigate the estimation properties of the first step estimator, we first focus on the case that the number of communities $K$ is fixed.
As a result, the number of coefficients to be estimated is fixed.
In the following theorem we establish the asymptotic result of the first step estimator.

\newpage

\bet\label{thm:first_step_estimator}
{\sc (First Step Estimator for Fixed $K$)}
Assume $\sigma_1(\bB_1)+|\beta_2|<1$,
$\sigma_1(\bSigma_e)\le c_e$, and $K$ is fixed,
where $c_e$ is a positive constant.
In addition, assume
    Assumption B.1
    and $\bU^\top \bLambda \ne \zero$.
Define $\balpha_{NT} = \diag\{\sqrt T\bI_{K^2}, \sqrt T, \sqrt{NT}\bI_{p}\}\in \RR^{(K^2+p+1)\times (K^2+p+1)}$.
Then the following conclusion holds.\\
(a) {\sc (Oracle Estimator)}.
Let $\bGamma_y(0) = \cov(\by_t)$
and $\bSigma_\eps = \cov(\beps_t)$. 
We have
$\balpha_{NT}\bSigma_x(\wt\btheta - \btheta)\rightarrow_d N(\zero, \bSigma_{2\eps})$,
where
\begin{align}
\bSigma_{x} =
\left(
\begin{array}{ccc}
\bSigma_u
& \bSigma_{uy}
& \zero \\
& \kappa_{y}
& \zero \\
&
& \bSigma_Z
\end{array}
\right),
~~~
\bSigma_{2\eps} =
\left(
\begin{array}{ccc}
\bSigma_{ue}
& \bSigma_{uye}
& \zero \\
& \kappa_{ye}
&\zero \\
&
& \kappa_e\bSigma_Z
\end{array}
\right)\label{Sig_x_2e}
\end{align}
with
$\bSigma_u = \lim_{N\rightarrow\infty}N^{-1}\bU^\top\bGamma_y(0)\bU$,
$\kappa_y = \lim_{N\rightarrow\infty}N^{-1}\tr\{\bGamma_y(0)\}$,
$\bSigma_{uy} = \lim_{N\rightarrow\infty}N^{-1}\vec\{\bU^\top\bGamma_y(0)\bU\}$,
$\bSigma_{ue} = \lim_{N\rightarrow\infty}N^{-2}\{\bU^\top\bGamma_y(0)\bU\}
\otimes (\bU^\top\bSigma_\ve \bU)$,
$\kappa_{ye} = \lim_{N\rightarrow\infty}N^{-2}\tr\{\bGamma_y(0)\bSigma_\eps\}$,
$\bSigma_{uye} = \lim_{N\rightarrow\infty}N^{-2}$ $\vec (\bU^\top \bSigma_{\ve}\bGamma_y(0)\bU)$,
$\kappa_{e} = \lim_{N\rightarrow\infty}N^{-1}\tr\{\bSigma_\eps\}$.\\
(b) {\sc (First Step Estimator).}
Define $\gamma_N$ as smallest absolute nonzero eigenvalue of $\bTheta \bB\bTheta^\top$.
Assume $\gamma_N\gg \sqrt{\log N}$ and $T\gg (\log N/\gamma_N)^2\bigvee 1$ .
We have $\balpha_{NT}\bSigma_x(\wh \btheta^{(1)} - \btheta)\rightarrow_d N(\zero, \bSigma_{2\eps})$.
\eet
The proof of Theorem \ref{thm:first_step_estimator} is given in Appendix \ref{append:first_step_estimator}.
Related to the results, we have two comments.
First, it implies that,
both $\wh\bB_1$ and $\wh \beta_2$ are $\sqrt T$-consistent while
$\wh\bgamma$ is $\sqrt{NT}$-consistent.
The slower convergence rates of $\wh \bB_1$ and $\wh \beta_2$ are mainly due to the pervasiveness
condition (i.e., Assumption B.1) of the factor structure.
Under this assumption, the largest eigenvalue of $\bSigma_\ve$ will diverge in the rate of $O_p(N)$,
which reduces the convergence rates of the autoregression related parameters.  When there is no pervasive factor, the rate of convergence is $\sqrt{NT}$.

{\bf Remark 4.}
Note that we require that $\bU^\top\bLambda \ne 0$
as a condition to establish the asymptotic results.
If it holds $\bU^\top\bLambda = \zero$,
then we could obtain that the $\bSigma_u\rightarrow \zero$
and $\bSigma_{ue}\rightarrow \zero$.
Under this case, the community structure diminishes the factor effect,
    which leads to different forms of the asymptotic result.  We state the result in the Appendix for completeness.


The above theoretical results are established under the condition that $K$ is fixed.
However, the fixed $K$ assumption can be restrictive in practice.
Typically, more groups will emerge as $N\rightarrow\infty$.
As a result, the number of parameters to be estimated is also diverging as $K\rightarrow\infty$.
The following theorem establishes the theoretical properties of $\wt\btheta$ and
$\wh\btheta^{(1)}$ under this situation.  It shows that the properties of Theorem~\ref{thm:first_step_estimator} continue to hold.

\bet\label{thm_first_step_K_diverge}
{\sc (First Step Estimator for $K\rightarrow\infty$)}
Assume $\bB_1$ is diagonalizable which holds automatically for symmetric matrices, i.e.,
$\bB_1 = \bP_B\bD_B \bP_B^{-1}$ with $\bD_B$ being a diagonal matrix.
Assume the same conditions as in Theorem \ref{thm:first_step_estimator}
and $K = o(T^{1/2})$.
Let $\bA_K\in \RR^{m\times (K^2+p+1)}$ satisfy $\bA_K\bA_K^\top\rightarrow \bH$
as $K\rightarrow\infty$, where $m$ is a fixed integer and $\bH\in\RR^{m\times m}$ is a non-negative symmetric matrix.
It then holds
\beq
\bA_K\bSigma_{2\ve}^{-1/2}\balpha_{NT}\bSigma_x(\wt\btheta - \btheta)\rightarrow_d
N(\zero, \bH). \label{Koracle_normal}
\eeq
Further assume $\gamma_N\gg \sqrt{K\log N}$
and $T\gg (\sqrt K\exp(K\log 21)\log N/\gamma_N)^2$. Then it holds,
\beq
\bA_K\bSigma_{2\ve}^{-1/2}\balpha_{NT}\bSigma_x(\wh\btheta^{(1)} - \btheta)\rightarrow_d
N(\zero, \bH).\label{Kfirst_normal}
\eeq
\eet
The proof of the theorem is given in Appendix \ref{thm_first_step_K_diverge_proof}.
Implied by the above theorem, any finite sub-vector of $\vec(\wh\bB_1^{(1)})$ is $\sqrt{T}$-consistent.
However, since it ignores the potential cross-sectional dependence structure in $\beps_t$, it is sub-optimal in terms of estimation efficiency. We further improve it by a second step estimator and
the corresponding theoretical properties are discussed in the next section.

\subsection{Asymptotic properties of the second-step estimation} \label{sec:second}

We apply the principal orthogonal complement thresholding method \cite[POET]{fan2013large} to estimate the covariance $\bSigma_\eps$.
However, the task here is slightly different from that in \cite{fan2013large} since direct observations for $\braces{\beps_t}_{t=1}^T$ are not available.
Instead, we are able to calculate $\braces{\hat\beps_t}_{t=1}^T$ -- the residual from the first step regression of the CNAR model.
That is, $\hat\beps_t=\by_t - \bX_{t-1}\hat\btheta^{(1)}$.
The estiamtion properties of $\wh \bSigma_\ve$ and $\wh \bSigma_\ve^{-1}$ are given in the following theorem.


\bet \label{thm:resid-cov-bound}
{\sc (Covariance Estimation)}
Suppose that $\kappa_{NT} = K/\sqrt T$
and assume the same conditions with Theorem \ref{thm_first_step_K_diverge}.
In addition, assume $\log(N) = \smlo{T^{\gamma/6}}$ where $\gamma$ is defined in Assumption \ref{assum:mixing}, and Assumption \ref{assum:pervasive} -- \ref{assum:regul-covari} hold.
Then, the proposed covariance-estimator based on the first-step CNAR residual satisfies
\begin{enumerate}[label=(\roman*)]
    \item $\norm{{\wh{\bSigma}_{\eps} }- \bSigma_{\eps}}_{\bSigma_{\eps}} = \Op{\frac{\sqrt{N} \log(N)}{T} + \frac{1}{\sqrt{N}} + \sqrt N\kappa_{NT}^2K^2}$.
    \item $\norm{{\wh{\bSigma}_{\eps}^{-1}} - \bSigma_{\eps}^{-1}} = \Op{\frac{1}{\sqrt{N}} + \sqrt{\frac{\log(N)}{T}} + \kappa_{NT} K }$.
\end{enumerate}

\eet

The proof of Theorem \ref{thm:resid-cov-bound} is given in Appendix \ref{appen_cov}.
{The result extends the theoretical results of the covariance estimation in \cite{fan2013large}
    under our CNAR modeling framework.
    The estimation error of the first step estimator is involved through the extra term $\kappa_{NT}$.
    As we could observe, estimating the precision matrix $\Omega_\ve = \Sigma_\ve^{-1}$ is much easier than $\bSigma_\ve$
    under this context.
    This will lead to the $\sqrt{NT}$-consistency result of the second step estimator.
    The details are stated in the following theorem.
}

{

    \bet\label{thm_sec_step_K_diverge}
    {\sc (Second Step Estimator)}
    Assume the same conditions with Theorem \ref{thm_first_step_K_diverge}
    and \ref{thm:resid-cov-bound}.
    Let $\omega_{NTK} =N^{-1/2} + (\log N/T)^{1/2} +  K^2/\sqrt T  $
    and assume that $\omega_{NTK} = o(1)$.
    Let $\bA_K\in \RR^{m\times (K^2+p+1)}$ satisfy $\bA_K\bA_K^\top\rightarrow \bH$
    as $K\rightarrow\infty$, where $m$ is a fixed integer and $\bH\in\RR^{m\times m}$ is a non-negative symmetric matrix.
    Then it holds,
    \beq
    \sqrt{NT}\bA_K\bSigma_{2\theta}^{1/2}(\wh\btheta^{(2)} - \btheta)\rightarrow_d
    N(\zero, \bH),\label{Ksec_normal}
    \eeq
    where
    \begin{align}
    \bSigma_{2\theta} =
    \left(
    \begin{array}{ccc}
    \bSigma_{ue2}
    &\bSigma_{uye2}
    & \zero \\
    & \kappa_{ye2}
    & \zero\\
    &
    & \kappa_{e2}\bSigma_Z
    \end{array}
    \right),\label{Sig_2theta}
    \end{align}
    where $\bSigma_{ue2} = \lim_{N\rightarrow\infty}\frac{1}{N}(\bU^\top\bGamma_y(0)\bU) \otimes (\bU^\top \Omega_\ve \bU)$,
    $\bSigma_{uye2} =  \lim_{N\rightarrow \infty}\frac{1}{N}\vec(\bU^\top\Omega_\ve\bGamma_y(0)\bU)$,
    $\kappa_{ye2} = \lim_{N\rightarrow\infty}\frac{1}{N} \tr\{\Omega_\ve \bGamma_y(0)\}$,
    and $\kappa_{e2} = \lim_{N\rightarrow\infty} \frac{1}{N}\tr(\Omega_\ve)$
    with $\Omega_\ve = \bSigma_\ve^{-1}$.
    \eet

    The proof of Theorem \ref{thm_sec_step_K_diverge} is given in Appendix \ref{proof_sec_est}.
    With respect to the convergence result, we have the following comments.
    First, the assumption involved here requires $T\gg \log N$ and $K\ll T^{1/6}$,
    which can be easily satisfied for large scale networks.
    Second, it should be noted although the largest eigenvalue of $\bSigma_\ve$ diverges in the speed of $O(N)$,
    the largest eigenvalue of the precision matrix (i.e., $\Omega_\ve$) is bounded by a finite constant \citep{fan2013large}.
    Hence the noise level is well controlled.
    As a result, this further enables us to achieve a better convergence rate (i.e., $\sqrt{NT}$-consistency) of the second step estimator.

}


\section{Numerical Studies} \label{sec:simu}

In this section, we use Monte Carlo simulations to assess the adequacy of the asymptotic convergence rates  by evaluating the finite sample performance of estimators.
For different simulation models and combinations of $T$, $N$, and $K$, we report the relative mean squared errors (ReMSE) for model parameter estimators and the time series predictors.
Specifically, for any estimator $\hat\bM$ of the true value $\bM$, the ReMSE is defined as ReMSE$={\norm{\hat\bM - \bM}_F}/{\norm{\bM}_F}$.
Since the systematic noise is difficult to predict, we report the prediction error with respect to the signal $\bs_{t}$ defined in \eqref{eqn:simul-model} below.
We compare our results with those estimated using NAR \citep{zhu2017network}.
All results are based on $200$ replications.

The synthetic data are generated according to the following model:
\begin{equation}  \label{eqn:simul-model}
\by_t = \underbrace{\bPhi \by_{t-1} + \beta_2 \by_{t-1} + \bZ_{t-1}\bgamma}_\text{Signal $\bs_t$} + \underbrace{\bLam\bff_t + \be_t}_\text{Noise $\beps_t$},
\end{equation}
where $\bPhi=\bTheta\bB\bTheta^\top$ for CNAR model and $= \beta_1 \tilde\bA$
with $\tilde\bA$ being the row-normalized adjacency matrix.
It has different structures in the following three examples, which are specified later.
We generate the covariates from multivariate normal distribution $\calN\paran{\bzero,\bI_p}$
with the dimension fixed at $p = 5$.
We fix the dimension of latent factor $\bff_t$ at $M=3$ and generate $\bff_t$ and $\be_t$ from $\calN\paran{\bzero,\bI}$ of appropriate dimensions. 
For true values of the coefficients, we fix $\beta_2$ and $\bgamma$ at $\beta_2=0.3$ and $\bgamma=\paran{-0.1, 0.2, -0.3, 0, 0}^\top$.
The entries of the loading matrix $\bLam \in \RR^{N\times M}$ are also fixed for each pair of $(N, M)$ but are generated randomly once from $\calN\paran{1,1}$.

\subsection{Comparison of CNAR and NAR under different network models}

In the following three examples, we generate the data by using different mechanisms.
The finite sample performances of CNAR and NAR model \citep{zhu2017network} are compared in terms of estimation and prediction accuracy.
Because the covariance of the noise $\beps_t$ in \eqref{eqn:simul-model} is not diagonal,
we use the 2nd-step weighted least square estimator for both NAR and CNAR estimation for a fair comparison.

\begin{enumerate}[label={\bf \sc Example \arabic*}, align=left, wide, labelwidth=!, labelindent=0pt]
    \item ({\bf Stochastic Block Model and CNAR})
    We construct our network using the stochastic block model with $K$ communities
    \[
    \bQ = \alpha_N\bQ_0; \qquad \bQ_0=\rho\bI_K + (1-\rho)\bbone_K\bbone_K^\top, \quad 0<\lambda<1.
    \]
    This setting assumes that the edge probability between any pair of nodes depends only on whether they belong to the same community: the edge probability is $\alpha_N$ within community and $\alpha_N(1-\rho)$ between community.
    The quantity $\rho$ reflects the relative difference in connectivity between and within communities.
    We set $\alpha_N = 0.9$, $\rho=8/9$, and $\bPhi=\bTheta\bB\bTheta^\top$. Here
    $\bB$ is a $K\times K$ diagonal matrix, whose diagonal elements are the first $K$ numbers of the sequence $\{0.1, -0.1, 0.2, -0.2, 0.3, -0.3, \cdots \}$.
    We vary $T\in\brackets{50,100,\cdots,450}$, $K \in\brackets{2,4,8}$ and $N\in\brackets{200,400,800}$.

    Figure \ref{fig:example-1-boxplot} shows, for Example 1, the ReMSE of estimated network AR coefficient $\bU\bB_1\bU^\top$ and 1-step prediction $y_{T+1}$ by NAR (first row) and CNAR (second row), respectively, with $K=2$ and $N=400$.
    Note the scales of the y-axes are very different and all $T \le N=400$.
    Under all settings, CNAR outperforms NAR with much lower ReMSE, as NAR is not flexible to fully capture the complex community effects.
    In addition, as $T$ increases, the ReMSE of CNAR model decreases, as expected.  The CNAR with estimated $\hat\bU$ and two-step weighted LS using covariance estimator \eqref{eqn:cov_eps} becomes closer and closer to the oracle estimator as $T$ increases.

    \begin{figure}[htpb!]
        \centering
        \begin{subfigure}[b]{0.35\textwidth}
            \centering
            \includegraphics[width=\linewidth]{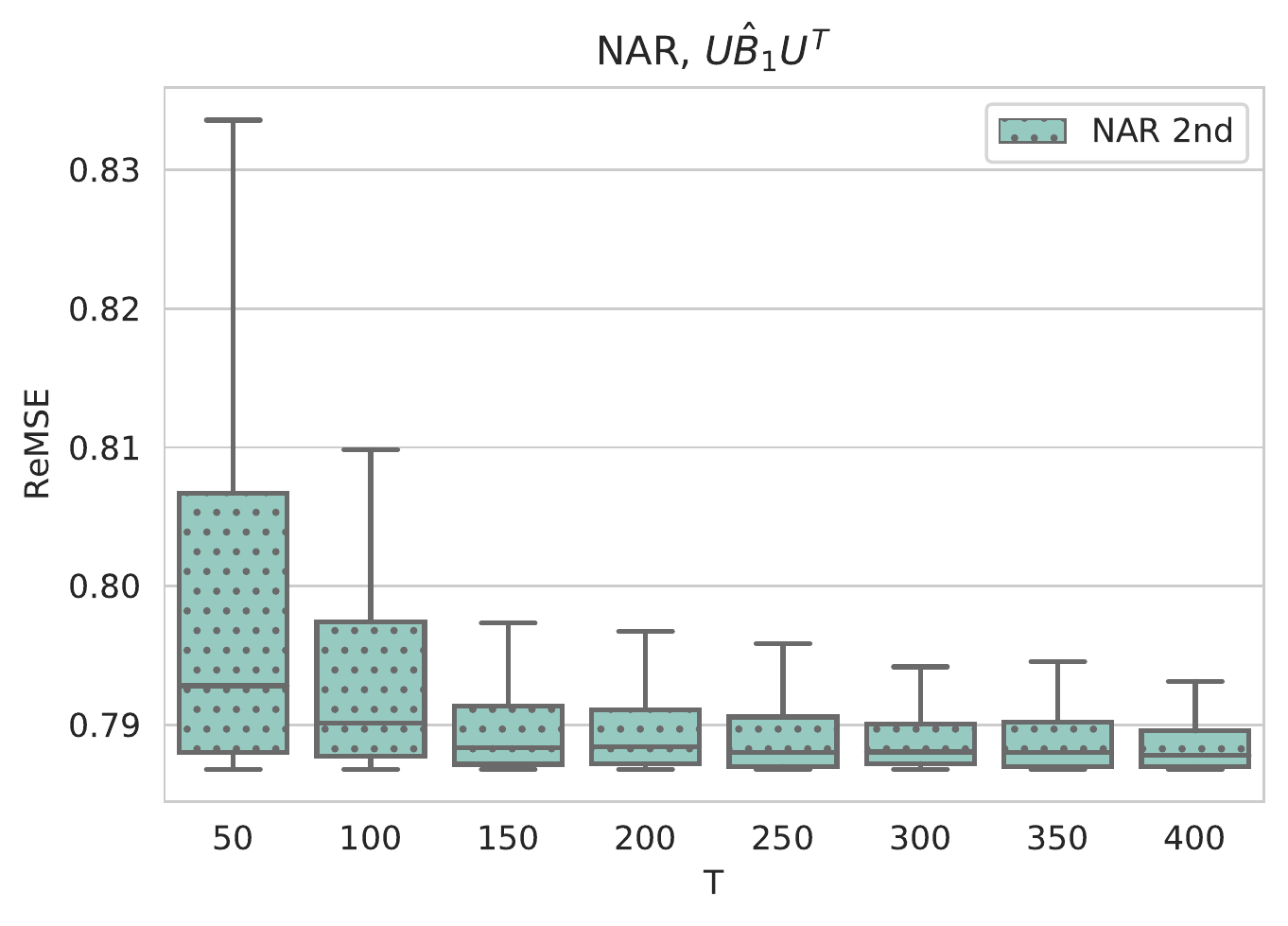}
        \end{subfigure}
        \hspace{3ex}
        \begin{subfigure}[b]{0.35\textwidth}
            \centering
            \includegraphics[width=\linewidth]{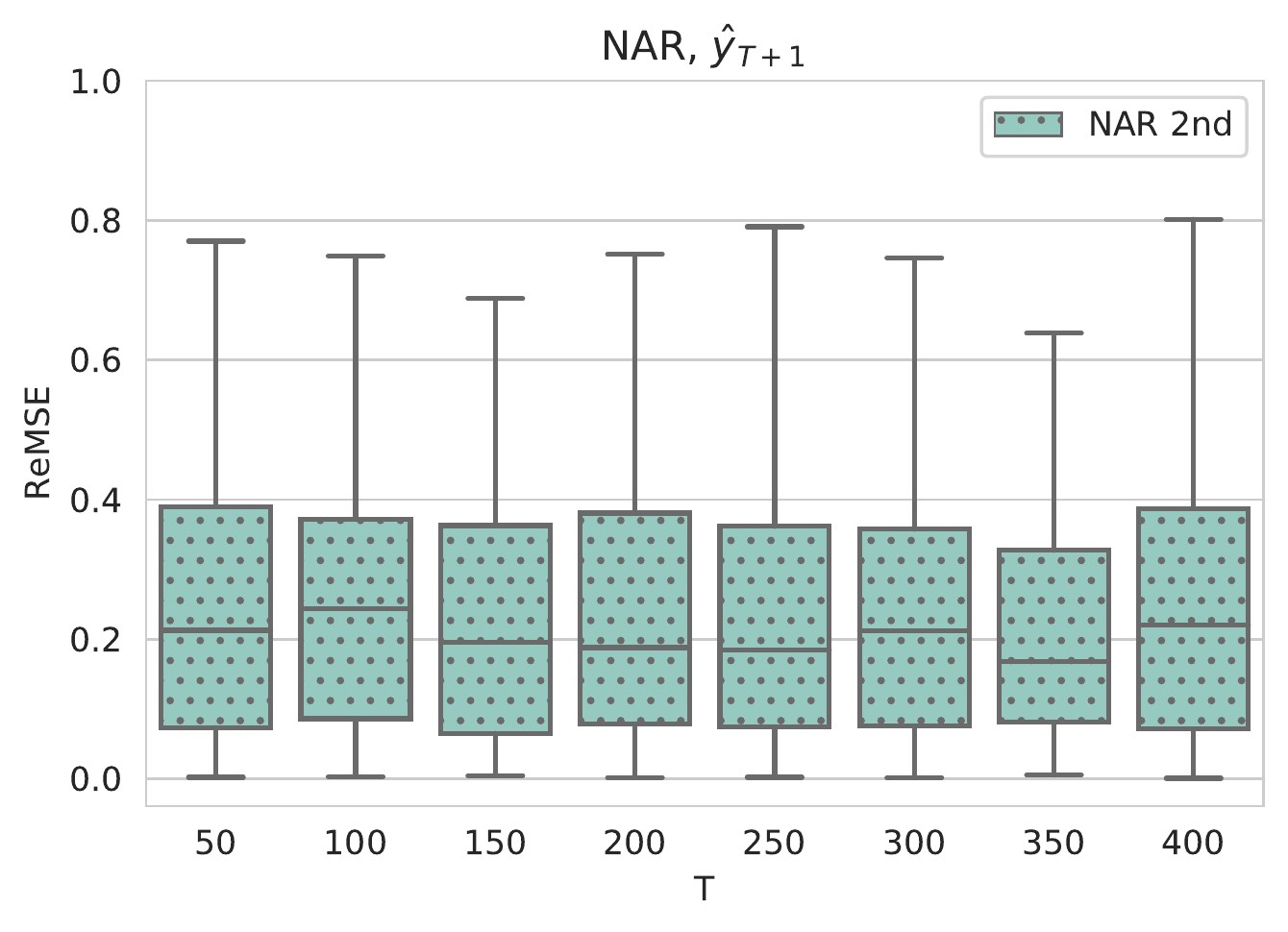}
        \end{subfigure}

        \begin{subfigure}[b]{0.35\textwidth}
            \centering
            \includegraphics[width=\linewidth]{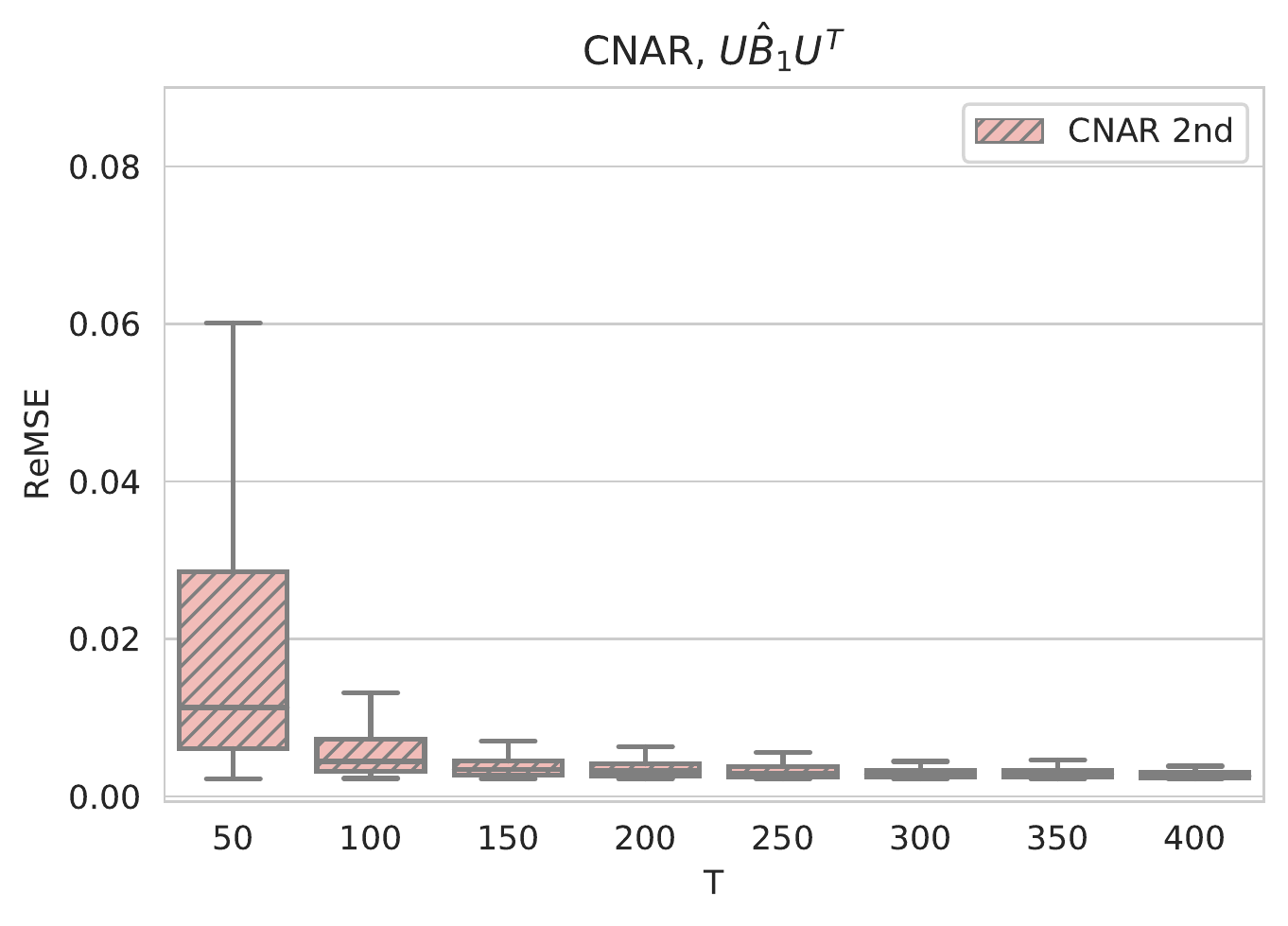}
        \end{subfigure}
        \hspace{3ex}
        \begin{subfigure}[b]{0.35\textwidth}
            \centering
            \includegraphics[width=\linewidth]{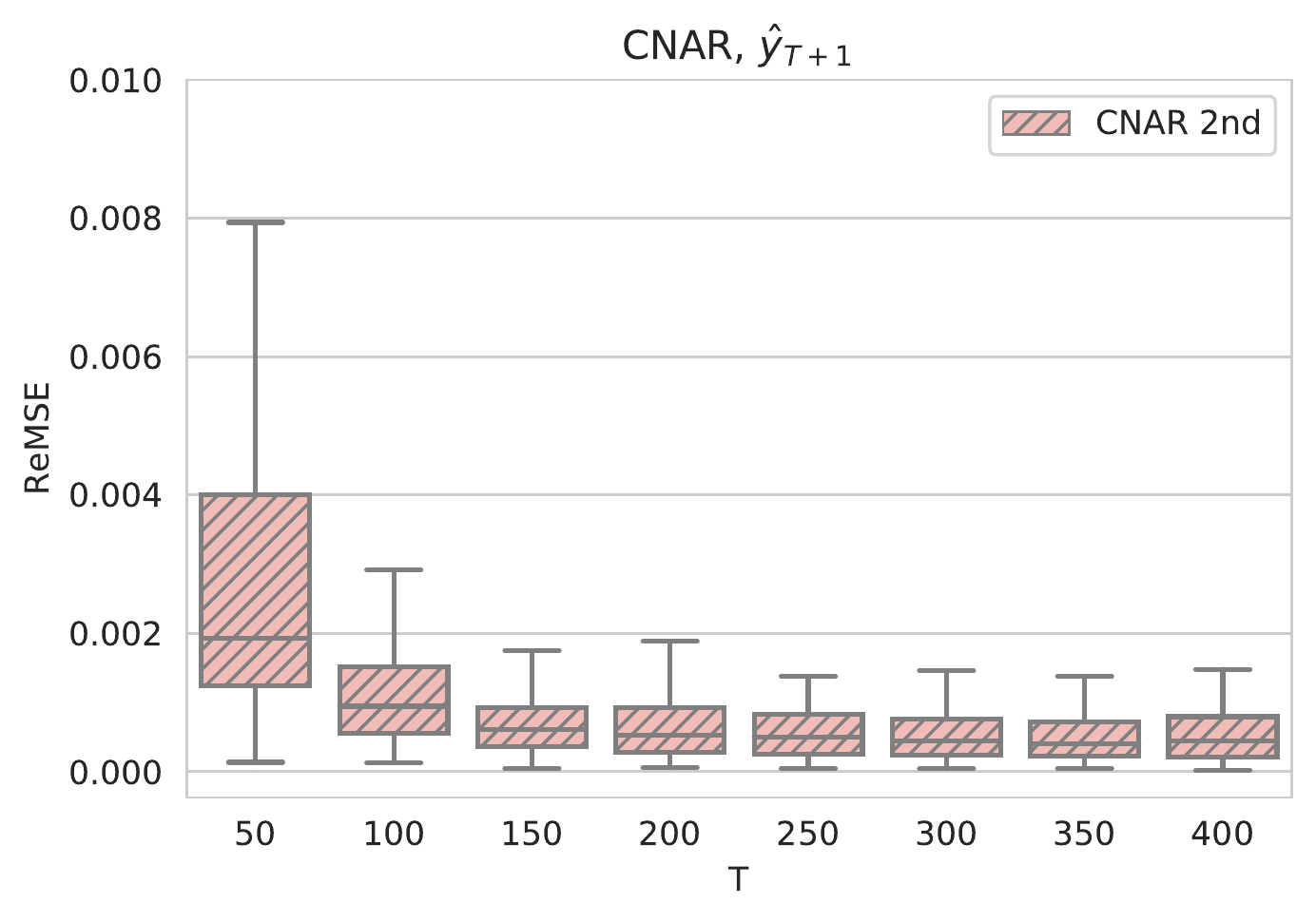}
        \end{subfigure}
        \caption{Example 1 (SBM and CNAR model with $K=2$ and $N=400$). Box plots for the ReMSE for estimated network AR coefficient $\bU\bB_1\bU^\top$ and 1-step prediction $y_{T+1}$ by NAR (first row) and CNAR (second row), respectively.
        }
        \label{fig:example-1-boxplot}
    \end{figure}

    In the supplementary material, we present a table comparing the performance of CNAR and NAR in Example 1 with different numbers of network communities $K\in\braces{2,4,8}$, network sizes $N\in\braces{200,400,800}$ and time series length $T=100,200,400$.
    The performance of CNAR is better than that of NAR under all settings.
    For the CNAR model, it can be observed that the errors reduce when $T$ or $N$ with given number of
    communities $K$, which corroborates with the theoretical findings.

    \item ({\bf Low-rank spectral network and CNAR})
    In this example, we focus on the sensitivity analysis when the network is not generated from the SBM model.
    We consider a general network with spectrum resembling that of an SBM.
    The network is generated by the Python graph function \textit{spectral\_graph\_forge} in the \textit{NetworkX} package \citep{hagberg2008exploring}, which computes the eigenvectors of a given SBM's adjacency matrix, filters them and builds a random graph with a similar eigenstructure ($\alpha=0.95$ for the second function argument).
    The membership matrix $\bTheta$ is generated using community assignments from a spectral clustering.
    Other model parameters, such as $\bB$, $\bbeta_2$ and $\gamma$, are set in the same way as in Example 1.

    Figure \ref{fig:sbm-spectral-CNAR-boxplot} shows the ReMSE of estimated network AR coefficient $\bU\bB_1\bU^\top$ and 1-step prediction $y_{T+1}$ by NAR (first row) and CNAR (second row), respectively, with $K=2$ and $N=400$.
    Under this setting with a different network generative model, CNAR works similarly to it does in the SBM case in Example 1 and outperforms NAR.
    In the supplementary material, we compares CNAR and NAR under two additional settings, namely the clusters of power law network and the random partition network.
    Under all setting,  the CNAR model still has much smaller  estimation and prediction error compared to the NAR model.
    Although our theories are based on the SBM generative model, all examples demonstrate empirically that the CNAR applies to a wider group of network models.
    Therefore, the CNAR model is shown to be more robust than the NAR model.

    \begin{figure}[htpb!]
        \centering
        \begin{subfigure}[b]{0.35\textwidth}
            \centering
            \includegraphics[width=\linewidth]{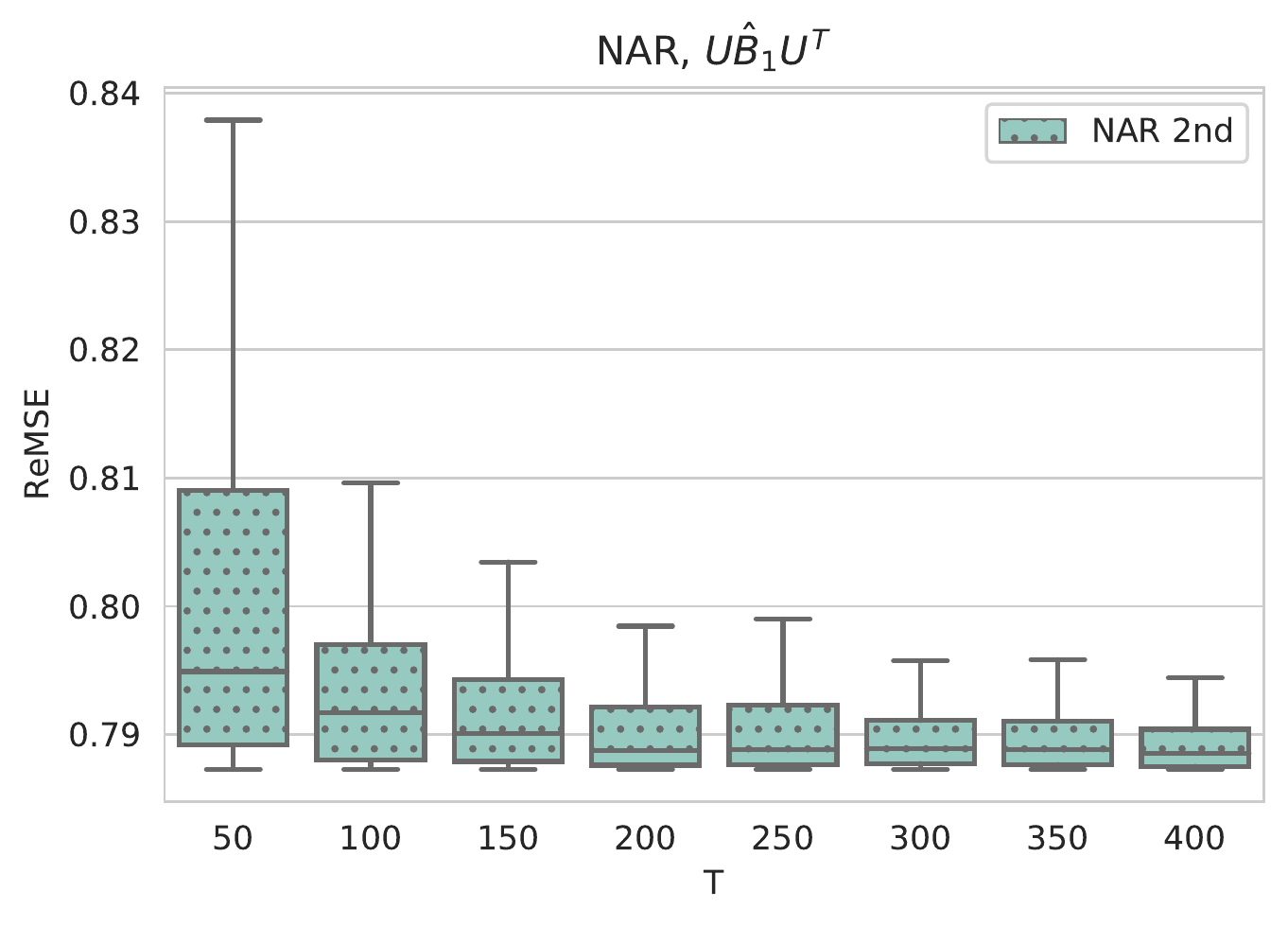}
        \end{subfigure}
        \hspace{3ex}
        \begin{subfigure}[b]{0.35\textwidth}
            \centering
            \includegraphics[width=\linewidth]{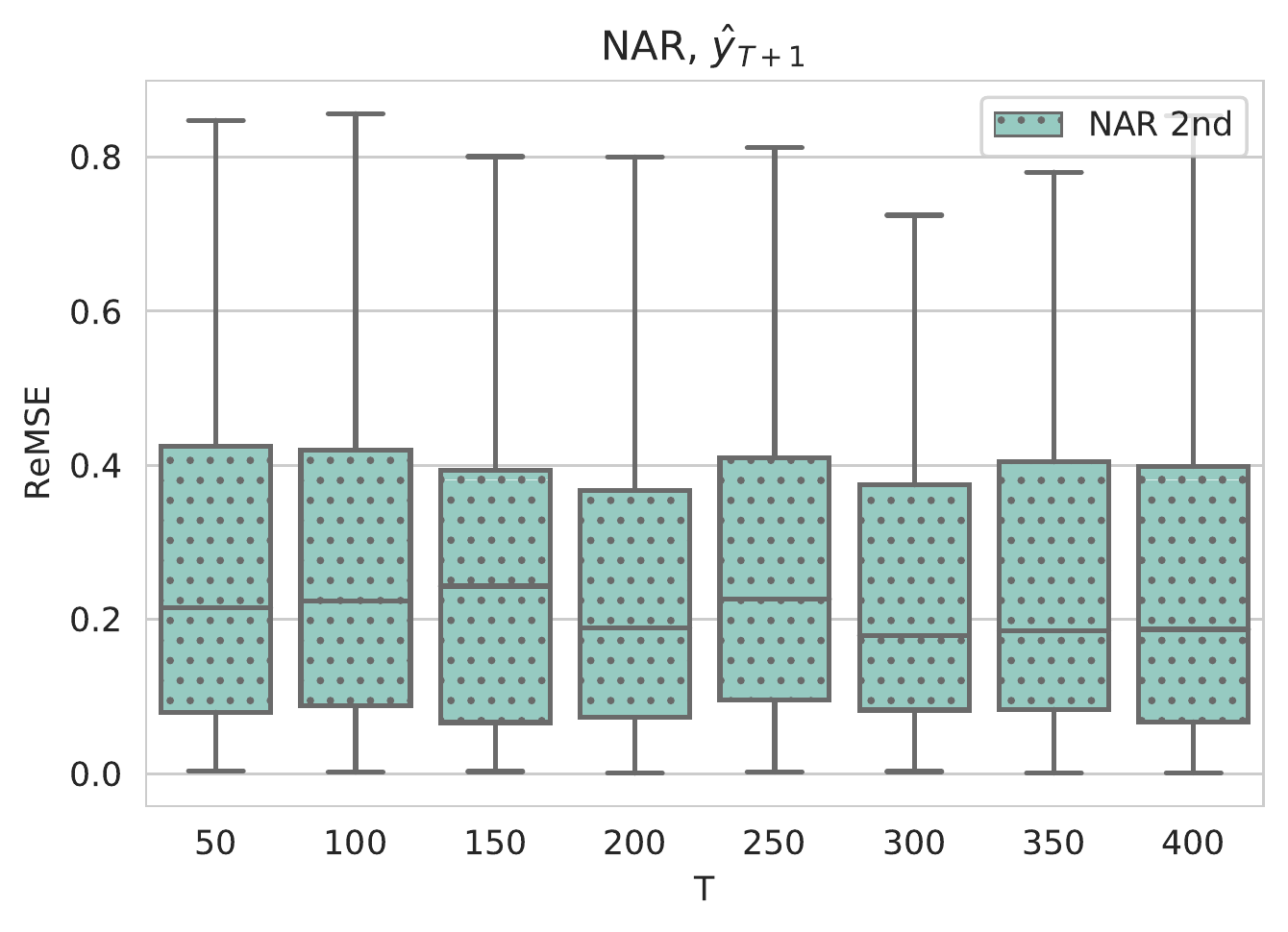}
        \end{subfigure}

        \begin{subfigure}[b]{0.35\textwidth}
            \centering
            \includegraphics[width=\linewidth]{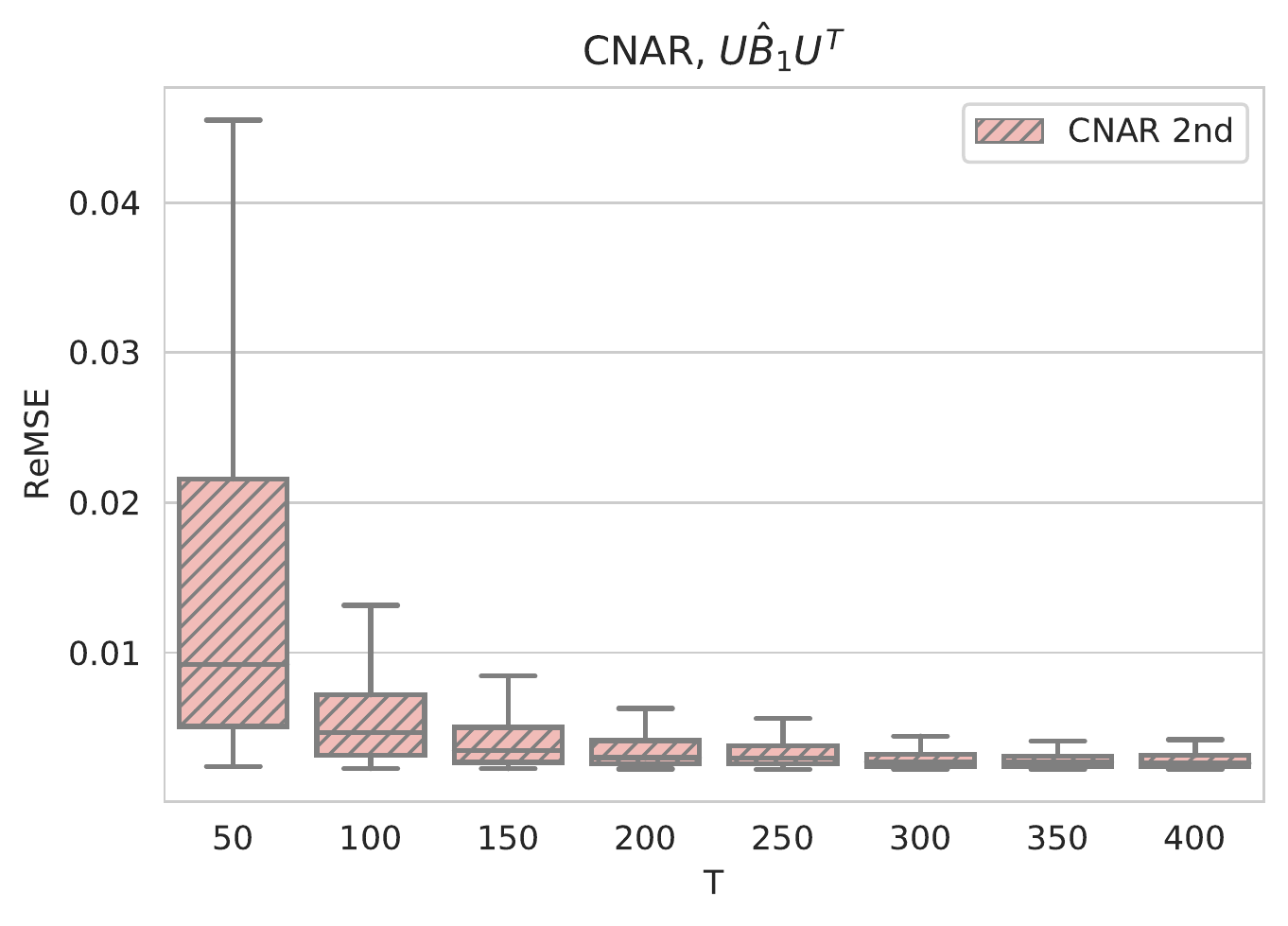}
        \end{subfigure}
        \hspace{3ex}
        \begin{subfigure}[b]{0.35\textwidth}
            \centering
            \includegraphics[width=\linewidth]{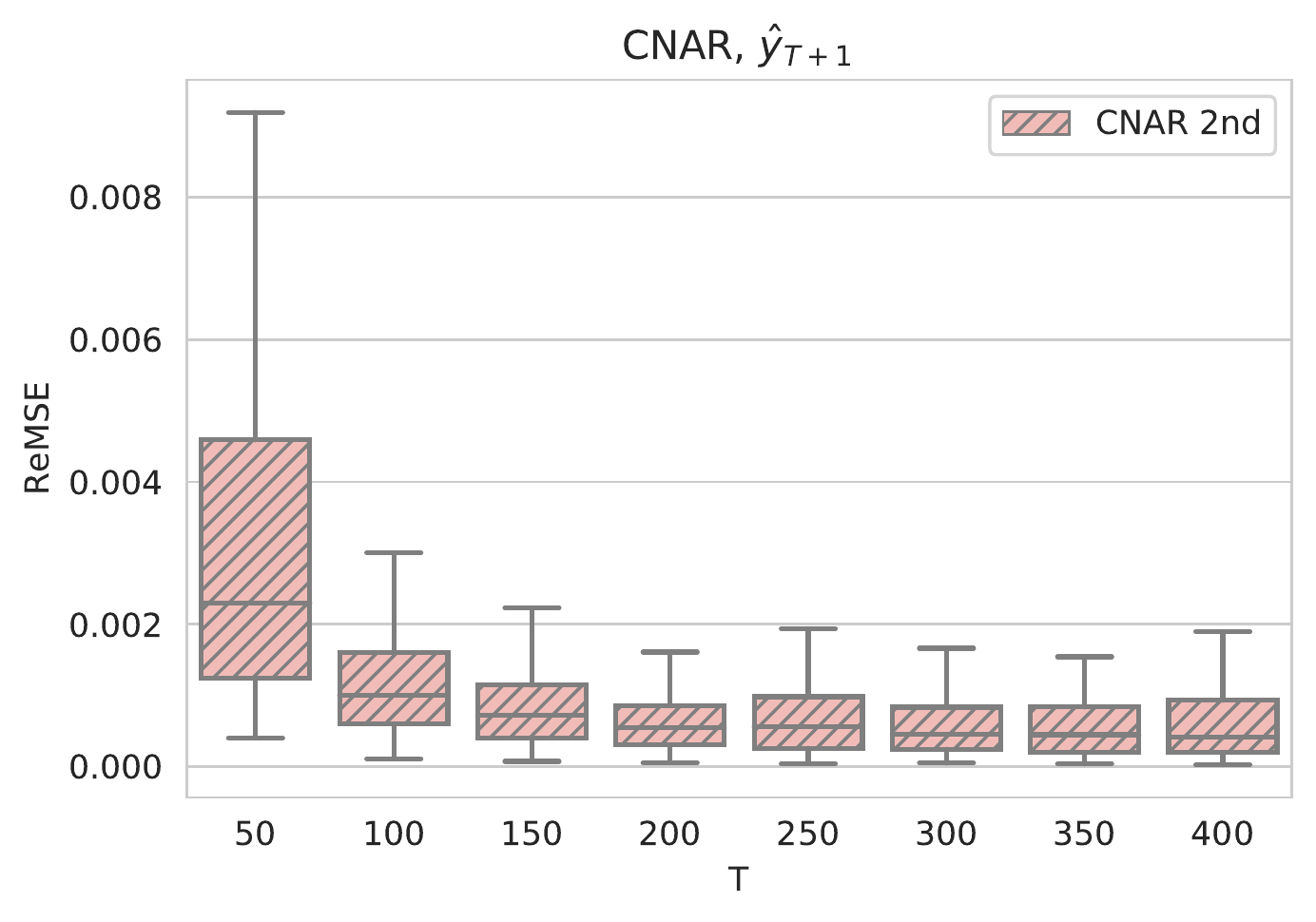}
        \end{subfigure}
        \caption{Example 2 (General low-rank network and CNAR model with $K=2$ and $N=400$). Box plots for the ReMSE for estimated network AR coefficient $\bU\bB_1\bU^\top$ and 1-step prediction $y_{T+1}$ by NAR (first row) and CNAR (second row), respectively.
        }
        \label{fig:sbm-spectral-CNAR-boxplot}
    \end{figure}

    \item ({\bf Stochastic Block Model and NAR})
    In this setting we consider a situation where the data is generated from {the NAR model. Thus the situation is more favorable to NAR in this case.}
    Specifically, the adjacency matrix $\bA$ is generated according to the {SBM} in Example 1.
    We set $\bPhi=\beta_1 \cdot \tilde\bA$ with $\beta_1 = 0.5$, where $\tilde\bA$ is the row-normalized adjacency matrix.
    The network size is given as $N=400$ and the time length is specified as $T\in\brackets{50,100,\cdots,450}$.
    For the CNAR model, we take $K=2$.

    We estimate $\beta_1$ for NAR, $\bB = \beta_1\bD_1^{-1} \bTheta^\top\wt\bA \bTheta\bD_1^{-1}$ for CNAR (see Remark 1).
    The ReMSE of the estimates and the 1-step predictions are depicted in Figure \ref{fig:example-2-boxplot}.
    While NAR achieves better estimation accuracy with small sample size (e.g., $T=50$), CNAR still estimates and predicts with comparable accuracy even under model misspecification.
    This again demonstrates that the CNAR model is flexible enough to capture the data structure.

    \begin{figure}[htpb!]
        \centering
        \begin{subfigure}[b]{0.35\textwidth}
            \centering
            \includegraphics[width=\linewidth]{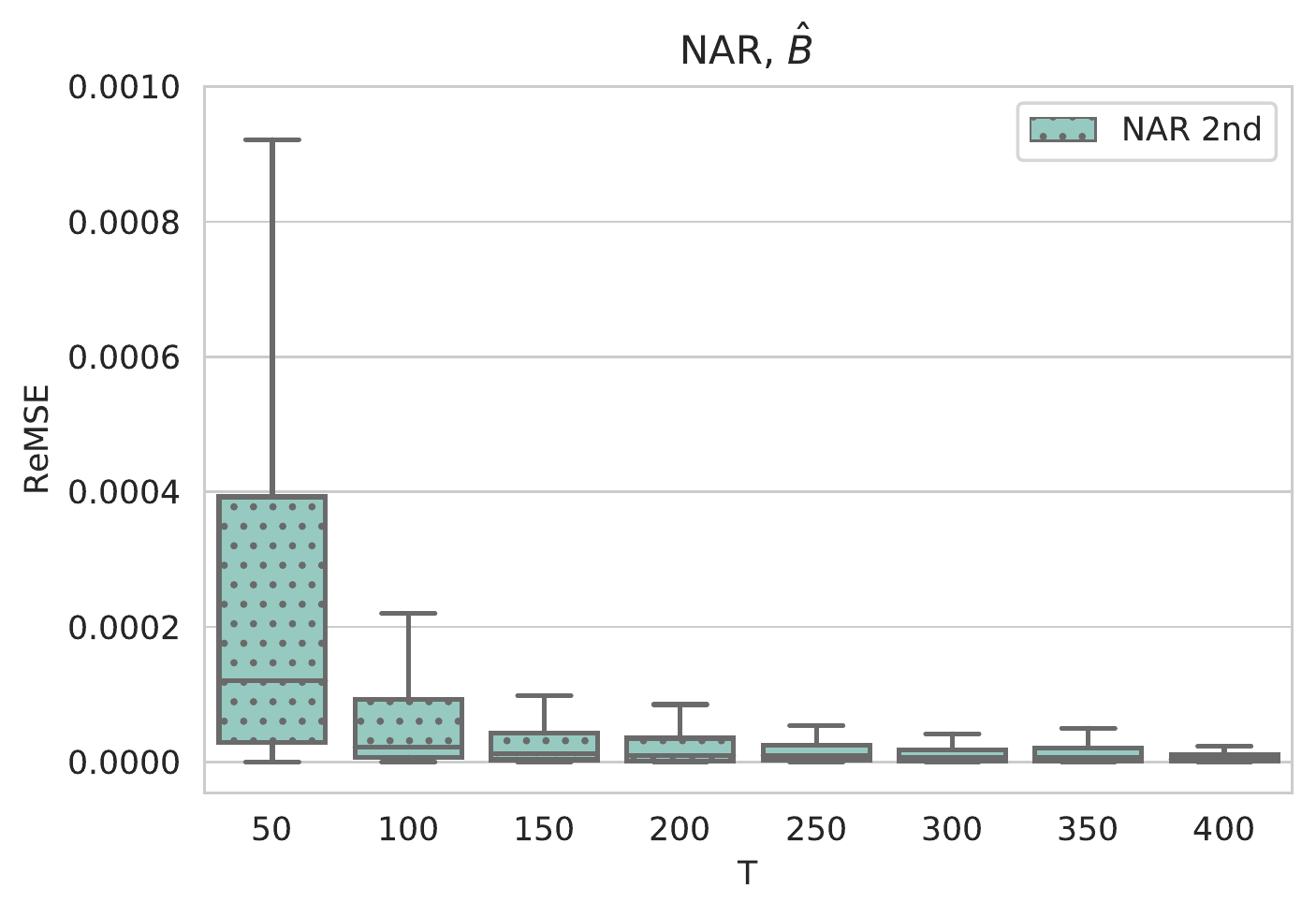}
        \end{subfigure}
        \hspace{3ex}
        \begin{subfigure}[b]{0.35\textwidth}
            \centering
            \includegraphics[width=\linewidth]{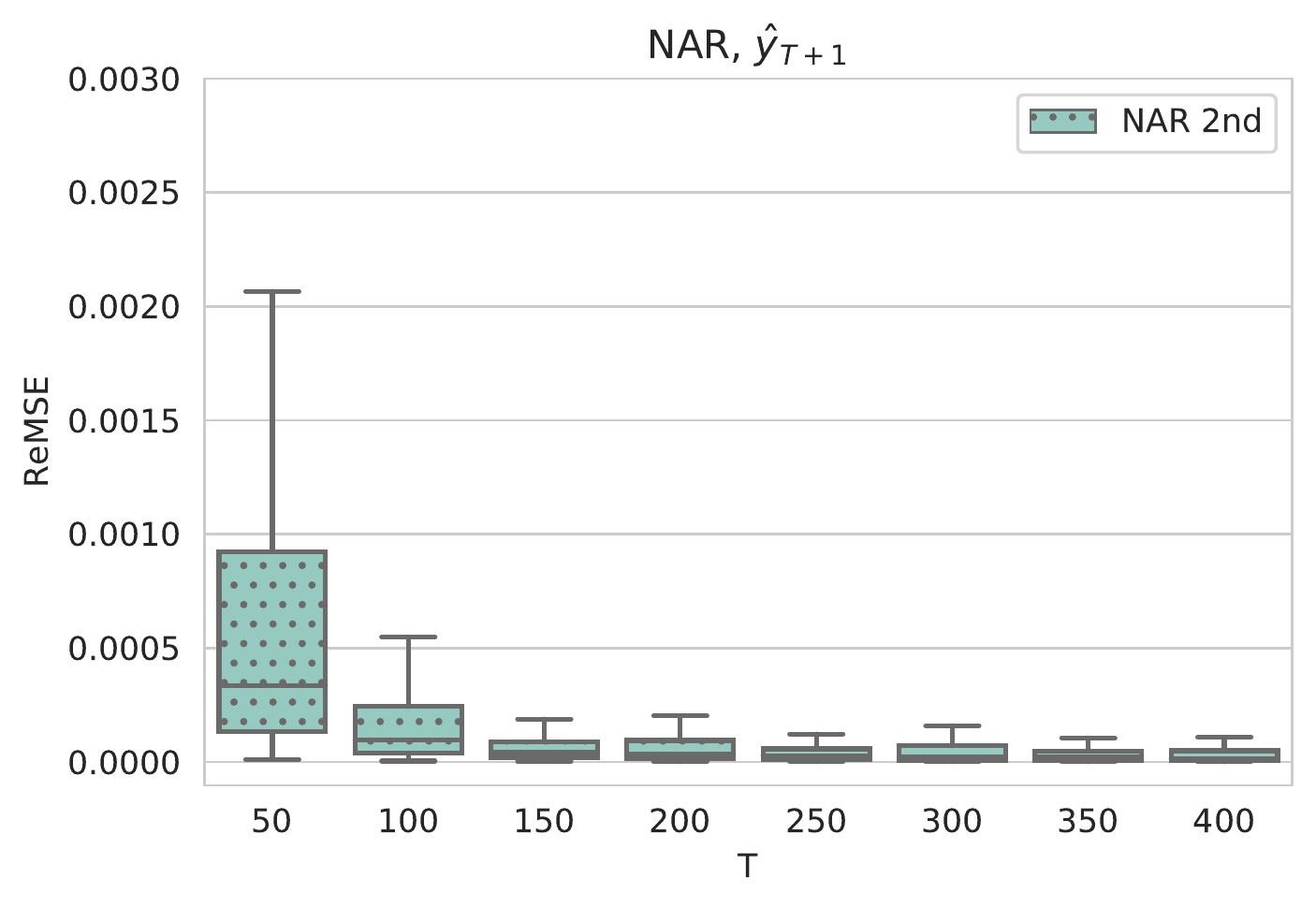}
        \end{subfigure}

        \begin{subfigure}[b]{0.35\textwidth}
            \centering
            \includegraphics[width=\linewidth]{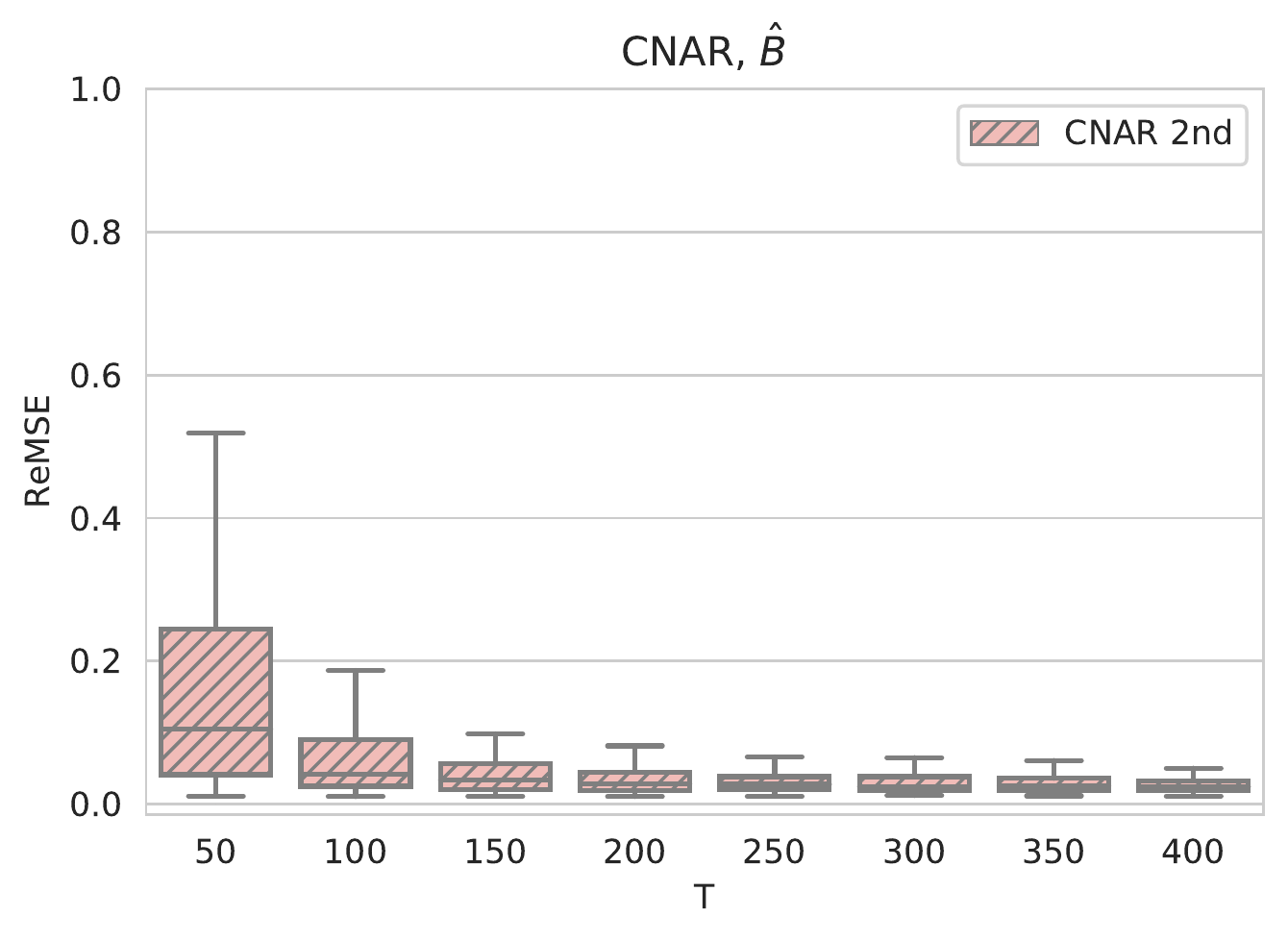}
        \end{subfigure}
        \hspace{3ex}
        \begin{subfigure}[b]{0.35\textwidth}
            \centering
            \includegraphics[width=\linewidth]{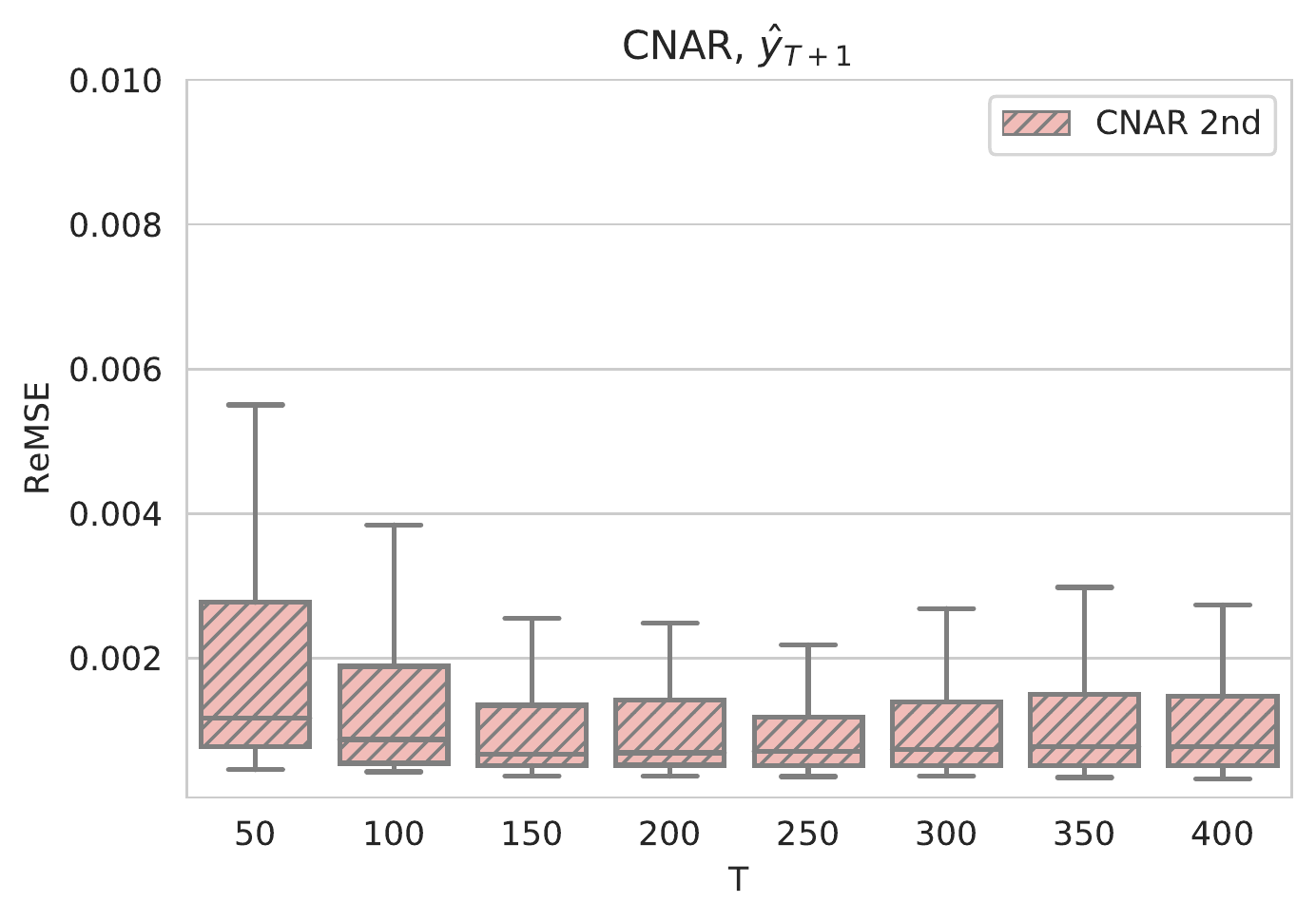}
        \end{subfigure}
        \caption{Example 3 (SBM and NAR model with $K=2$ and $N=400$). Box plots of ReMSE for estimated network AR coefficient $\bU\bB_1\bU^\top$ and 1-step prediction $y_{T+1}$ by NAR (first row) and CNAR (second row), respectively.
        }
        \label{fig:example-2-boxplot}
    \end{figure}

\end{enumerate}

\subsection{Comparison of CNAR 1st and 2nd-step estimators}

We investigate the different rates of convergence of the CNAR 1st and 2nd-step estimator when we increase $N$.
The data are generated following {\sc Example} 1 with $T = 400$, $K=2$ and $N$ varying in $\braces{200, 400, 800}$.
Figure \ref{fig:example-1-boxplot-1vs2} presents the ReMSE of estimated network AR coefficient $\bU\bB_1\bU^\top$ and 1-step prediction $y_{T+1}$ by CNAR 1st (first row) and 2nd-step estimator (second row), respectively.
We observe that estimation error of the 1st-step estimator does not change when $N$ increases while that of the 2nd-step estimator decreases.
This confirms our theoretical results that the 1st-step estimator converges at the rate $1/\sqrt{T}$ while the 2nd-step estimator converges at the rate $1/\sqrt{NT}$.

\begin{figure}[htpb!]
    \centering
    \begin{subfigure}[b]{0.35\textwidth}
        \centering
        \includegraphics[width=\linewidth]{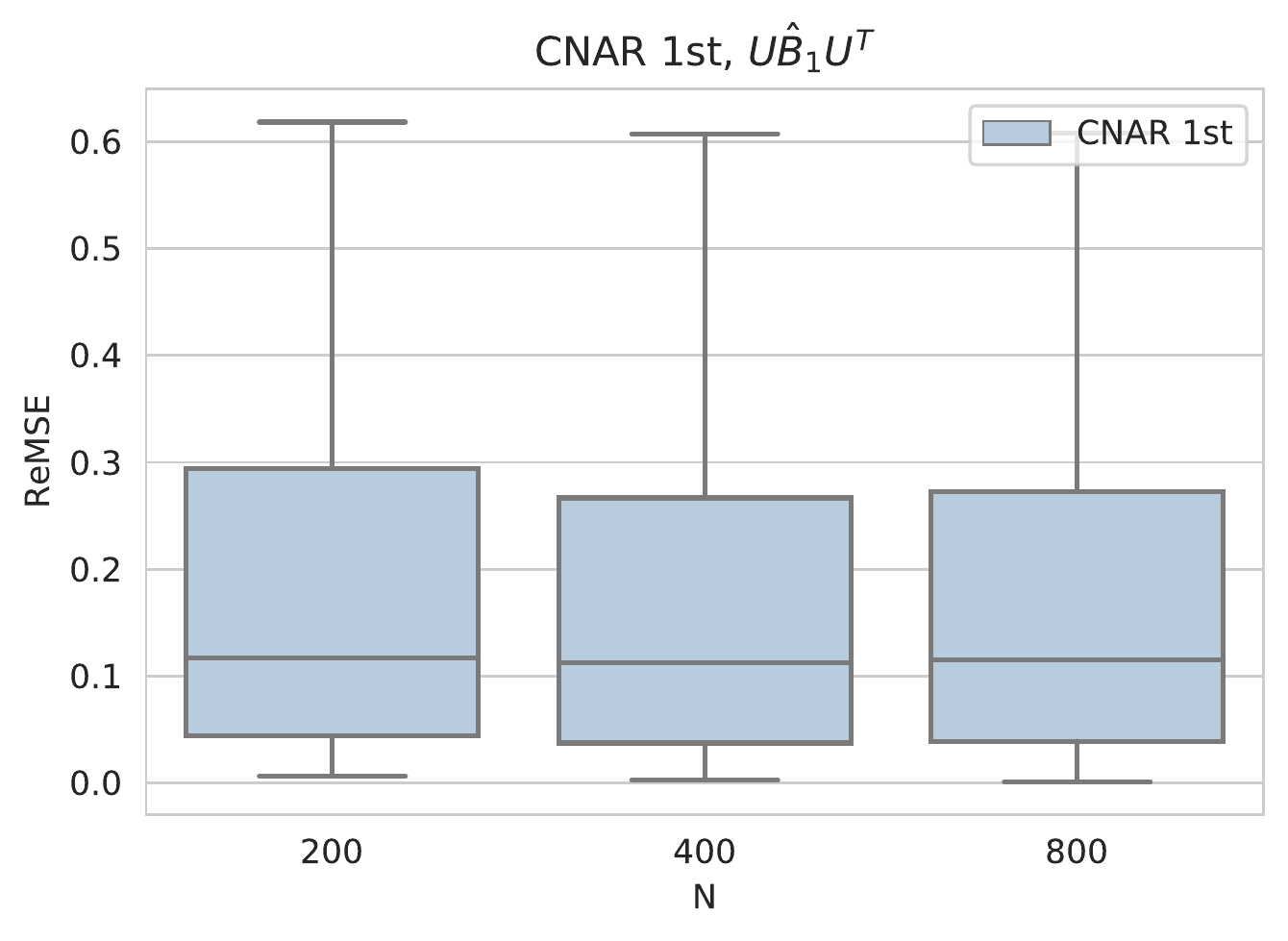}
    \end{subfigure}
    \hspace{3ex}
    \begin{subfigure}[b]{0.35\textwidth}
        \centering
        \includegraphics[width=\linewidth]{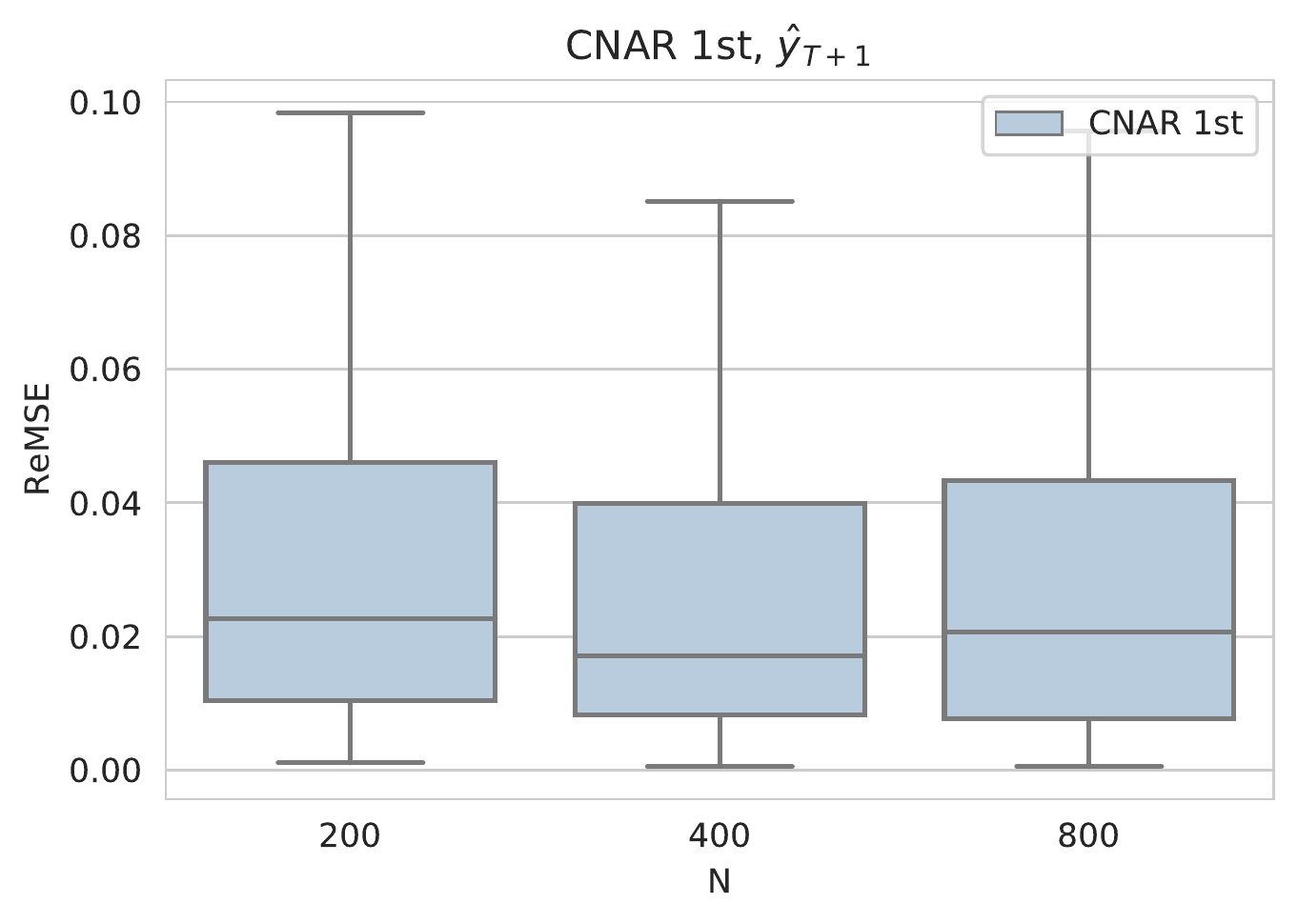}
    \end{subfigure}

    \begin{subfigure}[b]{0.35\textwidth}
        \centering
        \includegraphics[width=\linewidth]{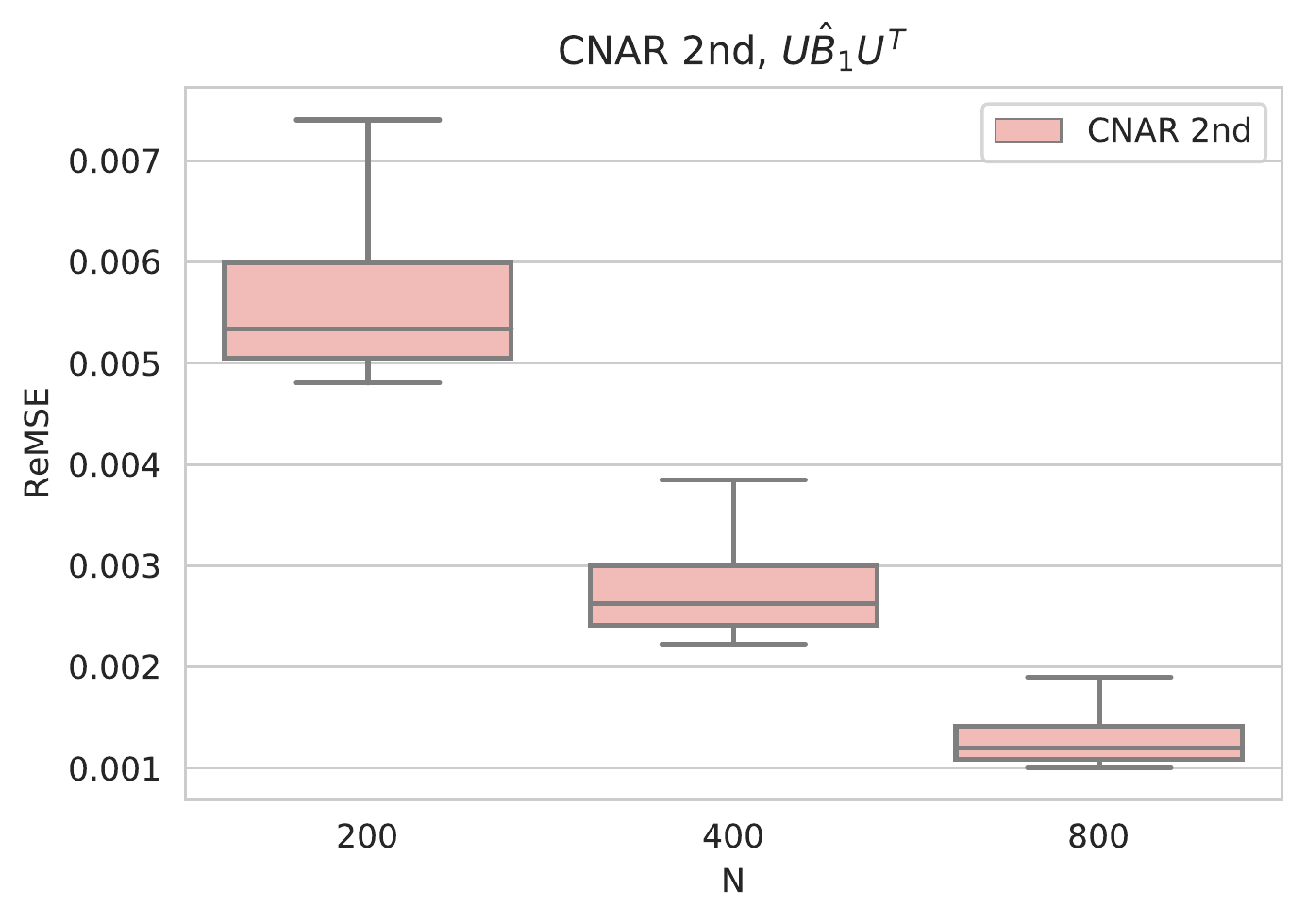}
    \end{subfigure}
    \hspace{3ex}
    \begin{subfigure}[b]{0.35\textwidth}
        \centering
        \includegraphics[width=\linewidth]{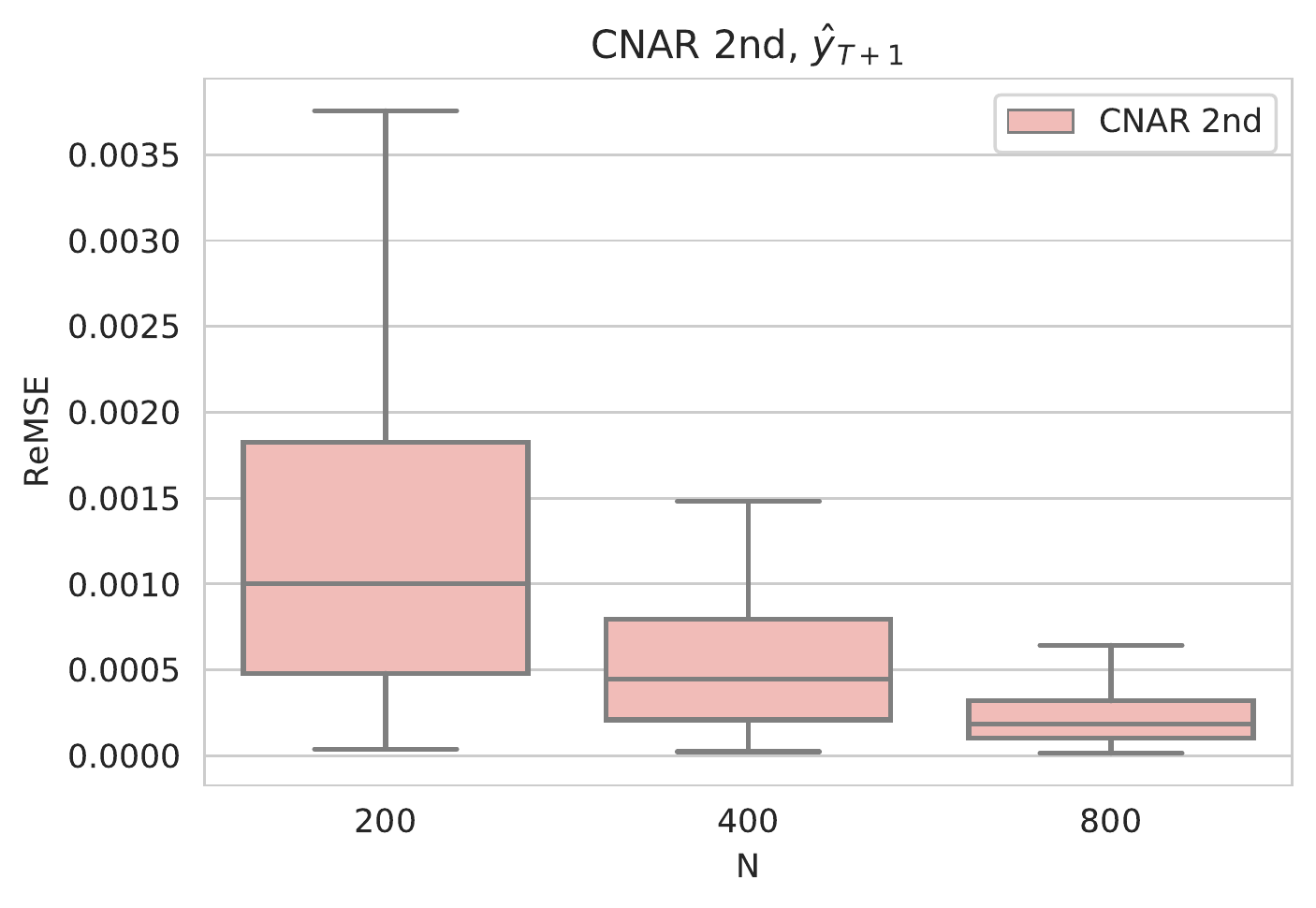}
    \end{subfigure}
    \caption{Example 1 box plots of estimated network AR coefficient $\bU\bB_1\bU^\top$ and 1-step prediction $y_{T+1}$ by CNAR 1st-step estimator (first row) and 2nd-step estimator (second row), respectively.
        For all plots, we have $K=2$ and $T=400$.}
    \label{fig:example-1-boxplot-1vs2}
\end{figure}


\section{Real Data Application}  \label{sec:appl}

In this section, we conduct data analysis using stock return data in Chinese A stock market,
which are traded in the Shanghai Stock Exchange and the Shenzhen Stock Exchange.
The response $Y_{it}$ is the daily stock return for $T = 244$ trading days of $N = 2187$ stocks
in the year of 2014.
Motivated by \cite{fama2015five},
we include the following sixcovariates related to corporations' fundamentals.
They are, {\sc SIZE} (i.e., the logarithm of market value),
{\sc BM} (i.e., book to market value),
{\sc PR} (i.e., yearly incremental profit ratio),
{\sc AR} (i.e., yearly incremental asset ratio),
{\sc LEV} (i.e., leverage ratio),
and {\sc CASH} (i.e., cash flow).

We first conduct a preliminary descriptive analysis of the data.
First, for each stock, we calculate the average return during the whole year, i.e.,
$T^{-1}\sum_t Y_{it}$.
This leads to the histogram in the left panel of Figure \ref{ave_return}.
The mean daily return of all stocks is given by 0.13\%.
Next, for each day, we calculate the average daily returns over all stocks, i.e.,
$N^{-1}\sum_i Y_{it}$.
This yields the line chart in the right panel of Figure \ref{ave_return}.
A higher volatility level can be captured during January to March and
November to December.

\begin{figure}[H]
    \centering
    \includegraphics[width=0.8\textwidth]{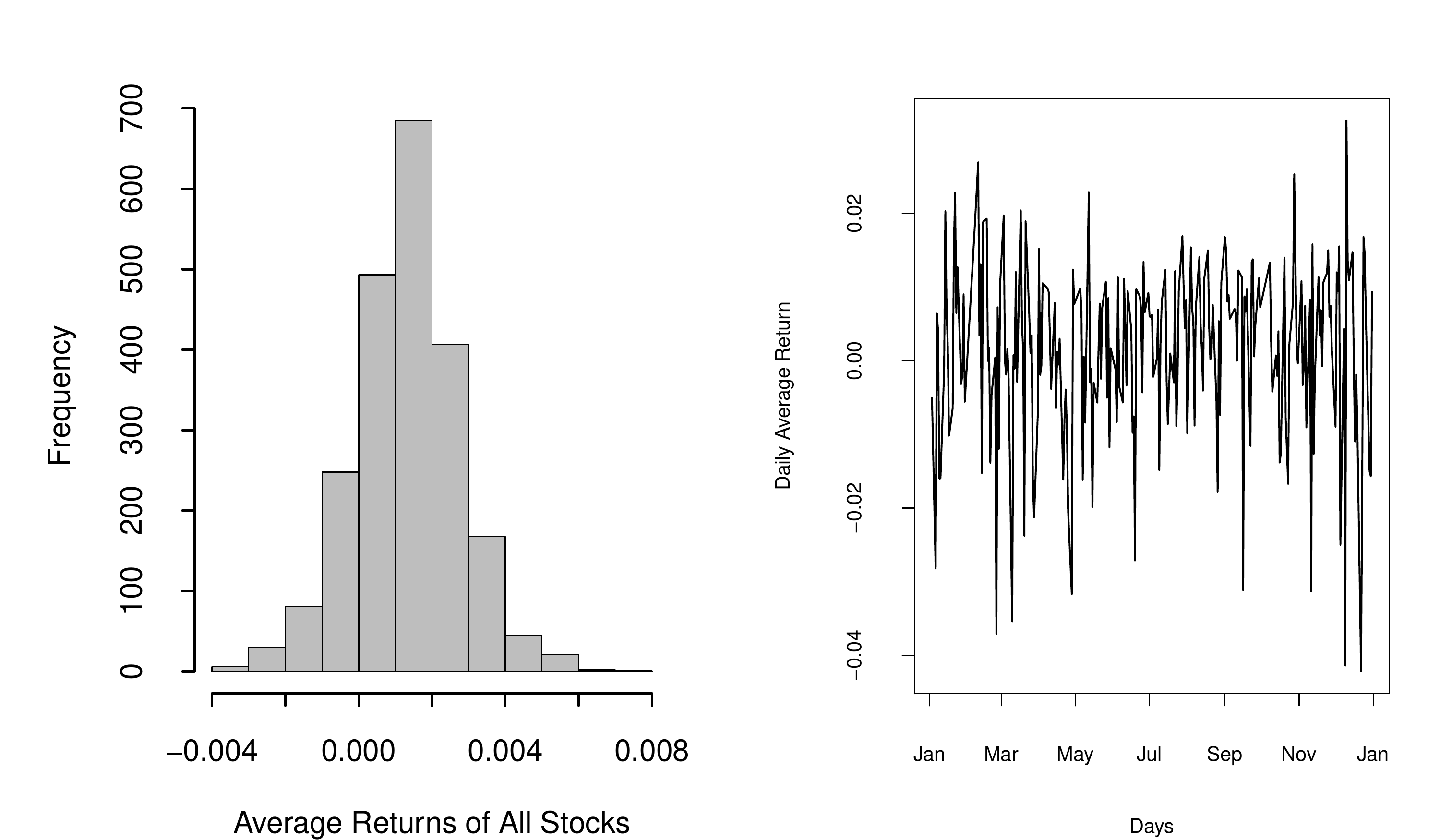}
    \caption{\small Left panel: the histogram of $N = 2187$ stocks' yearly average returns;
        right panel: the line plot of daily average returns of $N = 2187$ stocks.
    }\label{ave_return}
\end{figure}

We construct the network relationship by using the top ten shareholder information
of stocks.
Specifically, $a_{ij} = 1$ if the $i$th and $j$th stock shares at least one of the
top ten shareholders \citep{zhu2019network},
otherwise $a_{ij} = 0$.
We next conduct spectral clustering algorithm using the adjacency matrix
following \citep{lei2015consistency}.
By the screeplot of the eigenvalues (shown in Figure \ref{eigenvalue}),
we set the number of groups to be $K = 3$.
The group sizes are given as 78, 461, 1648 respectively.

\begin{figure}[H]
    \centering
    \includegraphics[width=3 in]{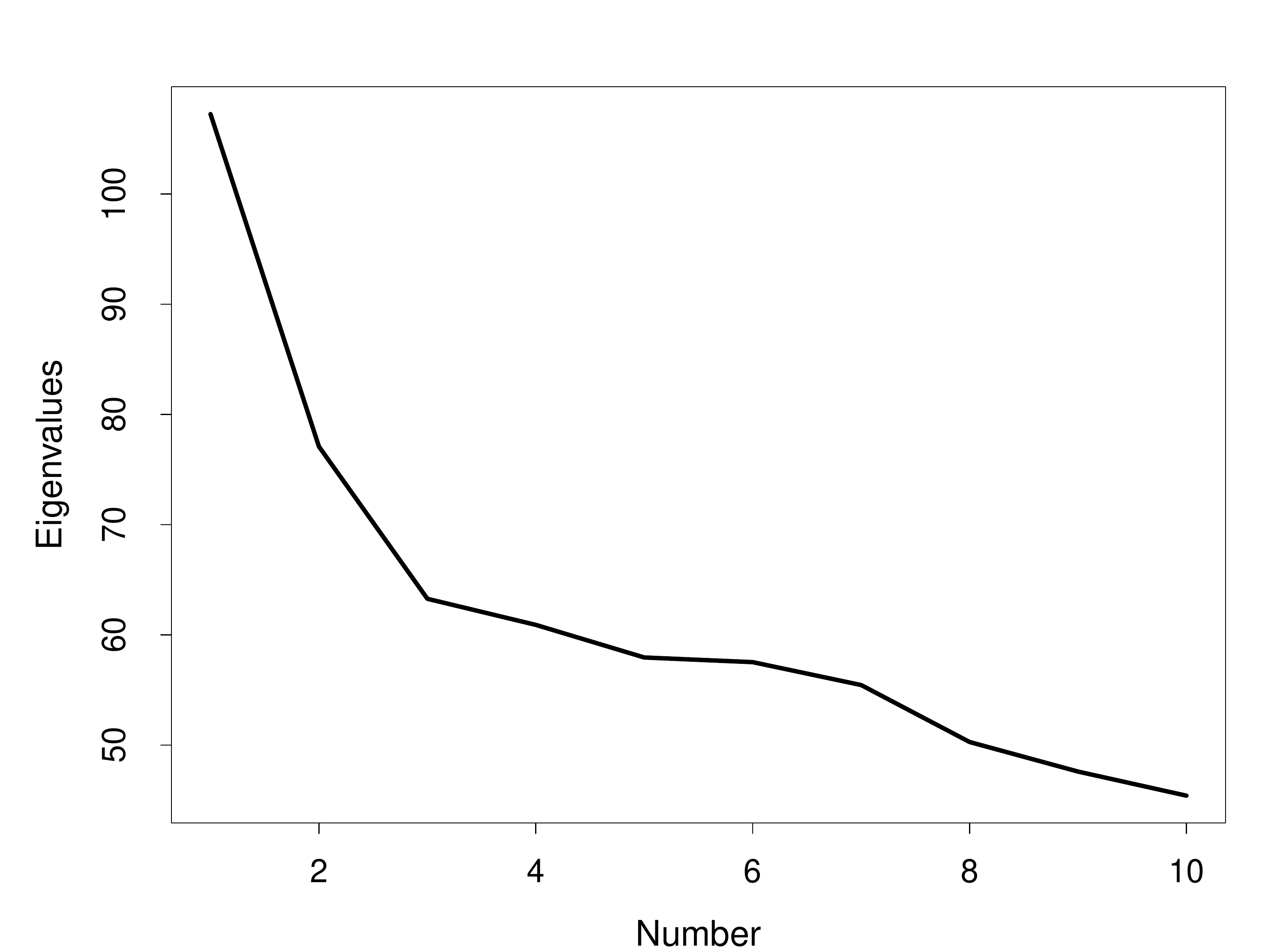}
    \caption{\small Decreasing eigenvalues of the adjacency matrix.
        The number of groups is set with $K = 3$.
    }\label{eigenvalue}
\end{figure}

We next estimate the parameters in the CNAR model using the data.
To evaluate the model performance, a sliding window approach is used.
The number of factors is selected to be $R = 3$ by using the screeplot approach.
The time window for model training is set to be $T_{train} = 150$ days.
Then, the subsequent $T_{test} = 25$ days are employed for model testing.
The mean square prediction error is calculated as
$\mbox{MPSE} = (NT_{test})^{-1}\sum_{t = t_{test,0}}^{t_{test,0}+T_{test}} \sum_i(\wh Y_{it} - Y_{it})^2$,
where $\wh Y_{it}$ is the predicted response and $t_{test,0}$
is the first time point of the testing data.
Further define a baseline MSPE as $\mbox{MSPE}_{0} = (NT_{test})^{-1}\sum_{t = t_{test,0}}^{t_{test,0}+T_{test}} \sum_i(\wh Y_{it} - \hat \mu_{i,train})^2$,
where $\hat \mu_{i,train} = T_{train}^{-1}\sum_{t = t_{test,0} - T_{train}}^{t_{test,0}-1} Y_{it}$ is the mean responses of the $i$th node within the training period.
A relative prediction error is then defined as
$\mbox{ReMSPE} = \mbox{MSPE}_{pred}/\mbox{MSPE}_{0}$.
For model comparison, the CNAR model is compared with the NAR model \citep{zhu2017network}
in terms of the prediction accuracy.
Specifically, the ReMSPE are given in Figure \ref{pred}.
As one could observe,
the prediction accuracy of the CNAR model is about 20\%--30\% higher than the NAR model,
which illustrates the prediction power of the proposed methodology.

\begin{figure}[H]
    \centering
    \includegraphics[width=5 in]{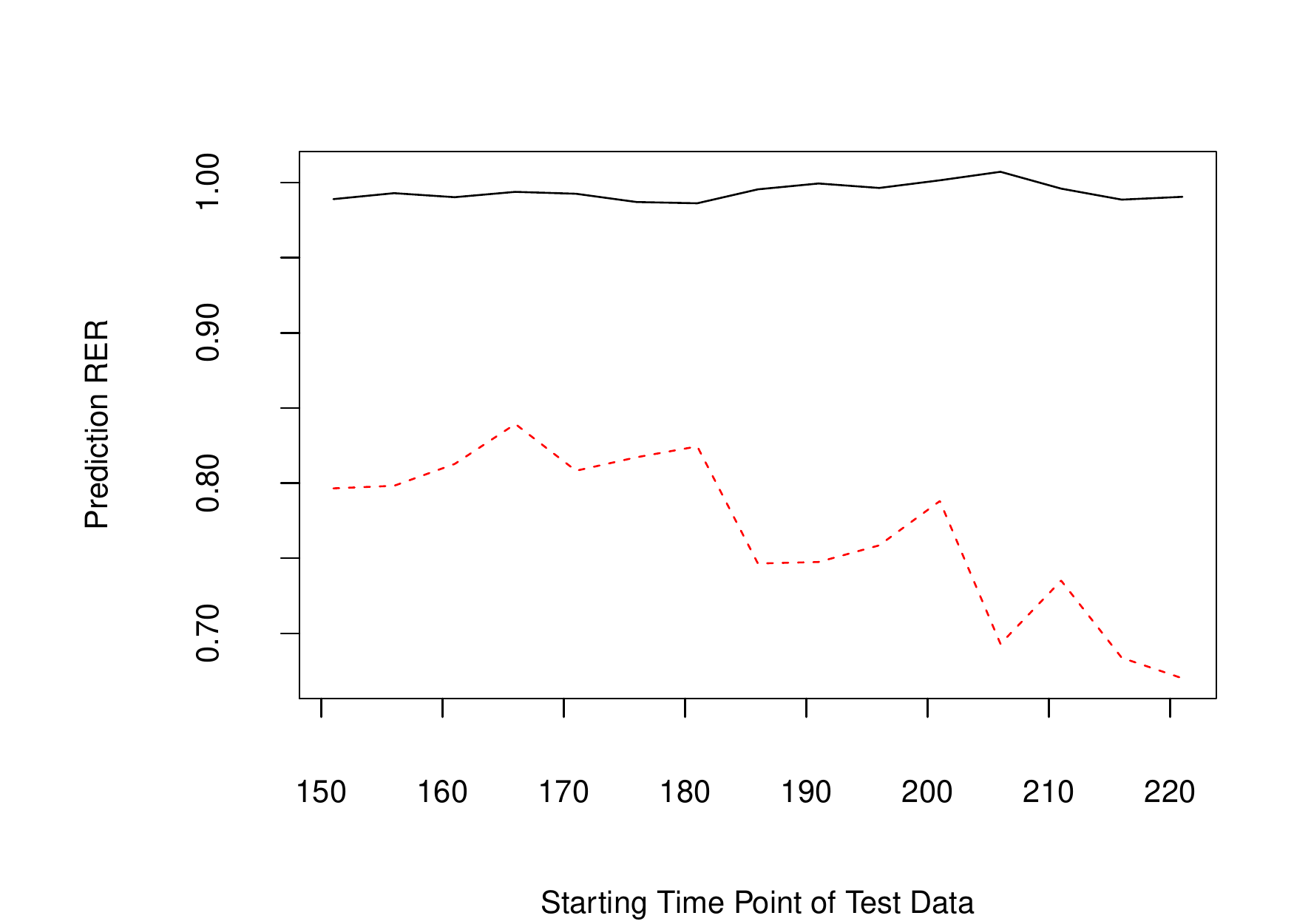}
    \caption{\small The ReMSPE of the CNAR model (red dash line) and the NAR model (black line). The CNAR model could achieve higher and more stable prediction accuracy rate than the NAR model.
    }\label{pred}
\end{figure}

\section{Summary} \label{sec:summ}

Modeling a high-dimensional time series with network dependence is a difficult yet commonly encountered problem in real-world applications.
The CNAR model mitigates the problem by incorporating the community structure in a network, yet provides more flexibility than the NAR model by allowing network coefficients to vary across different communities.
In addition, it takes account of possible non-network-related factors and thus corrects the unrealistic assumption of the NAR model that the noises are cross-sectional uncorrelated.
Both theoretical and empirical results confirm the advantage of CNAR model.

There are several interesting future research topics under the general framework of the CNAR model.
First, the theoretical analysis of the CNAR is based on the assumption of a underlying stochastic block model.
Generally speaking, any community structure will work if $\hat\bU$ can be estimated consistently.
This is illustrated empirically by a simulation study here.
Therefore, it is of great interest to investigate this problem formally under different network structures.
Second, the network structure discussed here is assumed to be static.
However, in reality the network structure typically evolve slowly with time.
Incorporating in the proposed CNAR model with a time-varying network structure is another interesting problem worth pursuing.
Last, in this work we allow for estimating network effects for a diverging number of communities.
In future study one could further consider high dimensional covariates and discuss variable screening and selection procedures under the model framework.


%
%

\clearpage
\bibliographystyle{\mybibsty}
\bibliography{\mybib}

%
%
\clearpage
\setcounter{page}{1}
\begin{appendices}

\section{Proofs of Parameter Estimation} \label{sec:proof}

\subsection{Proof of Theorem \ref{stat1}}\label{stat_proof}

Note that the solution given by Theorem \ref{stat1} satisfies the CNAR model.
We then verify that the solution satisfies strict stationarity.
Note that
\beq
\rho(\bG)\le \rho(\bU^\top\bB_1\bU) + |\beta_2|.\nonumber
\eeq
Specifically we have
$\rho(\bU^\top\bB_1\bU) = \max_{\|\bv\| = 1}|\bv^\top\bU^\top\bB_1\bU\bv|\le
\|\bU \bv\|^2\max_{\|\bv^*\| = 1}|\bv^{*\top}\bB_1\bv^*| = \rho(\bB_1)$.
It yields that
\beq
\rho(\bG)\le \rho(\bB_1) + |\beta_2|  <1\label{rho_G}
\eeq
by the condition in Theorem \ref{stat1}.
It then holds that $\lim_{m\rightarrow{\infty} }\sum_{j = 0}^m \bG^j\wt\bepsilon_{t-j}$ exists and is a strictly stationary process.

We next verify that the uniqueness of the stationary solution.
Assume $\{\wt\by_t\}$ is another stationary solution to the CNAR model (\ref{eqn:cnar})
with $E\|\wt\by_t\|_1<\infty$.
Then we have $\wt \by_t = \sum_{j = 1}^{m-1}\bG^j\wt\bepsilon_{t-j}
+ \bG^m \wt\by_{t-m}$ for any integer $m\ge 1$.
Consequently, by (\ref{rho_G}) one can conclude
$E\|\by_t - \wt\by_t\|_1 = E\|\sum_{j = m}^\infty \bG^j(\wt \bepsilon_{t-j}) - \bG^m\wt\by_{t-m}\|_1\le c\rho(\bG)^m$,
where $c$ is a finite constant.
Here $m$ is chosen arbitrarily. Hence we have $E\|\by_t - \wt\by_{t}\|_1 = 0$,
which implies $\by_t = \wt\by_t$ with probability 1.
This completes the proof.

\subsection{Proof of Theorem \ref{thm:first_step_estimator}} \label{append:first_step_estimator}

In the following we first prove that the asymptotic normality holds for $\wt\btheta$.
Then the result could also be obtained for $\wh\btheta^{(1)}$
by following the similar procedure.

{\sc Step 1. (Proof of $\balpha_{NT} \Sigma_x(\wt\btheta - \btheta)\rightarrow_d N(\zero, \Sigma_{2\ve})$)}
Recall that we have
\beq
\wt\btheta = \btheta + \big(\sum_t\bX_{t-1}^\top\bX_{t-1}\big)^{-1}\big(\sum_t\bX_{t-1}^\top\bepsilon_t\big).\nonumber
\eeq
It suffices to deal with $\wt \Sigma_x \defeq (NT)^{-1}\sum_t\bX_{t-1}^\top\bX_{t-1}$
and $\wt \Sigma_{x\ve} \defeq (NT)^{-1}\sum_t\bX_{t-1}^\top\bepsilon_t$ separately.
{Then we have
$\Sigma_x(\wt \btheta - \btheta) = -(\wt\Sigma_x - \Sigma_x)(\wt\btheta - \btheta)
+ \Sigma_{x\ve}$.}
Note that $\wt\Sigma_x$ and $\wt\Sigma_{x\ve}$ take the form as
\begin{align}
&\wt\Sigma_{x} = \frac{1}{NT}
\left(
\begin{array}{ccc}
(\sum_t  \bU^\top \by_{t-1}\by_{t-1}^\top \bU) \otimes \bI_K
& \sum_t (\bU^\top\by_{t-1})\otimes (\bU^\top \by_{t-1})
& \sum_t (\bU^\top\by_{t-1})\otimes (\bU^\top \bZ_{t-1}) \\
& \sum_t \by_{t-1}^\top \by_{t-1}
& \sum_t \by_{t-1}^\top\bZ_{t-1} \\
&
& \sum_t\bZ_{t-1}^\top\bZ_{t-1}
\end{array}
\right),\label{wt_Sigx}\\
&\mbox{~~~and~~~}
\wt\Sigma_{x\ve} =
\frac{1}{NT}\left(
\begin{array}{c}
\sum_t (\bU^\top \by_{t-1})\otimes (\bU^\top \bepsilon_t) \\
 \sum_t \by_{t-1}^\top\bepsilon_t\\
 \sum_t \bZ_{t-1}^\top\bepsilon_t
\end{array}
\right).\nonumber
\end{align}
Let $\bG = \bU\bB_1\bU^\top + \beta_2\bI_N$.
Note that $\by_t = \sum_{j = 0}^\infty
\bG^j(\bepsilon_{t-j} +
\bZ_{t-1-j}\bgamma)\defeq\sum_{j = 0}^\infty \bG^j\wt\bepsilon_{t-j}
$.
Without loss of generality, we assume $E(\bZ_t) = \zero$ in the following.

{\sc (1) $\wt\Sigma_x\rightarrow_p \Sigma_{x}$.}
{The result can be directly obtained by (\ref{del_Sigx}) of Lemma \ref{vSigx}.}

{\sc (2) $\balpha_{NT}\wt \Sigma_{x\ve} \rightarrow_d N(0,\Sigma_{2\ve})$.}
It suffices to show that
for any $\bfeta \in \RR^{K^2+p+1}$
we have $\bfeta^\top \balpha_{NT} \wt \Sigma_{x\ve}\rightarrow_d N(0,\bfeta^\top \Sigma_{2\ve}\bfeta)$.
Write
$\Sigma_{x\ve,t} = \bX_{t}^\top\bepsilon_t$.
Hence we have $\wt\Sigma_{x\ve} = \sum_t\Sigma_{x\ve,t} = \sum_t \bX_{t-1}^\top\bepsilon_t$.
Accordingly, denote $\xi_{Nt} = \bfeta^\top \wt\balpha_{NT}\Sigma_{x\ve,t} $
and $\mF_{Nt} = \sigma\{\bepsilon_{is}, \bZ_{i(s-1)}:1\le i\le N, -\infty <s\le t\}$,
where $\wt\balpha_{NT} = \diag\{N^{-1}T^{-1/2}\bI_{K^2}, N^{-1}T^{-1/2}, (NT)^{-1/2}\bI_{p}\}$.
Then $\{\sum_{s = 1}^t \xi_{Ns}, \mF_{Nt},-\infty<t\le T, N\ge 1\}$
constitutes a martingale array.
Then we employ the central limit theorem  (Corollary 3.1 of \cite{hall2014martingale}) to obtain the result.

First we have
\begin{align*}
\sum_{t = 1}^TE\{\xi_{Nt}^2I(|\xi_{Nt}|>\delta)|\mF_{N,t-1}\}
\le \delta^{-2}\sum_t E\{\xi_{Nt}^4|\mF_{N,t-1}\}
\end{align*}
Note that $\xi_{Nt}^2 = \bepsilon_t^\top \bX_{t-1}\wt\balpha_{NT} \bfeta\bfeta^\top\wt\balpha_{NT}
\bX_{t-1}^\top\bepsilon_t$.
According to (\ref{var_Q1}) of Lemma \ref{quad_lem},
we have $E\{\xi_{Nt}^4|\mF_{N,t-1}\}
\le c(NT)^{-2}(\bfeta^\top\balpha_{NT} \bX_{t-1}^\top\Sigma_{\ve}\bX_{t-1}\balpha_{NT} \bfeta)^2$,
where $c$ is a finite constant.
To verify the asymptotic normality, we prove the following two results.

{
(a) $\sum_t(\bfeta^\top\wt\balpha_{NT} \bX_{t-1}^\top\Sigma_{\ve}\bX_{t-1}\wt\balpha_{NT} \bfeta)^2\rightarrow_p0$.

Due to the similarity, we verify that
$N^{-4}T^{-2}\sum_t (\by_t^\top\Sigma_\ve \by_t)^2\rightarrow_p0 $.
It suffices to verify that
$N^{-4}T^{-2}\sum_t (\by_t^\top\bLambda\Sigma_f\bLambda^\top \by_t)^2\rightarrow_p0 $ and
$N^{-4}T^{-2}\sum_t (\by_t^\top\Sigma_e \by_t)^2\rightarrow_p0 $.
Note that $\by_t$ follows $N(\zero, \Gamma_y(0))$.
For the first one we have
$N^{-4}T^{-2}\sigma_1(\Sigma_f)^2\sum_t (\by_t^\top\bLambda\bLambda^\top \by_t)^2$.
Then $E\{(\by_t^\top\bLambda\bLambda^\top \by_t)^2\}
\le \{\tr(\bLambda^\top\Gamma_y(0)\bLambda)\}^2 \le c\sigma_1(\Gamma_y(0))^2 = O(N^2)$ by using (\ref{sigma_Gamma}).
Hence $N^{-4}T^{-2}\sigma_1(\Sigma_f)^2\sum_t E\{(\by_t^\top\bLambda\bLambda^\top \by_t)^2\} = O(T^{-1})\rightarrow 0$.
Then the result holds.
The second could be similarly proved by noticing $\sigma_1(\Sigma_e) = O(1)$.

(b) $\sum_t E\{\xi_{Nt}^2|\mF_{N,t-1}\} =
\sum_t \bfeta^\top\wt \balpha_{NT}\bX_{t-1}^\top \Sigma_{\ve}\bX_{t-1}\wt \balpha_{NT}\bfeta
\rightarrow_p \bfeta^\top\Sigma_{2\ve}\bfeta$.

Denote $\wt \Sigma_{2\ve} = \sum_t \wt \balpha_{NT}\bX_{t-1}^\top \Sigma_{\ve}\bX_{t-1}\wt\balpha_{NT}$.
Following the same procedure of proving (\ref{del_Sigx}), we could show that
\begin{align*}
P\big\{\big|\bfeta^\top(\wt\Sigma_{2\ve} - \Sigma_{2\ve})\bfeta
\big|\ge \xi\big\}\le 2\exp(- c\min\{T^2\xi^2, T\xi\}).
\end{align*}
Hence the conclusion (b) holds.

}

{\sc Step 2. (Proof of $\balpha_{NT}\Sigma_x(\wh\btheta^{(1)} - \btheta) \rightarrow_d N(\zero, \Sigma_{2\ve})$).}

{
First, assuming $\gamma_N\gg \sqrt{\log N}$ and $ T\gg (\log N/\gamma_N)^2$, we have
$\wh\Sigma_x- \Sigma_x = o_p(1)$ by using (\ref{del_hat_Sigx}).
Next, $\balpha_{NT}\wh \Sigma_{x\ve} \rightarrow_d N(0,\Sigma_{2\ve})$
could be similarly proved as Step (2) previously,
where $\wh \Sigma_{x\ve} = (NT)^{-1}\sum_t\wh\bX_{t-1}^\top\bepsilon_t$.
Particularly in proving (b) we use
\begin{align*}
P\big\{\big|\bfeta^\top(\wh\Sigma_{2\ve} - \Sigma_{2\ve})\bfeta
\big|\ge \xi\big\}\le 2\exp(- c\min\{T^2\xi^2, T\xi\})+
 c\frac{\sqrt K\log N}{\sqrt T\xi\gamma_N} + N^{-1}.
\end{align*}
by following (\ref{del_hat_Sigx}), where $\wh \Sigma_{2\ve} =
\sum_t \wt\balpha_{NT}\wh\bX_{t-1}^\top \Sigma_{\ve}\wh\bX_{t-1}\wt\balpha_{NT}$.

}

{

}

\subsection{Proof of Theorem \ref{thm_first_step_K_diverge}}\label{thm_first_step_K_diverge_proof}


{\sc 1. Proof of (\ref{Koracle_normal}).}
First we have
\begin{align*}
\balpha_{NT}\Sigma_x(\wt \btheta - \btheta)
 = -\balpha_{NT}(\wt \Sigma_x - \Sigma_x)(\wt \btheta - \btheta)
 + \frac{\balpha_{NT}}{NT}\big(\sum_t \bX_{t-1}^\top \bepsilon_{t}\big),
\end{align*}
where $\wt\Sigma_x$ is defined in (\ref{wt_Sigx}).
To prove the result, it suffices to show that for any $\bA_K\in \RR^{m\times (K^2+p+1)}$ with $\bA_K\bA_K^\top\rightarrow \bH$,
we have 
$\|(\wt\Sigma_x - \Sigma_x)(\wt\btheta - \btheta)\| =
o_p(\|\Sigma_x(\wt\btheta - \btheta)\|)$
and $\bA_K\Sigma_{2\ve}^{-1/2}\wt\balpha_{NT}$ $(\sum_t\bX_{t-1}^\top\bepsilon_t)
\rightarrow_d N(\zero, \bH)$,
where $\Sigma_{2\ve}$ is defined in (\ref{Sig_x_2e})
and $\wt \balpha_{NT} = \balpha_{NT}/(NT)$.
For the first conclusion, it suffices to verify that
$\max_i |\lambda_i(\wt\Sigma_x - \Sigma_x)| = o_p(\lambda_{\min}(\Sigma_x)).$
Note by (\ref{sup_delx1}) of Lemma \ref{sup_vSigx} and specifying $\xi
= \lambda_{\min}(\Sigma_x)$, $\max_i |\lambda_i(\wt\Sigma_x - \Sigma_x)| = o_p(\lambda_{\min}(\Sigma_x))$ can be obtained
by the assumption that $K = o(T^{1/2})$.
We then show
\begin{align}
\bA_K\Sigma_{2\ve}^{-1/2}\wt\balpha_{NT}\Big(\sum_t\bX_{t-1}^\top\bepsilon_t\Big)
\rightarrow_d N(\zero, \bH)\label{normal_theta1}
\end{align}
in the following.

Note here $m$ is finite, therefore
it suffices to show that for any $\bfeta\in \RR^m$
with $\|\bfeta\| = 1$
that $\bfeta^\top\bA_K\Sigma_{2\ve}^{-1/2}\wt\balpha_{NT}(\sum_t\bX_{t-1}^\top\bepsilon_t)
\rightarrow_d N(\zero, \bfeta^\top\bH\bfeta).$
Denote $\xi_{Nt} = (NT)^{-1/2}\bfeta^\top \bA_K\Sigma_{2\ve}^{-1/2}\wt\balpha_{NT}\bX_{t-1}^\top \bepsilon_t$
and $\mF_{Nt} = \sigma\{\bepsilon_{is}, \bZ_{i(s-1)}:1\le i\le N, -\infty <s\le t\}$.
Then  $\{\sum_{s = 1}^t \xi_{Ns}, \mF_{Nt},-\infty<t\le T, N\ge 1\}$
constitutes a martingale array.
The rest of the proof can be obtained by following the {Step 1 (2)} in the proof of  Theorem \ref{thm:first_step_estimator} to employ the central limit theorem  (Corollary 3.1 of \cite{hall2014martingale})
and (\ref{sup_delx1}) of Lemma \ref{sup_vSigx}.
%
%
%

\noindent
{\sc 2. Proof of (\ref{Kfirst_normal}).}
When we use $\wh \bU$, we still have
\begin{align*}
\balpha_{NT}\Sigma_x(\wh \btheta^{(1)} - \btheta) =
-\balpha_{NT}(\wh\Sigma_x - \Sigma_x)(\wh\btheta^{(1)} - \btheta) +
\frac{\balpha_{NT}}{NT}\Big(\sum_t \wh\bX_{t-1}^\top\bepsilon_t
\Big).
\end{align*}
The rest of the proof follows the proof of (\ref{Koracle_normal}) but using (\ref{sup_delx2}) of
Lemma \ref{sup_vSigx}
and $T\gg (\sqrt K\exp(K^2\log 21)\log N/\gamma_N)^2$.

\subsection{Proof of Theorem \ref{thm_sec_step_K_diverge}}\label{proof_sec_est}

Define $\Omega_\ve = \Sigma_\ve^{-1}$.
Let $\wh \Omega_\ve = \wh\Sigma_\ve^{-1}$ be the estimated precision
matrix after the first step estimation.
Correspondingly, let $\wh \Sigma_{2\theta} =
(NT)^{-1}\sum_t\wh\bX_t^{\top}\wh\Omega_\ve\wh\bX_t$.
We have the following expression,
\beq
\sqrt{NT}\Sigma_{2\theta}(\wh\btheta^{(2)} - \btheta) =
-\sqrt{NT}(\wh\Sigma_{2\theta} - \Sigma_{2\theta})(\wh\btheta^{(2)} - \btheta) +
(NT)^{-1/2}\big(\sum_t \wh\bX_{t-1}^\top \Omega_{\ve} \ve_t\big)\nonumber
\eeq
It suffices to show
for any $\bA_K\in \RR^{m\times (K^2+p+1)}$ with $\bA_K\bA_K^\top \rightarrow \bH$, that
\begin{align}
&\big\|(\wh \Sigma_{2\theta} - \Sigma_{2\theta})(\wh\btheta^{(2)} - \btheta)\big\| =
o_p\big(\big\|\Sigma_{2\theta}(\wh\btheta^{(2)} - \btheta)\big\|\big)\label{sig2theta1}\\
& (NT)^{-1/2}\bA_K \Sigma_{2\theta}^{-1/2}\big(\sum_t \wh\bX_{t-1}^\top \Omega_{\ve} \ve_t\big) \rightarrow_d N(\zero, \bH).\label{sig2theta2}
\end{align}
(\ref{sig2theta1}) can be obtained by using $\max_i|\lambda_i(\wh\Sigma_{2\theta} - \Sigma_{2\theta})| = o_p(\lambda_{\min}(\Sigma_{2\theta}))$ by (\ref{del_hatSig2theta}) of Lemma \ref{sup_vSigx}.
Next, (\ref{sig2theta2})
can be subsequently proved by using the central limit theorem  (Corollary 3.1 of \cite{hall2014martingale}) of the martingale sequence by following the proof of (\ref{normal_theta1}) in the proof of  Theorem \ref{thm_first_step_K_diverge}.

\subsection{Theoretical Properties under $\bU^\top\bLambda = \zero$}

In this scenario, we present the theoretical properties in the following theorem.

\bet\label{coro_ulam0}
Assume the same conditions as in Theorem \ref{thm:first_step_estimator} but
$\bU^\top\bLambda = \zero$.
Define $\bSigma_x$ and $\bSigma_{2\ve}$ as in (\ref{Sig_x_2e}) with modifications:
$\bSigma_u = \lim_{N\rightarrow\infty}\bU^\top\bGamma_y(0)\bU$,
$\bSigma_{uy} = \zero$,
$\bSigma_{ue} = \lim_{N\rightarrow\infty}\{\bU^\top\bGamma_y(0)\bU\}
\otimes (\bU^\top\bSigma_\ve \bU)$,
$\bSigma_{uye} = \zero$.
Then we have
$\balpha_{NT}\bSigma_x(\wt \btheta - \btheta)\rightarrow_d N(\zero, \bSigma_{2\eps})$.
\eet

Under $\bU^\top\bLambda = \zero$, the convergence rate of $\wt\btheta$ is still the same but
$\wt \bB_1$ is asymptotically uncorrelated with $\wt \beta_2$.
The proof is given as follows.

\begin{proof}

Under the case $\bU^\top \bLambda\ne \zero$,
we replace $\bX_{t-1}$ by
 $\bX_{t-1} = ((\by_{t-1}^\top \bU_N)\otimes \bU_N, \by_{t-1}, \bZ_{t-1})\in\RR^{N\times (K^2+p+1)}$
 with $\bU_N = N^{1/4}\bU$,
and $\btheta^b = (\vec(\wt \bB_1)^\top, \beta_2, \bgamma^\top)^\top$
with $\wt\bB_1 = \bB_1/\sqrt N$.
In addition, define $\wh \btheta^b$ as the estimator for $\btheta^b$.
Denote $\wt\Sigma_x$ and $\wt\Sigma_{x\ve}$ as in (\ref{wt_Sigx}) except replacing $\bU_N$ with $\bU$.
In addition, define $\wt \balpha_{NT} = \diag\{\sqrt{NT}\bI_{K^2}, \sqrt T, \sqrt{NT}\bI_{p}\}$.

Note that $\bU^\top\bLambda = \zero$ and $\wh \bU^\top \bLambda = \zero$, then we have
$\wh\bU_N^\top\by_{t - 1} = \wh\bU_N^\top\by_{t - 1}^*$
and $\wh\bU_N^\top\bepsilon_{t - 1} = \wh\bU_N^\top\bepsilon_{t - 1}^*$ in (\ref{wt_Sigx}),
where $\by_{t - 1}^* = \sum_{j = 0}^\infty
\bG^j(\bepsilon_{t-1-j}^* +
\bZ_{t-2-j}\bgamma)$ and $\bepsilon_{t - 1}^* = \bZ_{t-1}\bgamma+ \be_{t-1}$.
As a result, by
following the same procedure as in the proof of Theorem \ref{thm:first_step_estimator}, we can show that
$\wt \balpha_{NT}\Sigma_x(\wh \btheta^b - \btheta) \rightarrow_d N(\zero, \Sigma_{2\ve})$.
Equivalently, we could achieve the final result given in Theorem  \ref{coro_ulam0}.
\end{proof}


\subsection{Technical Lemmas}

{In this section we present several technical Lemmas, which will be used in the proof of the main Theorems.
Specifically,
Lemma \ref{lem_y_conv} to \ref{mu_G} establish convergence conditions for the high dimensional time series $\{\by_t\}$.
Lemma \ref{thm:SBM-eigvec-L2-bound} provides
Frobenius or $\ell_2$ bounds of the leading $K$ eigenvectors of the adjacency matrix $\bA$ under SBM.
Lemma \ref{vSigx} to \ref{sup_vSigx} establish nonasymptotic bounds related to
$\{\by_t\}$,
which are critical for proving main theoretical results of the estimators.

}

\bel\label{lem_y_conv}
Let $\bx_t = (x_{1t},\cdots, x_{Nt})^\top\in\RR^N$, where $x_{it}$s are independent and identically distributed
 random variables with mean zero, variance $\sigma_x^2$ and finite fourth order moment.
Let $\wt\by_t = \sum_{j = 0}^\infty \bG^j\bepsilon_{t-j}$, where $\bG\in\RR^{N\times N}$
and $\bepsilon_t\in\RR^N$ independently follows multivariate normal distribution with $N(\zero, \Sigma_{\ve})$.
Define $\Gamma(h) = \cov(\wt\by_t,\wt\by_{t-h})$.
In addition, assume the independence between $\bX_t$ and $\wt\by_t$.
Let $\bM\in \RR^{N\times N}$ be an arbitrary square matrix.
Then the following conclusion holds.\\
(a) $(NT)^{-1}\sum_{t = 1}^T  \wt \by_t^\top \bM \wt\by_t \rightarrow_p \lim_{N\rightarrow \infty}N^{-1}\tr\{\bM\Gamma(0)\}$ if the limit exists
and
\begin{align}
T^{-1/2}N^{-1}\sum_{i = 0}^\infty\sum_{j = 0}^\infty
\big[\tr\{(\bG^\top)^i \bM \bG^j  \Sigma_{\ve}(\bG^\top)^j
\bM^\top \bG^i\Sigma_{\ve}\}\big]^{1/2}\rightarrow 0\label{yMy}
\end{align}
as $N\rightarrow \infty$.\\
(b) $(NT)^{-1}\sum_{t = 1}^T \bX_t^\top \bM\wt\by_t\rightarrow_p0$
if
\beq
T^{-1/2}N^{-1}\sum_{j = 0}^\infty\big[
\tr\{\bM\bG^j\Sigma_{\ve}(\bG^\top)^j\bM^\top\}\big]\rightarrow0\label{yMX}
\eeq
as $N\rightarrow\infty$.\\
(c) $(NT)^{-1}\bX_t^\top\bM \bX_t\rightarrow_p \sigma_X^2\lim_{N\rightarrow\infty}\tr(\bM)$ if the limit exists
and $N^{-2}T^{-1}\tr(\bM\bM^\top)\rightarrow0$ as $N\rightarrow0$.
\eel

\begin{proof}
  Let $\wt\bepsilon_t = \Sigma_{\ve}^{-1/2}\bepsilon_t$.
Then we have $\wt\by_t = \sum_{j = 0}^\infty \bG^j \Sigma_{\ve}^{1/2}\wt\bepsilon_t$.
Using the conclusion (d) of Lemma 1 in \cite{zhu2017network}, the result can be obtained.
\end{proof}

\bel\label{lem_GG_upper}
Let $\bG = \bU\bB_1\bU^\top +\beta_2\bI_N$
and $\Sigma_\ve = \sigma_{e}^2 \bI_N +
\bLambda\Sigma_f\bLambda^\top$.
Assume $\sigma_1(\Sigma_e)\le c_e$ and $\sigma_1(\Sigma_f) /N\le c_f$,
where $c_e$ and
$c_f$ are finite constants.
Assume $\bU^\top\bU = \bI_K$ and $\bLambda^\top\bLambda = I_r$.
Then it holds for any integer $n\ge 0$ that
\begin{align}
\bfeta^\top \bG^n\Sigma_{\ve}(\bG^\top)^n\bfeta \le c_1c_\eta N\big\{|\beta_2|+\sigma_1(\bB_1)\big\}^{2n}\label{etaG}
\end{align}
for any vector $\bfeta\in\RR^N$ satisfying
$\bfeta^\top \bfeta = c_\eta$,
 and $c_1$ is a finite positive constant.
In addition, we have
\begin{align}
& \tr\big\{\bG^{n}\Sigma_{\ve}(\bG^\top)^{n}\big\}\le c_2\{|\beta_2|+ \sigma_1(\bB_1)\}^{n_1+n_2}N,\label{G_tr}\\
&\tr\big\{\bG^{n_1}\Sigma_{\ve}(\bG^\top)^{n_1} \bG^{n_2}\Sigma_{\ve}(\bG^\top)^{n_2}\big\}\le c_3\{|\beta_2|+ \sigma_1(\bB_1)\}^{2n_1+2n_2}N^2\label{GG_tr}
\end{align}
where $c_2$ and $c_3$ are  finite constants.
\eel

\begin{proof}
Note that $\Sigma_\ve = \Sigma_e +\bLambda\Sigma_f\bLambda^\top$.
Hence it suffices to deal with the upper bounds
respectively for
$\bG^n\Sigma_e(\bG^\top)^n$
and
$\bG^n\bLambda\Sigma_f\bLambda^\top(\bG^\top)^n$.

Note that $\bG = \beta_2\bI_N + \bU\bB_1\bU^\top$.
Then we have
\begin{align}
&\bG^n = \sum_{j = 0}^n C_n^j (\bU\bB_1\bU^\top)^j\beta_2^{n-j}
= \sum_{j = 1}^nC_n^j
\beta_2^{n-j}(\bU\bB_1^j\bU^\top)+ n\beta_2^n\bI_N\label{Gn}\\
&\bG^n(\bG^\top)^n = \sum_{j_1,j_2 = 0, j_1+j_2\ne 0}^n
C_n^{j_1}C_{n}^{j_2}
\beta_2^{2n-j_1-j_2}\bU\bB_1^{j_1}(\bB_1^\top)^{j_2}\bU^\top
+ n^2\beta_2^{2n}\bI_N.\label{GG_upper1}
\end{align}

\noindent
{\bf Proof of (\ref{etaG}).}
It suffices to deal with the upper bounds
respectively for
$\bfeta^\top \bG^n\Sigma_e(\bG^\top)^n\bfeta$
and
$\bfeta^\top \bG^n\bLambda\Sigma_f\bLambda^\top(\bG^\top)^n\bfeta$.

{\sc (1) Upper Bound for $\bfeta^\top \bG^n\Sigma_e(\bG^\top)^n\bfeta$.}
First note $\bfeta^\top \bG^n\Sigma_e(\bG^\top)^n\bfeta\le
c_e \bfeta^\top \bG^n(\bG^\top)^n\bfeta$.
Then it suffices to derive upper bounds for
$\bfeta^\top \bG^n(\bG^\top)^n\bfeta$ directly.
By (\ref{GG_upper1}) we have
$\bfeta^\top \bG^n (\bG^\top)^n\bfeta =$
\begin{align}
&\sum_{j_1,j_2 = 0, j_1+j_2\ne 0}^n
C_n^{j_1}C_{n}^{j_2}
\beta_2^{2n-j_1-j_2}\bfeta^\top(\bU\bB_1^{j_1}(\bB_1^\top)^{j_2}\bU^\top)\bfeta
+ n\beta_2^n c_\eta\nonumber\\
&\le\sum_{j_1,j_2 = 0, j_1+j_2\ne 0}^n
C_n^{j_1}C_{n}^{j_2}
|\beta_2|^{2n-j_1-j_2}\sigma_1(\bB_1)^{j_1+j_2}c_\eta
+ n|\beta_2|^n c_\eta =
\{|\beta_2|+ \sigma_1(\bB_1)\}^{2n}c_\eta,
\label{GG_upper}
\end{align}
where the inequality is due to the Cauchy's inequality that
$\bfeta^\top\{\bU\bB_1^{j_1}(\bB_1^\top)^{j_2}\bU^\top\}\bfeta
\le \{\bfeta^\top\bU\bB_1^{j_1}\\(\bB_1^{j_1})^\top\bU^\top\bfeta\}^{1/2}
\{\bfeta^\top\bU\bB_1^{j_2}(\bB_1^{j_2})^\top\bU^\top\bfeta\}^{1/2}
\le \sigma_1(\bB_1)^{j_1+j_2}c_\eta$.

{\sc (2) Upper bound for $\bfeta^\top \bG^n\bLambda\Sigma_f\bLambda^\top(\bG^\top)^n\bfeta$.}
Note that by $\bLambda^\top\bLambda = \bI_r$ and (\ref{GG_upper})
we have
$\bfeta^\top \bG^n\bLambda\Sigma_f\bLambda^\top(\bG^\top)^n\bfeta\le
\sigma_1(\Sigma_f)\bfeta^\top \bG^n(\bG^\top)^n\bfeta\le
\sigma_1(\Sigma_f)\{|\beta_2|+ \sigma_1(\bB_1)\}^{2n}c_\eta\le
c_fc_\eta N\{|\beta_2|+ \sigma_1(\bB_1)\}^{2n}$.

\noindent
{\bf Proof of (\ref{G_tr}) and (\ref{GG_tr}).}
Due to the similarity of the proof, we only show the proof of (\ref{GG_tr})
in the following.
Recall that $\Sigma_{\ve} = \Sigma_{e}+ \bLambda \Sigma_f\bLambda^\top$.
Then it suffices to derive the upper bound for
$\tr\{\bG^{n_1}(\bG^\top)^{n_1} \bG^{n_2}(\bG^\top)^{n_2}\}$
and $\tr\{\bG^{n_1}\bLambda\Sigma_f\bLambda^\top(\bG^\top)^{n_1} \bG^{n_2}
\bLambda\Sigma_f\bLambda^\top(\bG^\top)^{n_2}\}$
respectively.
By (\ref{GG_upper1}),
we have
$\bG^{n_1}(\bG^\top)^{n_1}\bG^{n_2}(\bG^\top)^{n_2}
 = $
\begin{align*}
\sum_{i_1,i_2 = 0}^{n_1}\sum_{j_1,j_2 = 0}^{n_2}
 C_{n_1}^{i_1}C_{n_1}^{i_2}
 C_{n_2}^{j_1}C_{n_2}^{j_2}
 \beta_2^{2n_1+2n_2 - i_1-i_2 - j_1-j_2}
 \bU \bB_1^{i_1}(\bB_1^\top)^{i_2}\bB_1^{j_1}(\bB_1^\top)^{j_2}\bU^\top
 I(i_1+i_2+j_1+j_2\ne 0)
\end{align*}
$+ n_1^2n_2^2 \beta_2^{2n_1+2n_2}\bI_N$.

Note that $\tr\{\bU \bB_1^{i_1}(\bB_1^\top)^{i_2}\bB_1^{j_1}(\bB_1^\top)^{j_2}\bU^\top
\} = \tr\{\bB_1^{i_1}(\bB_1^\top)^{i_2}\bB_1^{j_1}(\bB_1^\top)^{j_2}\}$.
By Cauchy's inequality, we have
\begin{align*}
 \big|\tr\{\bB_1^{i_1}(\bB_1^\top)^{i_2}\bB_1^{j_1}(\bB_1^\top)^{j_2}\}\big|&\le
 \big[\tr\big\{\bB_1^{i_1}(\bB_1^\top)^{i_2}\bB_1^{i_2}(\bB_1^\top)^{i_1}\big\}\big]^{1/2}
 \big[\tr\big\{\bB_1^{j_1}(\bB_1^\top)^{j_2}\bB_1^{j_2}(\bB_1^\top)^{j_1}\big\}\big]^{1/2}\\
& \le \sigma_1(\bB_1)^{i_1+i_2+j_1+j_2}K.
\end{align*}
Hence we have
\beq
\tr\{\bG^{n_1}(\bG^\top)^{n_1} \bG^{n_2}(\bG^\top)^{n_2}\}\le
\{|\beta_2|+ \sigma_1(\bB_1)\}^{2n_1+2n_2}N.\label{G4_upper}
\eeq
Next, note that by Cauchy's inequality
$\tr\{\bG^{n_1}\bLambda\Sigma_f\bLambda^\top(\bG^\top)^{n_1} \bG^{n_2}
\bLambda\Sigma_f\bLambda^\top(\bG^\top)^{n_2}\}\le$
\begin{align*}
&\sigma_{1}(\Sigma_f)\tr\big\{
\bLambda^\top(\bG^\top)^{n_1}\bG^{n_2}\bLambda\Sigma_f\bLambda^\top
(\bG^\top)^{n_2}\bG^{n_1}\bLambda\big\}\\
& \le \sigma_1(\Sigma_f)^2\tr\{\bLambda^\top(\bG^\top)^{n_1}\bG^{n_2}
\bLambda\bLambda^\top(\bG^\top)^{n_2}\bG^{n_1}\bLambda\} \\
&\le \sigma_1(\Sigma_f)^2\tr\{\bLambda^\top(\bG^\top)^{n_1}\bG^{n_2}
\bLambda\bLambda^\top(\bG^\top)^{n_2}\bG^{n_1}\bLambda\} \\
& \le  \sigma_1(\Sigma_f)^2\tr\{\bLambda^\top(\bG^\top)^{n_1}\bG^{n_1}
\bLambda\}\tr\{\bLambda^\top(\bG^\top)^{n_2}\bG^{n_2}
\bLambda\}\\
&\le c\sigma_1(\Sigma_f)^2 M^2 \{|\beta_2|+ \sigma_1(\bB_1)\}^{2n_1+2n_2},
\end{align*}
where $c$ is a finite constant.
Then by (\ref{G4_upper}) the result can be obtained.

\end{proof}

\bel\label{quad_lem}
Let $\{v_i {\in\RR^{1}}:1\le i\le N\}$ be a set of identically distributed random variables.
Assume that (a) $E(v_i) = 0$ for $1\le i\le N$;
(b) $E(v_i,v_j) = 0$ for any $i\ne j$;
(c) $E(v_iv_jv_k) = 0$ for any $1\le i,j,k\le N$;
(d) $E(v_i^2) = 1$ and $E(v_i^4) = \kappa_4$,
where $\kappa_4$ is a finite positive constant.
Let $\bv = (v_1,v_2,\cdots, v_N)^\top\in\RR^N$, $Q_1 = \bv^\top \bM_1 \bv + \bu_1^\top \bv$,
and $Q_2 = \bv^\top \bM_2 \bv + \bu_2^\top \bv$, where $\bM_1 = (m_{1,ij})\in\RR^{N\times N}$ and
$\bM_2= (m_{2,ij})\in\RR^{N\times N}$ are $N\times N$ dimensional matrices, $\bu_1, \bu_2\in\RR^{N}$ are $N$-dimensional vectors.
We then have
\begin{align}
&\cov(Q_1, Q_2) = \tr(\bM_1\bM_2^\top) + \tr(\bM_1\bM_2) +
(\kappa_4 - 3)\tr\big\{\diag(\bM_1)\diag(\bM_2)\big\}+
\bu_1^\top \bu_2,\label{cov_qq}\\
&\var(Q_1)\le c \tr(\bM_1^\top \bM_1) + \bu_1^\top \bu_1,\label{var_Q1}
\end{align}
where $c$ is a finite constant.
\eel

\begin{proof}
The proof of (\ref{cov_qq}) is given in Lemma 2 of \cite{zhu2019multivariate}.
Next, (\ref{var_Q1}) can be shown by using Cauchy's
inequality, i.e., $|\tr(\bM_1^2)|\le \tr(\bM_1^\top \bM_1)$.
\end{proof}

\bel\label{mu_G}
Assume $\bB_1$ is diagonalizable, i.e.,
$\bB_1 = \bP_B\bD_B \bP_B^{-1}$,
where $\bD_B$ is a diagonal matrix.
Let $\bG = \bU\bB_1\bU^\top$ as defined in Section 2.3.
Then
\begin{align*}
\mu_{\min}(\bG)\ge \big\{1-\sigma_1(\bB_1) - |\beta_2|\big\}^2\sigma_K(\bP_B)^{-2}
\sigma_K(\bP_B^{-1})^{-2},
\end{align*}
where
$\mu_{\min}(\bG) = \min_{| z| = 1, \bz\in \mC}
\{(\bI_N - \bG z)^*(\bI_N - \bG z)\}$.
\eel

\begin{proof}
Let $\bU_c\in\RR^{N\times (N-K)}$ collect the complement orthogonal
vector of $\bU$.
Further denote $\wt\bU = (\bU, \bU_c)\in\RR^{N\times N}$,
thus we have $\wt\bU^\top \wt\bU = \bI_N$.
Further denote $\wt\bP_B = \diag(\bP_B, \bI_N)$, $\wt \bP_B^{-} = \diag(\bP_B^{-1}, \bI_N)$,
and $\wt\bB_1 = \diag(\bB_1, \zero)$.
Therefore we have
$\bU\bB_1\bU^\top = \wt\bU\wt\bP_B\Lambda(\wt\bB_1) \bP_B^{-}\wt\bU^\top$,
where $\Lambda(\wt\bB_1) = \diag\{\lambda_1(\wt\bB_1),\cdots,
\lambda_N(\wt\bB_1)\}$ is a diagonal matrix.
This leads to
$\bG = \wt\bU\wt\bP_B\Lambda(\wt\bB_1) \bP_B^{-}\wt\bU^\top + \beta_2 \bI_N$.
First by definition we have we have
\begin{align*}
\mu_{\min}(\bG) = \min_{|z| = 1}
\{(\bI_N - \bG z)^*(\bI_N - \bG z)\} =
\min_{|\bz| = 1}
\{(z\bI_N - \bG)^*(z\bI_N - \bG)\} .
\end{align*}
Further by the decomposition of $\bG$
we have $z\bI_N - \bG = \wt\bU\wt\bP_B\bD_z\bP_B^{-}\wt\bU^\top$,
where $\bD_z$ is diagonal with entries
$z - \beta_2 - \lambda_i(\wt\bB_1)$.
By the condition we have $|\beta_2| + |\lambda_i(\wt\bB_1)|\le
|\beta_2| + |\sigma_i(\bB_1)|<1$.
Therefore $\bD_z$ is invertible for all $|z| = 1$
and the eigenvalues of $\bD_z^*\bD_z$
are $|z - \beta_2 - \lambda_i(\wt\bB_1)|^2\ge \{1-|\beta_2| -
\sigma_1(\bB_1)\}^2$ for all $|z| = 1$ and $1\le i\le N$.
Hence we have
\begin{align*}
\mu_{\min}(\bG) = \min_{|z| = 1}\big\{
\sigma(\wt\bU\wt\bP_B\bD_z\wt\bP_B^{-}
(\wt\bP_B^{-})^\top \bD_z^* \wt\bP_B^\top\wt\bU^\top)\big\}\ge
\sigma_K(\bP_B)^{2}\sigma_K(\bP_B^{-1})^2\{1-
|\beta_2| - \sigma_1(\bB_1)\}^2.
\end{align*}

\end{proof}

\begin{lemma} \label{thm:SBM-eigvec-L2-bound}
Let $\bA$ be an adjacency matrix generated from a stochastic block model $(\bTheta, \bB)$.
Assume that $\bA^\star = \bTheta \bB \bTheta^\top$ is of rank $K$, with smallest absolute nonzero eigenvalue at least $\gamma_N$ and $\underset{k,l}{\max} \, b_{kl} \le \alpha_N$ for some $\alpha_N \le \log N/N$.
Let $\hat\bU, \bU^\star \in \RR^{N\times K}$ be the $K$ leading eigen-vectors of $\bA$ and $\bA^\star$, respectively.
Then, there exist a $K\times K$ orthonormal matrix $\bH$ and a constant $C$ such that
\begin{equation} \label{eqn:SBM-eigvec-L2-bound}
    \norm{\hat\bU - \bU^\star\bH}_F \le \frac{2 \sqrt{2K}}{\gamma_N} C \sqrt{N \alpha_N}
\end{equation}
with probability at least $1 - N^{-1}$.
\end{lemma}
\begin{proof}
Combining Lemma 5.1 and Theorem 5.2 of \cite{lei2015consistency}, we obtain that, for some $K\times K$ orthonormal matrix $\bH$,
\[
    \norm{\hat\bU - \bU^\star\bH}_F \le \frac{2 \sqrt{2K}}{\gamma_N} \norm{\bA - \bA^\star} \le \frac{2 \sqrt{2K}}{\gamma_N} C \sqrt{N \alpha_N}
\]
with probability at least $1 - N^{-1}$.
The constant $C$ is the absolute constant invovled in Theorem 5.2.
The maximum degree in Theorem 5.2 becomes $N \alpha_N$ in the current setting.
\end{proof}

\bel\label{vSigx}
Assume the same conditions as in Theorem  \ref{thm_first_step_K_diverge}.
{Define $\alpha_{uN} = 1$ if
$ \bU^\top\bLambda = \zero$ and $\alpha_{uN} = N$ otherwise.}
For any vector $\bv\in \RR^K$, $\bw, \bw_1\in\RR^{K^2}$,
$\bw_2\in\RR^{p}$
 with $\max\{\|\bv\|, \|\bw\|, \|\bw_1\|, \|\bw_2\|\} \le  1$ and $0<\xi<1$, it holds
\begin{align}
&P\Big\{
\frac{1}{ \alpha_{uN}T}\big|\sum_t\bv^\top\bU^\top \by_{t-1}\by_{t-1}^\top\bU\bv - T
\bv^\top\bU^\top\Gamma_y(0)\bU\bv
\big|\ge \xi
\Big\}\le 2\exp(-c\min\{T\xi, T^2\xi^2\})\label{uyyu_conv}\\
& P\Big\{
\frac{1}{ \alpha_{uN}T}\big|\sum_t
\bw^\top \big\{(\bU^\top\by_{t-1})\otimes (\bU^\top\by_{t-1})
-\vec(\bU^\top \Gamma_y(0)\bU)\big\}
\big|\ge \xi
\Big\}\le 2\exp(2\log K -  c\min\{T\xi, T^2\xi^2\})\label{uy_conv}\\
& P\Big\{
\frac{1}{\alpha_{uN}T}\big|\sum_t
\bw_1^\top \{(\bU^\top\by_{t-1})\otimes (\bU^\top\bZ_{t-1})\}\bw_2
\big|\ge \xi
\Big\}\le 2\exp(2\log K -  c\min\{T\xi, T^2\xi^2\})\label{uyz}\\
& P\Big\{
\frac{1}{TN}\big|\sum_t
\by_{t-1}^\top\by_{t-1} - T\tr(\Gamma_y(0))
\big|\ge \xi
\Big\}\le 2\exp( - c\min\{T\xi, T^2\xi^2\})\label{yy}\\
& P\Big\{
\frac{1}{TN}\big|\sum_t
\bw_2^\top\bZ_{t-1}^\top\bZ_{t-1}\bw_2 - NT\bw_2^\top\Sigma_Z\bw_2
\big|\ge \xi
\Big\}\le 2\exp( - cN\min\{T\xi, T^2\xi^2\}).\label{zz}
\end{align}
In addition, let $\wt\bepsilon_t = \bepsilon_t + \bZ_{t-1}\bgamma$,
$\wt\by_t = \sum_{j = 0}^\infty \bG^j \wt\bepsilon_{t-j}$,
$\wt\mu_y = T^{-1}\sum_t\wt\by_t$.
Then we have
\begin{align}
& P\Big\{{\frac{1}{\alpha_{uN}}}\big|\bv^\top\bU^\top \wt\mu_y\wt\mu_y^\top\bU\bv\big|\ge \xi
\Big\}\le  2\exp\{-c\min\{T\xi, T^2\xi^2\}\},\label{muy1}\\
& P\Big\{N^{-1}\big|\wt\mu_y^\top\wt\mu_y\big|\ge \xi
\Big\}\le  2\exp\{-c\min\{T\xi, T^2\xi^2\}\}.\label{muy2}
\end{align}
{If $\bU^\top\bLambda \ne 0$ and $\wh\bU^\top\bLambda \ne 0$, then define $\bX_{t-1} = ((\by_{t-1}^\top \bU)\otimes \bU, \by_{t-1}, \bZ_{t-1})\in\RR^{N\times (K^2+p+1)}$
and $\wh\bX_{t-1} = ((\by_{t-1}^\top \wh\bU)\otimes \wh\bU, \by_{t-1}, \bZ_{t-1})\in\RR^{N\times (K^2+p+1)}$.
Otherwise define $\bX_{t-1}$ by replacing $\bU$ (or $\wh \bU$) with $\bU_N$ (or $\wh \bU_N$).
Let $\wt \Sigma_x \defeq (NT)^{-1}\sum_t\bX_{t-1}^\top\bX_{t-1}$
and $\wh \Sigma_x \defeq (NT)^{-1}\sum_t\wh\bX_{t-1}^\top\wh\bX_{t-1}$.
Accordingly, define $\Sigma_x$ as in (\ref{Sig_x_2e}).}
It holds that for any $\bv\in\RR^{K^2+p+1}$
\begin{align}
&P\big\{\big|\bv^\top(\wt\Sigma_x - \Sigma_x)\bv
\big|\ge \xi\big\}\le 2\exp(- c\min\{T\xi, T^2\xi^2\}).\label{del_Sigx}
\end{align}
Furthermore by assuming $\gamma_N\gg  \sqrt{K\log N}/\xi$ we have
\begin{align}
&P\big\{\big|\bv^\top(\wh\Sigma_x - \Sigma_x)\bv
\big|\ge \xi\big\}\le 2\exp(-c\min\{T\xi, T^2\xi^2\}) + c\frac{\sqrt K\log N}{\sqrt T\xi\gamma_N}+ N^{-1}.\label{del_hat_Sigx}\\
&P\Big\{{\frac{1}{\alpha_{uN}}}\big|\bv^\top\wh\bU^\top \wt\mu_y\wt\mu_y^\top\wh\bU\bv\big|\ge \xi
\Big\}\le  2\exp\{-c\min\{T\xi, T^2\xi^2\}\} + c\frac{\sqrt K\log N}{\sqrt T\xi\gamma_N}+N^{-1}\label{uyyuhat_conv}
\end{align}
\eel

\begin{proof}
  We prove the results one by one in the following.
  {Recall that
$\by_t = \sum_{n = 0}^\infty\bG^n \bepsilon_{t-n}$
and $\bG^n$ takes the form in (\ref{Gn}).
Hence if $\bU^\top \bLambda = \zero$, we could drop $\bLambda\bff_t$ in $\wt \bepsilon_t$ and the proof procedure is the same.
In this case we could verify that $\bfeta^\top\bU^\top\Gamma_y(0)\bU\bfeta = O(1)$.
Assume $\bU^\top\bLambda\ne \zero$ (thus $\alpha_{uN} = N$) and
$\wh\bU^\top \bLambda\ne \zero$ for convenience in the following.
This leads to $N^{-1}\bfeta^\top\bU^\top\Gamma_y(0)\bU\bfeta \ge cN^{-1}\sigma_M(\Sigma_f)$ for any $\|\bfeta\| = 1$.

}

\noindent
{\bf Proof of (\ref{uyyu_conv}):}
Note that
$T^{-1}\bv ^\top \sum_t\bU^\top \by_{t-1}\by_{t-1}^\top \bU \bv
= T^{-1}\sum_t \by_{t-1}^\top \bU \bv \bv^\top \bU^\top \by_{t-1}$.
Let $\by = (\by_0^\top ,\cdots, \by_{T-1}^\top)^\top \in\RR^{NT}$
and denote $\Gamma = \cov(\by)$ and
$\bz = \Gamma^{-1/2}\by$.
Hence we have
$T^{-1}\sum_t\by_{t-1}^\top \bU \bv \bv^\top \bU^\top \by_{t-1}
= T^{-1}\bz^\top (\Gamma^{1/2})^\top \bD_U(\Gamma^{1/2})\bz$,
where $\bD_U = \bI_T\otimes (\bU \bv \bv^\top \bU^\top)\in\RR^{(NT)\times (NT)}$.
By the Hanson-Wright inequality of \cite{rudelson2013hanson}, we have
\begin{align}
P\Big\{
\big|\bz^\top \bQ\bz - E(\bz^\top \bQ\bz) \big|>{NT}\xi
\Big\}\le 2\exp\Big\{-c\min\Big(
\frac{N^2T^2\xi^2}{\|\bQ\|_F^2}, \frac{NT\xi}{\sigma_1(\bQ)}\Big) \Big\},\label{HW_ieq}
\end{align}
where $\bQ = (\Gamma^{1/2})^\top \bD_U(\Gamma^{1/2})$.
Note that $\|\bQ\|_F^2 = \tr(\bD_U\Gamma \bD_U\Gamma)\le
\sigma_1(\Gamma)^2 \tr(\bD_U^2)=T\sigma_1(\Gamma)^2$.
For any $\bw\in \RR^{NT}$ satisfying $\|\bw\| = 1$
we have
\begin{align*}
\bw^\top\bQ\bw &\le
\rho(\bD_U)(\bw^\top \Gamma\bw)\le
\rho(\bU\bv\bv^\top \bU^\top)(\bw^\top \Gamma\bw)\le \bw^\top \Gamma\bw\le \sigma_1(\Gamma)
\end{align*}
Recall that $\bG = \bU\bB_1\bU^\top + \beta_2\bI_N$.
According to \cite{basu2015regularized} (Proposition 2.3 and (2.6))
and Lemma \ref{mu_G}
we have
\begin{align}
\sigma_1(\Gamma) \le \frac{\sigma_1(\Sigma_{\ve})}{\mu_{\min}(\bG)}
\le \frac{\sigma_1(\Sigma_\ve)\sigma_K(\bP_B)^2\sigma_K(\bP_B^{-1})^2}
{\{1-|\beta_2| - \sigma_1(\bB_1)\}^2},\label{sigma_Gamma}
\end{align}
where $\mu_{\min}(\bG) = \min_{| z| = 1, \bz\in \mC}
\{(\bI_N - \bG z)^*(\bI_N - \bG z)\}$,
and $\bB_1 = \bP_B\bD_B \bP_B^{-1}$ is assumed to be diagonalizable.
Furthermore, note that
$\Sigma_{\ve} = \bLambda\Sigma_f\bLambda^\top + \Sigma_{e}$.
Therefore we have $\sigma_1(\Sigma_\ve)\le \sigma_1(\Sigma_f) + \sigma_1(\Sigma_{e}) = O(N)$.
Hence by the conditions in Theorem  \ref{thm_first_step_K_diverge}
and (\ref{sigma_Gamma}), we have $\sigma_1(\Gamma) = O(N)$,
which leads to the final result.

\noindent
{\bf Proof of (\ref{uy_conv}):}
Note that we have $(\bU^\top \by_{t-1})\otimes (\bU^\top\by_{t-1}) =
\vec(\bU^\top\by_{t-1}\by_{t-1}^\top\bU)$.
Let $\bw = (w_1,\cdots, w_{K^2})^\top$.
It can be derived that
\begin{align*}
&\big|\sum_t
\bw^\top\vec\big(\bU^\top\by_{t-1}\by_{t-1}^\top\bU\big)-
T\bw^\top\vec\big(\bU^\top\Gamma_y(0)\bU\big)
\big|\\
&\le \max_k |w_k|\max_{k_1k_2}\big|\sum_t\be_{k_1}^\top \bU^\top \by_{t-1}\by_{t-1}^\top \bU\be_{k_2} - T\be_{k_1}^\top\bU^\top \Gamma_y(0)\bU\be_{k_2}
\big|
\end{align*}
Using (\ref{HW_ieq}) and follow the same technique in proving (\ref{uyyu_conv}), we can show for any $1\le k_1,k_2\le K$
\begin{align}
 P\Big\{
\frac{1}{T N}\big|\sum_t
\be_{k_1}^\top(\bU^\top\by_{t-1}\by_{t-1}^\top\bU)\be_{k_2}
\big|\ge \xi
\Big\}\le 2\exp(-c\min\{T\xi, T^2\xi^2\}).
\end{align}
Then by the maximum inequality property, (\ref{uy_conv}) can be proved.

\noindent
{\bf Proof of (\ref{uyz})--(\ref{zz}):}
The proofs of (\ref{uyz}) and (\ref{yy}) are similar to
(\ref{uy_conv}).
{Note that for (\ref{uyz}) $E[\bw_1^\top \{(\bU^\top\by_{t-1})\otimes (\bU^\top\bZ_{t-1})\}\bw_2]
= \bw_1^\top\vec\{\bU^\top E(\by_t(\bZ_t\bw_2)^\top) \bU\}
=
\bw_1^\top\vec\{\bU^\top  \bU\}(\bgamma^\top\Sigma_Z\bw_2) = O(1)$.
Hence $N^{-1}E[\bw_1^\top \{(\bU^\top\by_{t-1})\otimes (\bU^\top\bZ_{t-1})\}\bw_2]\rightarrow 0$ and omitted here.}

Lastly, (\ref{zz}) is simple to prove by directly using (\ref{HW_ieq})
and the fact that $p$ is finite.

\noindent
{\bf Proof of (\ref{del_Sigx}):}
The result of (\ref{del_Sigx}) is implied by
(\ref{uyyu_conv})--(\ref{zz})
by noting $\wt\Sigma_x$ takes the form in (\ref{wt_Sigx}).

\noindent
{\bf Proof of (\ref{del_hat_Sigx}) and (\ref{uyyuhat_conv}):}
The result of (\ref{del_hat_Sigx}) can be obtained by replacing
$\bU$ by $\wh\bU$ in (\ref{uyyu_conv})--(\ref{uyz}).
Due to the similarity we only show the first one.
Note that we have we have decomposition as follows
\begin{align*}
&\frac{1}{NT}\big|\sum_t\bv^\top\wh\bU^\top \by_{t-1}\by_{t-1}^\top\wh\bU\bv - T
\bv^\top\bU^\top\Gamma_y(0)\bU\bv
\big|\\
& \le \frac{1}{NT}\big|\sum_t\bv^\top\wh\bU^\top \by_{t-1}\by_{t-1}^\top\wh\bU\bv - T
\bv^\top\wh\bU^\top\Gamma_y(0)\wh\bU\bv
\big| +
\frac{1}{N}\big|\bv^\top\wh\bU^\top\Gamma_y(0)\wh\bU\bv -
\bv^\top \bU^\top\Gamma_y(0) \bU\bv\big|\\
&\defeq \Delta_1 + \Delta_2
\end{align*}
We then bound the above two parts respectively.

{

\noindent
{\sc (a) Bound on $\Delta_2$.}

Note that
\begin{align}
&\big|\bv^\top\wh\bU^\top\Gamma_y(0)\wh\bU\bv -
\bv^\top \bU^\top\Gamma_y(0) \bU\bv\big|\nonumber \\
&\le
\big|\bv^\top\wh\bU^\top\Gamma_y(0)\wh\bU\bv -
\bv^\top \bU^\top\Gamma_y(0) \wh\bU\bv\big|+
\big|\bv^\top \bU^\top\Gamma_y(0)\wh\bU\bv -
\bv^\top \bU^\top\Gamma_y(0) \bU\bv\big|\nonumber\\
& \le \big\{
\bv^\top (\wh\bU - \bU)^\top(\wh\bU - \bU)\bv
\big\}^{1/2}
\big\{\bv^\top\wh\bU^\top \Gamma_y(0)^2\wh\bU \bv
\big\}^{1/2}+
\big\{
\bv^\top (\wh\bU - \bU)^\top(\wh\bU - \bU)\bv
\big\}^{1/2}
\big\{\bv^\top \bU^\top \Gamma_y(0)^2 \bU \bv
\big\}^{1/2}\nonumber\\
&\le 2\sigma_1\{\Gamma_y(0)\}
\big\{
\bv^\top (\wh\bU - \bU)^\top(\wh\bU - \bU)\bv
\big\}^{1/2}\le
2\sigma_1\{\Gamma_y(0)\}\|\wh \bU - \bU\|_F\label{del_UU}
\end{align}
Hence we have
\begin{align*}
P\Big\{\frac{1}{N}
\big|\bv^\top\wh\bU^\top\Gamma_y(0)\wh\bU\bv -
\bv^\top \bU^\top\Gamma_y(0) \bU\bv\big| \ge \xi
\Big\}\le P\Big\{
\big\|\wh\bU - \bU\big\|_F\ge \frac{\xi  N}{2\sigma_1\{\Gamma_y(0)\}}
\Big\}
\end{align*}
By letting $N\xi/\{2\sigma_1(\Gamma_y(0))\} \ge 2C\sqrt{2K\log N}/\gamma_N$
and noting $\gamma_N\gg \sqrt{K\log N}/\xi$
and $\sigma_1(\Gamma_y(0))/N\le \sigma_1(\Gamma)/N < \infty$ by (\ref{sigma_Gamma}), the conclusion can be obtained
by using Lemma \ref{thm:SBM-eigvec-L2-bound}.

\noindent
{\sc (b) Bound on $\Delta_1$.}

Let $\Delta_y = (NT)^{-1}\sum_t \by_{t-1}\by_{t-1}^\top - N^{-1}\Gamma_y(0)$.
Note that we have
$\bv^\top\wh\bU^\top \Delta_y\wh\bU\bv = \bv^\top\wh\bU^\top \Delta_y\wh\bU\bv -
\bv^\top \bU^\top \Delta_y \bU\bv + \bv^\top \bU^\top \Delta_y \bU\bv$.
Then we have $P\{|\bv^\top \bU^\top \Delta_y \bU\bv|\ge \xi\}\le 2\exp(-c\min(T\xi^2, T^2\xi^2))$ by (\ref{del_Sigx}).
By using the technique in deriving (\ref{del_UU}), we have
$|\bv^\top\wh\bU^\top \Delta_y\wh\bU\bv -
\bv^\top \bU^\top \Delta_y \bU\bv|\le 2\tr\{\Delta_y^2\}^{1/2}\|\wh \bU - \bU\|_F$.
It can be verified that $E\tr(\Delta_y^2)^{1/2}\le \{E\tr(\Delta_y^2)\}^{1/2} = O(1/\sqrt T)$.
Consider the events $\mE = \{\tr(\Delta_y^2)^{1/2}>a/\sqrt T\}$.
Then by Chebyshev's Inequality,
\begin{align*}
P\Big(\big|\bv^\top\wh\bU^\top \Delta_y\wh\bU\bv -
\bv^\top \bU^\top \Delta_y \bU\bv\big| > \xi\Big)&\le
P(\mE) + P\big\{\|\wh\bU - \bU\|_F\ge \sqrt T\xi/(2a)\big\}
\le c/a + N^{-1}
\end{align*}
as long as $\sqrt T\xi/a\gg \sqrt{K\log N}/\gamma_N$.
Letting $a = \sqrt T\xi\gamma_N/(\sqrt K\log N)$, the result can be obtained.


Lastly, (\ref{uyyuhat_conv}) can be obtained similarly by replacing $\bU$ with $\wh \bU$
in (\ref{muy1}).

%

}

By using the same technique in (\ref{uyyu_conv}), we can firstly show that
\begin{align*}
P\Big\{
\frac{1}{NT}\big|\sum_t\bv^\top\wh\bU^\top \by_{t-1}\by_{t-1}^\top\wh\bU\bv - T
\bv^\top\wh\bU^\top\Gamma_y(0)\wh\bU\bv
\big|\ge \xi
\Big\}\le 2\exp(-c\min\{T\xi, T^2\xi^2\}).
\end{align*}

\noindent
{\bf Proof of (\ref{muy1}) and (\ref{muy2}):}
Let $\wt\by = (\wt\by_1^\top,\cdots, \wt\by_T^\top)^\top\in\RR^{NT}$
and let $\Gamma = \cov(\wt\by)$,
$\bz = \Gamma^{-1/2}\wt\by$.
Then we have $\wt\mu_y = T^{-1}(\one_T^\top\otimes \bI_N)\wt\by$.
Hence we have
$\bv^\top\bU^\top \wt\mu_y\wt\mu_y^\top \bU\bv =
T^{-2} \bz^\top (\Gamma^{1/2})^\top (\one_T\one_T^\top)\otimes (\bU\bv\bv^\top\bU^\top) \Gamma^{1/2}\bz\defeq T^{-1}\bz^\top\bQ \bz$
for any $\bv\in \RR^N$.
It can be easily verified that
$\|\bQ\|_F^2 \le \sigma_1(\Gamma)^2$
and $\sigma_1(\bQ)\le \sigma_1(\Gamma)$.
Similar to (\ref{uyyu_conv}),
by using the Hanson-Wright inequality given in \cite{rudelson2013hanson}, it can concluded that
\begin{align*}
P\Big\{{\frac{1}{N}}\big|\bv^\top\bU^\top \wt\mu_y\wt\mu_y^\top\bU\bv\big|\ge \xi
\Big\}\le  2\exp\{-c\min(T^2\xi^2,T\xi)\},
\end{align*}
which proves (\ref{muy1}).

Similarly, note that $N^{-1}\wt\mu_y^\top\wt\mu_y = N^{-1} T^{-2}\wt\by^\top\{(\one_T\one_T^\top)\otimes \bI_N\}\wt\by=
N^{-1}T^{-2}\bz^\top (\Gamma^{1/2})^\top
\{(\one_T\one_T^\top)\otimes \bI_N\}\Gamma^{1/2} \bz\defeq T^{-1}N^{ -1}\bz^\top \bQ\bz$.
Then it can be verified that $\|\bQ\|_F^2\le \sigma_1(\Gamma)^2$
and $\sigma_1(\bQ)\le \sigma_1(\Gamma)$.
Therefore by the Hanson-Wright inequality (\ref{muy2}) can be proved.

\end{proof}

\bel\label{sup_vSigx}
Assume the same conditions as in Theorem  \ref{thm_first_step_K_diverge}.
Let $\bX_t$, $\wh\bX_t$,
$\wt \Sigma_x $
and $\wh \Sigma_x $ be defined as in Lemma \ref{vSigx}.
In addition, let $\wt \mu_y$ be defined as in Lemma \ref{vSigx}.
We then have
\begin{align}
& P\Big\{\sup_{\|\bv\|\le 1}\big|\bv^\top(\wt\Sigma_x - \Sigma_x)\bv
\big|\ge \xi\Big\}\le 2\exp\big\{K^2\log 21 - c\min\{T\xi, T^2\xi^2\}\big\},\label{sup_delx1}\\
& P\Big\{\sup_{\|\bv\|\le 1}\big|\bv^\top(\wh\Sigma_x - \Sigma_x)\bv
\big|\ge \xi\Big\}\le 2\exp\big\{K^2\log 21 - c_2\min\{T\xi, T^2\xi^2\}\big\}\nonumber\\
&~~~~~~~~~~~\quad \quad\quad\quad\quad\quad
+ N^{-1}\exp(K^2\log 21)+c\frac{\sqrt K\log N\exp(K^2\log 21)}{\sqrt T\xi\gamma_N},\label{sup_delx2}\\
& P\Big\{\sup_{\|\bv\|\le 1}\frac{1}{\alpha_{uN}}\big|\bv^\top(\bU^\top \wt\mu_y\wt\mu_y^\top \bU)\bv
\big|\ge \xi\Big\}\le 2\exp\big\{K^2\log 21 - c\min\{T\xi, T^2\xi^2\}\big\}.\label{muy_del}
\end{align}
Let $\omega_{NTK} = N^{-1/2}+ (\log N/T)^{1/2} + \kappa_{NT}K$.
Further let
$\wt\Sigma_{2\theta} = (NT)^{-1}\sum_t\bX_{t-1}^\top \wh\Omega_\ve \bX_{t-1}$
$\wh\Sigma_{2\theta} = (NT)^{-1}\sum_t\wh\bX_{t-1}^\top \wh\Omega_\ve \wh\bX_{t-1}$.
Then we have
\begin{align}
& \sup_{\|\bv\|\le 1} \big|
\bv^\top (\wt\Sigma_{2\theta} - \Sigma_{2\theta})\bv
\big| = o_p(1),\label{del_Sig2theta}\\
& \sup_{\|\bv\|\le 1} \big|
\bv^\top (\wh\Sigma_{2\theta} - \Sigma_{2\theta})\bv
\big| = o_p(1).\label{del_hatSig2theta}
\end{align}
\eel

\begin{proof}
\noindent
{\bf Proof of (\ref{sup_delx1}):}
We consider the set $\mV = \{v\in\RR^{K^2+p+1}: \|v\|\le 1\}$.
Choose $\mV^* = \{\bv_1,\cdots, \bv_m\}$ as a $1/10$-net of $\mV$,
where the definition of covering net is given in \cite{vershynin2011lectures}.
By Lemma 3.5 of \cite{vershynin2011lectures},
we have $|\mV^*|\le 21^{K^2+p+1}$.
It holds that for every $\bv\in\mV$, there exists some $\bv_i\in\mV^*$
such that $\|\Delta \bv\|\le 1/10$, where $\Delta \bv = \bv - \bv_i$.
Define $\Delta_x = \wt\Sigma_x - \Sigma_x$.
Then we have
\begin{align*}
\gamma\defeq \sup_{\bv\in \mV} |\bv^\top \Delta_x\bv|
\le \max_i |\bv_i^\top \Delta_x \bv_i| +
2\sup_{\bv\in\mV}\max_j |\bv_j^\top\Delta_x\Delta \bv|
+ \sup_{\bv\in \mV}|\Delta \bv^\top \Delta_x \Delta \bv|.
\end{align*}
Note that $10\Delta\bv \in \mV$.
Therefore the third term is bounded by $\gamma /100$.
Next by Cauchy's inequality we have
$2\sup_{\bv}\max_j |\bv_i^\top \Delta_x \Delta\bv|
\le 2/10 \max_i(\bv_i^\top \Delta_x\bv_i)^{1/2}
\sup_{\bv\in\mV}\{(10\Delta \bv)^\top \Delta_x (10\Delta \bv)\}^{1/2}\le 2/10\gamma$.
Hence one could derive that
$\gamma\le 2\max_i\bv_i^\top\Delta_x \bv_i$.
Combining the result of Lemma \ref{vSigx},
it implies,
\begin{align*}
P\Big\{\sup_{\bv\in \mV}\big|\bv^\top(\wt\Sigma_x - \Sigma_x)\bv
\big|\ge \xi\Big\}\le 2\exp\big\{K^2\log 21 - cT\xi^2\big\}.
\end{align*}
This completes the proof.

\noindent
{\bf Proof of (\ref{sup_delx2}):}
The proof is the same as (\ref{sup_delx1}) by using (\ref{del_hat_Sigx}) of Lemma \ref{vSigx}.

\noindent
{\bf Proof of (\ref{muy_del}):}
The proof is the same as (\ref{sup_delx1}) by using (\ref{muy1}).

\noindent
{\bf Proof of (\ref{del_Sig2theta}) and (\ref{del_hatSig2theta}):}
Note that
\beq
\wt\Sigma_{2\theta} - \Sigma_{2\theta} = (NT)^{-1}\sum_t\bX_t^\top (\wh\Omega_\ve - \Omega_\ve)\bX_t + \Big\{
(NT)^{-1}\sum_t \bX_t^\top\Omega_\ve \bX_t - \Sigma_{2\theta}\Big\}\nonumber
\eeq
can be decomposed into two parts.
Therefore we deal with each part separately.
First
note that
$\max_i|\lambda_i(\wh\Omega_\ve - \Omega_\ve)|=
\Op{\omega_{NTK}}$ by Theorem \ref{thm:resid-cov-bound}.
Hence we have
\beq
\sup_{\|\bv\|\le 1}
\Big|\balpha_{NT}^{-1}\sum_t\bv^\top \bX_t^\top (\wh\Omega_\ve - \Omega_\ve)\bX_t\bv\Big|
\le \max_i\big|\lambda_i(\wh\Omega_\ve - \Omega_\ve)\big|
\sup_{\|\bv\|\le 1}\big|({NT})^{-1}\sum_t\bv^\top \bX_t^\top\bX_t\bv\big| = O_p(\omega_{NTK})\nonumber
\eeq
due to $({NT})^{-1}\sup_{\|\bv\|\le 1}\sum_t\bv^\top \bX_t^\top\bX_t\bv = O_p(1)$ by (\ref{sup_delx1}).

Next, by the same argument to (\ref{sup_delx1}), we could establish that
\beq
P\Big\{\sup_{\|\bv\|\le 1}\big|\bv^\top\big((NT)^{-1}\sum_t \bX_t^\top\Omega_\ve \bX_t - \Sigma_{2\theta}\big)\bv
\big|\ge \xi\Big\}\le 2\exp\big\{K^2\log 21 - c\min\{T\xi, T^2\xi^2\}\big\}\nonumber
\eeq
where the details are omitted here.
We have
$\sup_{\|\bv\|\le 1}\big|\bv^\top\big((NT)^{-1}\sum_t \bX_t^\top\Omega_\ve \bX_t - \Sigma_{2\theta}\big)\bv
\big| = o_p(1)$. This leads to (\ref{del_Sig2theta}).
Lastly, (\ref{del_hatSig2theta}) can be obtained by
following the same technique by using (\ref{sup_delx2}).

\end{proof}

\section{Property of the covariance estimator of \texorpdfstring{$\hat\beps_t$}{first regression residuals} }  \label{sec:poet-of-hat-eps}

In this section, we establish the convergence rate of the covariance estimator from $\hat\beps_t$ -- the residual from the first-step regression residual of the CNAR model.
Recall that the oracle estimators $\tilde\beps_t$ and $\tilde\btheta^{(1)}$ are the first-step estimators when the rotated community membership matrix $\bU$ is known and the feasible estimator $\hat\beps_t$ and $\hat\btheta^{(1)}$ are those when $\hat\bU$ is estimated from the adjacency matrix $\bA$.
Here we only show results for the oracle estimators $\tilde\beps_t$ and $\tilde\btheta^{(1)}$.
Since $\bA$ is assume to be uncorrelated with $\bY_t$ and $\bZ_t$, so is $\hat\bU$.
Lemma \ref{thm:SBM-eigvec-L2-bound} establishes that $\hat\bU$ is a consistent estimator of $\bU$.
Thus, the proof of the non-oracle case where a consistent estimator $\hat\bU$ is plugged in is a trivial extension whose technique is the same as that in Theorem \ref{thm:first_step_estimator} (ii).

We make use of some results in \cite{fan2013large}.
However, there are several differences:
first, even with oracle $\tilde\beps_t$, the task here is slightly different from that in \cite{fan2013large} since direct observations for $\braces{\beps_t}_{t=1}^T$ is not available.
Instead, we observe $\braces{\tilde\beps_t}_{t=1}^T$ -- the residual from the first step regression of the CNAR model.
That is, $\tilde\beps_t=\by_t - \bX_{t-1}\tilde\btheta^{(1)}$;
second, the dimension of $\btheta^{(1)}$ is $K^2 + p + 1$, where the number of communities $K$ is allow to grow with the number of network nodes $N$.
Thus, $K$ shows up in our results;
third, we are considering a time series here, thus $\bX_{t-1}$ is correlated with $\braces{\beps_{s-1}}_{s\le t}$.
Recall that the first-step estimator
\[
\tilde\btheta^{(1)} = \paran{\sum_{t=1}^{T}\bX_{t-1}^\top\bX_{t-1}}^{-1}\paran{\sum_{t=1}^{T}\bX_{t-1}^\top\by_t} = \btheta + \tilde\bSigma_x^{-1}\tilde\bSigma_{x\eps} = \btheta + \bz,
\]
where
\beq
\tilde\bSigma_x = \frac{1}{NT}\sum_{t=1}^{T}\bX_{t-1}^\top\bX_{t-1}, \quad\tilde\bSigma_{x\eps} = \frac{1}{NT} \sum_{t=1}^{T}\bX_{t-1}^\top\beps_t, \quad\text{and } \bz = \tilde\bSigma_x^{-1}\tilde\bSigma_{x\eps}.\label{Sigma_z_def}
\eeq
Plugging in the formula of $\tilde\btheta^{(1)}$ in $\tilde\beps_t=\by_t - \bX_{t-1}\tilde\btheta^{(1)}$, we have
\begin{equation} \label{eqn:wt_eps}
\begin{split}
\beps_t & = \by_t - \bX_{t-1}\btheta = \bLam\bff_t + \be_t \\
\tilde\beps_t & = \beps_t - \bX_{t-1}\bz = \beps_t + \bdelta_t.
\end{split}
\end{equation}
We denote $\bdelta_t = \tilde\beps_t - \beps_t = - \bX_{t-1}\bz$  and use $\bz = \Op{\kappa_{NT}}$ in general.
In our first CNAR case, we have $\kappa_{NT} = \frac{K}{\sqrt{T}} $ by Theorem \ref{thm:first_step_estimator}.

\subsection{Assumptions}

The following Assumptions B.1 -- B.4 are the same as Assumptions 1-4 from \cite{fan2013large}.
They are included here for completeness.
See \cite{fan2013large} for more discussions on these assumptions.
Assumption \ref{assum:regul-covari} is special for our setting.
It puts regularity condition on the CNAR covariates.
Sub-Gaussian covariates satisfy Assumption \ref{assum:regul-covari}.

\begin{assumption}[\bf Pervasiveness] \label{assum:pervasive}
All eigenvalues of the $M\times M$ matrix $N^{-1}\blam^\top\blam$ are bounded away from both 0 and $\infty$ as $N\rightarrow\infty$.
\end{assumption}

Assumption \ref{assum:pervasive} requires the factors to be pervasive.
Because of the fast diverging eigenvalues, it is hard to achieve a good rate of convergence for estimating $\bSigma_\eps$ under whether the spectral norm or Frobenius norm when $N\ge T$.
However, the quantity of interest in our case is $\bSigma_\eps^{-1}$, which can be consistently estimated with good rate.

\begin{assumption}[\bf Factor and idiosyncratic components] \label{assum:e-n-f}~
\begin{enumerate}[label=(\alph*)]
\item $\braces{\bff_t,\be_t}_{t\in[T]}$ is strictly stationary.
In addition, $\E{e_{t,i}}=0$, $\E{e_{t,i}f_{jt}}=0$ for all $i\in[N]$, $j\in[M]$ and $t\in[T]$.
\item There are constants $C_1, C_2>0$ such that $\lambda_{\min}(\bSigma_e)>C_1$, $\norm{\bSigma_e}_1<C_2$ and
\begin{equation*}
\underset{i,j\in[N]}{\min} \Var\paran{e_{t,i}e_{jt}} > C_1.
\end{equation*}
\item There are $r_1, r_2 > 0$ and $b_1, b_2>0$ such that, for any $s>0$, $i\in[N]$ and $j\in[M]$,
\begin{equation*}
\Pr\paran{\abs{e_{t,i}}>s}\le \exp\paran{-(s/b_1)^{r_1}},\quad
\Pr\paran{\abs{f_{jt}}>s}\le \exp\paran{-(s/b_2)^{r_2}}.
\end{equation*}
\end{enumerate}
\end{assumption}

\begin{assumption}[\bf Strong mixing] \label{assum:mixing}
Let $\calF_{-\infty}^0$ and $\calF_T^0$ denote the $\sigma$-algebras that are generated by $\braces{(\bff_t,\be_t):t\le 0}$ and $\braces{(\bff_t,\be_t):t\ge T}$ respectively.
Define the mixing coefficient
\begin{equation*}
\alpha(T) \defeq \underset{A\in\calF_{-\infty}^0, B\in\calF_T^0}{\sup} \abs{\Pr(A)\Pr(B)-\Pr(AB)}.
\end{equation*}
There exists $r_3>0$ such that $\gamma \defeq 3r_1^{-1}+1.5r_2^{-1}+r_3^{-1}>1$, and $C>0$ satisfying, for all $T>0$,
\begin{equation*}
\alpha(T) \le \exp(-C T^{r_3}).
\end{equation*}
\end{assumption}

Note that Assumption \ref{assum:mixing} is weaker than the assumption that $\braces{(\bff_t,\be_t)}$ are uncorrelated overtime.
Thus it is automatically satisfied under Assumptions in Appendix A.

\begin{assumption}[\bf Regularity conditions] \label{assum:regul}
There exists $C > 0$ such that, for all $i\in[N]$, $t\in[T]$, and $s\in[T]$,
\begin{enumerate}[label=(\alph*)]
\item $\norm{\blam_{i\cdot}}_{\max} < C$ where $\lambda_{i\cdot}$ is the $i$-th row of $\bLam$.
\item $\E{N^{-1/2}\paran{\be_s^\top\be_t-\E{\be_s^\top\be_t}}}^4 < C$.
\item $\E{\norm{N^{-1/2}\blam^\top\be_t}^4} < C$.
\end{enumerate}
\end{assumption}

\begin{assumption}[\bf Regularity condition on CNAR covariates]~ \label{assum:regul-covari}
Re call that $\bX_{t-1} = ((\by_{t-1}^\top \bU_N)\otimes \bU_N, \by_{t-1}, \bZ_{t-1})\in\RR^{N\times (K^2+p+1)}$.
Let $D = K^2+p+1$,
\begin{enumerate}[label=(\alph*)]
\item For all $t\in[T]$, $i \in [N]$, $j\in[D]$, $\E{x_{t,ij}^4} < C$,
\item For all $t\in[T]$, $j\in[D]$, $\E{\norm{N^{-1}\blam^\top\bx_{t,\cdot j}}^4} < C$ .
\item For all $s,t\in[T]$, $j\in[D]$, $\E{\norm{N^{-1}\be_s^\top\bx_{t,\cdot j}}^4} < C$.
\item For all $i \in [N]$ and $j \in [D]$, define random variable
$$\zeta_{t-1,ij} = \E{\frac{1}{\sqrt{T}} \sum_{s=1}^T
\Big(x_{t-1,ij}^2 x_{s-1, ij}^2 - \E{x_{t-1,ij}^2 x_{s-1, ij}^2} \Big)
 \mid x_{t-1,ij} },$$
we assume $\E{\zeta_{t-1,ij} ^4} = \bigO{1}$.
\end{enumerate}
\end{assumption}

\subsection{Technical lemmas}

\begin{lemma}
    \label{thm:VAR-properties}
    Consider a stationary auto-regressive model of order 1 AR(1): $y_t = \phi y_{t-1} + \eps_t$, $t\in[T]$ where $\abs{\phi} \le 1$, then for any $s, t \in [T]$,
    \begin{enumerate}[label = (\roman*)]
        \item $\underset{T\rightarrow \infty}{\lim} \sum_{t=1}^T \abs{ \E{  y_s y_t } } = \bigO{1}$.
        \item $\underset{T\rightarrow \infty}{\lim}  \sum_{t=1}^T \abs{ \E{  y_s e_t } } = \bigO{1}$.
        \item $\underset{T\rightarrow \infty}{\lim}  \sum_{t=1}^T \abs{ \E{  e_s y_t } } =  \bigO{1}$.
    \end{enumerate}
\end{lemma}

    \begin{proof}
        \begin{enumerate}[label = (\roman*)]
            \item We have $\abs{ \E{  y_0 y_t } } = \abs{\phi}^{\abs{t}}$.
            Then, $\underset{T\rightarrow \infty}{\lim}  \sum_{t=1}^T \abs{ \E{  y_0 y_t } }  = \sum_{t=1}^T \abs{\phi}^{\abs{t}}  = \paran{1 - \phi^2}^{-1} = \bigO{1}$ by stationary condition $\abs{\phi} \le 1$.
            For any $s\in[T]$, we have $\underset{T\rightarrow \infty}{\lim}  \sum_{t=1}^T \abs{ \E{  y_s y_t } } \le 2 \sum_{t=1}^T \abs{\phi}^{\abs{t}} = \bigO{1}$.
        \end{enumerate}
        Proofs of (ii) and (iii) are the same.
    \end{proof}

%

%
%
%

\begin{lemma}
Suppose that the random variables $Z_1$ and $Z_2$ both satisfy the exponential-type tail condition: There exist $r_1, r_2\in(0,1)$ and $b_1,b_2>0$, such that, $\forall s>0$,
\begin{equation*}
\Pr\paran{\abs{Z_i}>s} \le \exp\paran{1-(s/b_i)^{r_i}}, \quad i=1,2.
\end{equation*}
Then for some $r_3$ and $b_3>0$, and any $s>0$,
\begin{equation*}
\Pr\paran{\abs{Z_1Z_2}>s} \le \exp\paran{1-(s/b_3)^{r_3}}.
\end{equation*}
\end{lemma}
\begin{proof}
Lemma A.2 in \cite{fan2011high}.
\end{proof}

\begin{lemma}\label{thm:max-bound-4th}
Suppose that the random variable $Z_t$, $t\in[T]$, has bounded fourth moments, that is, $\EE\abs{Z_t-\mu}^4 < C$ where $\bmu=\E{Z_t}$ and $C$ is a positive constant.
Then $\underset{t\in[T]}{\max}\;\abs{Z_t-\mu} = \Op{T^{1/4}}$.
\end{lemma}
\begin{proof}
Applying the Markov's inequality to the fourth moments, we have $\Pr\paran{\abs{Z_t-\mu} \ge \delta} \le \delta^{-4}\EE\abs{Z_t-\mu}^4 \le C\delta^{-4}$, $\forall\; \delta > 0$.
By Bonferroni's method, we have
\[
\Pr\paran{\underset{t\in[T]}{\max}\;\abs{Z_t-\mu} \ge \delta} \le \sum_{t\in[T]} \Pr\paran{\abs{Z_t-\mu} \ge \delta} \le C\delta^{-4} T, \quad \forall\; \delta > 0.
\]
For any $\varepsilon\in(0,1)$, this probability can be driven below $\varepsilon$ by choosing $\delta \ge \varepsilon^{-1/4} C^{1/4} T^{1/4}$.
Thus, $\underset{t\in[T]}{\max}\;\abs{Z_t-\mu} = \Op{T^{1/4}}$.

\end{proof}

\begin{lemma} \label{thm:first-resid}
Recall that $\tilde\bdelta_t = -\bX_{t-1}\bz$ with $\bz$ defined in (\ref{Sigma_z_def}).
Suppose that $\bz = \Op{\kappa_{NT}}$
{and $\log N\ll T$}. Under the Assumption \ref{assum:regul-covari}, we have
\begin{enumerate}[label=(\alph*)]
\item $\underset{i\le N}{\max}\frac{1}{T}\sum_{t=1}^T \abs{\tilde\delta_{t,i}}^2 = \Op{\kappa_{NT}^2\cdot K^2}$.

\item $\underset{i\le N}{\max}\frac{1}{T}\sum_{t=1}^T \abs{\tilde\delta_{t,i}} = \Op{\kappa_{NT}\cdot K}$.
\end{enumerate}
\end{lemma}

\begin{proof}
\begin{enumerate}[label=(\alph*)]
\item
Similar to the derivation of \eqref{uyyu_conv} in Lemma \ref{vSigx} we have for any $j\in[D]$, $i\in[N]$,
$\Pr\paran{ \abs{\frac{1}{T} \sum_{t=1}^T x_{t-1, ij}^2 - \E{x_{t-1, ij}^2}} > s } \le 2\exp(-cT s^2)$.
Then,
\beq
\Pr\paran{ \underset{i\le N}{\max} \abs{\frac{1}{T} \sum_{t=1}^T x_{t-1, ij}^2 - \E{x_{t-1, ij}^2}} > s } \le 2 N \cdot \exp(-cT s^2).\label{x_ine}
\eeq
Thus, $\underset{i\le N}{\max} \frac{1}{T} \sum_{t=1}^T x_{t-1, ij}^2 = \Op{1 + \sqrt{\frac{\log(N)}{T}}} = \Op{1}$ because $\E{x_{t-1, ij}^2}$ is bounded.

Further, we have
\begin{equation*}
\begin{aligned}
\underset{i\le N}{\max}\frac{1}{T}\sum_{t=1}^T \abs{\tilde\delta_{t,i}}^2  & = \underset{i\le N}{\max} \frac{1}{T} \sum_{t=1}^T \abs{\bz^\top\bx_{t-1, i\cdot}\bx_{t-1, i\cdot}^\top\bz} \le \underset{i\le N}{\max} \frac{1}{T} \sum_{t=1}^T \norm{\bx_{t-1, i\cdot}}^2\norm{\bz}^2 \\
& \le \norm{\bz}^2 \sum_{j=1}^D \underset{i\le N}{\max} \frac{1}{T} \sum_{t=1}^T x_{t-1, ij}^2
= \Op{\kappa_{NT}^2\cdot K^2},
\end{aligned}
\end{equation*}
since $D=K^2 + p + 1$ grows at the rate of $K^2$.

\item Similar to (a), by (\ref{x_ine}), we have
$\underset{i\le N}{\max}\frac{1}{T}\sum_{t=1}^T \abs{\tilde\delta_{t,i}}
= \underset{i\le N}{\max}\frac{1}{T}\sum_{t=1}^T \abs{\bx_{t-1, i\cdot}^\top\bz} \\
\le \underset{i\le N}{\max}\frac{1}{T}\sum_{t=1}^T \norm{\bx_{t-1, i\cdot}}\norm{\bz}
\le \norm{\bz} \paran{\underset{i\le N}{\max}\frac{1}{T}\sum_{t=1}^T \norm{\bx_{t-1, i\cdot}}^2}^{1/2}
= \Op{\kappa_{NT}\cdot K}$.
\end{enumerate}
\end{proof}

\begin{lemma}  \label{thm:delta-e}
Suppose that $\bz = \Op{\kappa_{NT}}$, $\log(N) = \smlo{T^{\gamma/6}}$ where $\gamma$ is defined in Assumption \ref{assum:mixing}, $T=\smlo{N^2}$ and Assumptions \ref{assum:pervasive} -- \ref{assum:regul-covari} hold,
we have
\begin{enumerate}[label=(\alph*)]
\item  $\underset{i\le N, j\le N}{\max} \abs{\frac{1}{T} \sum_{t=1}^T \delta_{t,i}\delta_{t,j}} = \Op{\kappa_{NT}^2 K^2}$.
\item $\underset{i\le N, j\le N}{\max} \abs{\frac{1}{T} \sum_{t=1}^T \delta_{t,i}e_{t,j}}= \Op{\kappa_{NT}K\sqrt{\frac{\log(N)}{T}}}$.
\item $\max_{i\le M,j\le N}\abs{\frac{1}{T}\sum_{t=1}^T f_{it}\delta_{t,j}} = \Op{\kappa_{TN}K\sqrt{\frac{\log(N)}{T}}}$.
\item $\underset{i\le N, j\le N}{\max} \abs{\frac{1}{T} \sum_{t=1}^T \delta_{t,i}\eps_{t,j}}= \Op{\kappa_{NT}K\sqrt{\frac{\log(N)}{T}}}$.
\end{enumerate}
\end{lemma}

\begin{proof}
\begin{enumerate}[label=(\alph*)]
\item $\underset{i\le N, j\le N}{\max} \abs{\frac{1}{T} \sum_{t=1}^T \delta_{t,i}\delta_{t,j}} = \underset{i\le N, j\le N}{\max} \abs{\frac{1}{T} \sum_{t=1}^T \bz^\top\bx_{t-1, i\cdot}\bx_{t-1, j\cdot}^\top\bz} \\
\le \underset{i\le N, j\le N}{\max} \norm{\frac{1}{T} \sum_{t=1}^T \bx_{t-1, i\cdot}\bx_{t-1, j\cdot}^\top - \E{\bx_{t-1, i\cdot}\bx_{t-1, j\cdot}^\top }}\norm{\bz}^2
+ \underset{i\le N, j\le N}{\max} \norm{\E{\bx_{t-1, i\cdot}\bx_{t-1, j\cdot}^\top}} \norm{\bz}^2 \\
= \Op{\kappa_{NT}^2K^2\sqrt{\frac{\log(N)}{T}} + \kappa_{NT}^2K^2}$,
since $\underset{i\le N, j\le N}{\max}\norm{\E{x_{t-1, ik}x_{t-1, jk}}} \le C$ under Assumption \ref{assum:regul-covari}.
\item
First, similar to the derivation of \eqref{uyyu_conv} in Lemma \ref{vSigx} we have for any $k\in[D]$, $i,j\in[N]$,
$\Pr\paran{ \abs{\frac{1}{T} \sum_{t=1}^T e_{jt} x_{t-1, ik} } > s } \le 2\exp(-cT s^2)$.
Then
$\Pr\paran{ \max_{i,j}\abs{\frac{1}{T} \sum_{t=1}^T e_{jt}x_{t-1, ik} } > s } \le 2N^2\exp(-cT s^2)$.
Thus, $\underset{i\le N, j\le N}{\max} \abs{\frac{1}{T}\sum_{t=1}^T e_{jt}x_{t-1, ik}} = \Op{\sqrt{\frac{2\log(N)}{T}}}$.

Then, we have
$\underset{i\le N, j\le N}{\max} \abs{\frac{1}{T} \sum_{t=1}^T \delta_{t,i} e_{t,j}}
= \underset{i\le N, j\le N}{\max} \abs{\frac{1}{T}\sum_{t=1}^T e_{jt}\bx_{t-1, i\cdot}^\top\bz} \\
\le \underset{i\le N, j\le N}{\max} \norm{\frac{1}{T}\sum_{t=1}^T e_{jt}\bx_{t-1, i\cdot}^\top}\norm{\bz}
= \Op{\kappa_{NT}K\sqrt{\frac{2\log(N)}{T}}}$.

\item
We have
$\max_{i\le M,j\le N}\abs{\frac{1}{T}\sum_{t=1}^T f_{it}\delta_{t,j}} = \max_{i\le M,j\le N}\abs{\frac{1}{T}\sum_{t=1}^T f_{it}\bx_{t-1,j\cdot}^\top\bz} \\
\le \max_{i\le M,j\le N}\norm{\frac{1}{T}\sum_{t=1}^T f_{it}\bx_{t-1,j\cdot}^\top}\norm{\bz}
= \Op{\kappa_{TN}K\sqrt{\frac{\log(N)}{T}}}$,
where we use for all $l\in[D]$, $\underset{i\le N, j\le N}{\max} \norm{\frac{1}{T}\sum_{t=1}^T f_{it}x_{t-1, il}} = \Op{\sqrt{\frac{\log(N)}{T}}}$ which can be derived using similar technique as that in (b).

\item Similar to that in (b).
\end{enumerate}

\end{proof}

\subsection{Covariance estimation of \texorpdfstring{$\tilde\beps_t$}{estimated residual}}

\subsubsection{Loading and factor estimators}

The common factors $\braces{\bff_t}_{t=1}^T$ need to be estimated uniformly in $t\le T$.
When we do not observe $\braces{\bff_t}_{t=1}^T$, in addition to the factor loadings, there are $TM$ factors to estimate.
Intuitively, the condition $T = \op{N^2}$ requires the number of parameters that are introduced by the unknown factors to be `not too many', so we can estimate them uniformly.
Technically, as demonstrated by \cite{bickel2008covariance}, \cite{cai2011adaptive} among many others, achieving uniform accuracy is essential for large covariance estimation.

Recall that $\tilde \beps_t$ is defined in (\ref{eqn:wt_eps}).
Let $\tilde\calE=\begin{pmatrix} \tilde\beps_1 & \cdots & \tilde\beps_T \end{pmatrix}^\top\in \RR^{T\times N}$, $\bF=\begin{pmatrix} \bff_1 & \cdots & \bff_T \end{pmatrix}^\top\in \RR^{T\times M}$.
In addition, let $\bV$ be the $M\times M$ diagonal matrix of the first $M$ largest eigenvalues of $\frac{1}{TN}\tilde\calE\tilde\calE^\top$ in decreasing order,
and $\tilde\bF$ is $\sqrt{T}$-times the corresponding eigenvector (i.e., the estimated factors).
By the definition of eigenvectors and eigenvalues, we have
$\frac{1}{TN}\tilde\calE\tilde\calE^\top\tilde\bF = \tilde\bF\bV$ (or $\frac{1}{TN}\tilde\calE\tilde\calE^\top\tilde\bF\bV^{-1} = \tilde\bF$) and $\frac{1}{T}\tilde\bF^\top\tilde\bF=\bI_M$.
Let $\bDelta=\begin{pmatrix} \bdelta_1 ,& \cdots, & \bdelta_T \end{pmatrix}^\top\in \RR^{T\times N}$ with $\bdelta_t$ defined in \eqref{eqn:wt_eps},
\begin{equation*}
\tilde\calE = \bF\bLam^\top + \bE + \bDelta.
\end{equation*}
Let $\bH = \frac{1}{TN}\bV^{-1}\tilde\bF^\top\bF\bLam^\top\bLam$, we have a similar decomposition of factor estimation errors as expression (A.1) in \cite{bai2003inferential}:
\begin{equation} \label{eqn:fac-err-decomp}
\begin{aligned}
\tilde\bff_t - \bH\bff_t & = \bV^{-1}
\left\{\frac{1}{T}\sum_{s=1}^T\tilde\bff_s\frac{\E{\be_s^\top\be_t}}{N}
+ \frac{1}{T}\sum_{s=1}^T\tilde\bff_s\frac{\be_s^\top\be_t-\E{\be_s^\top\be_t}}{N} \right. \\
& + \frac{1}{T}\sum_{s=1}^T\tilde\bff_s\frac{\be_s^\top\bLam\bff_t}{N}
+ \frac{1}{T}\sum_{s=1}^T\tilde\bff_s\frac{\bff_s^\top\bLam^\top\be_t}{N} \\
& \left. + \frac{1}{T}\sum_{s=1}^T\tilde\bff_s\frac{\bdelta_s^\top\bdelta_t}{N}
+ \frac{1}{T}\sum_{s=1}^T\tilde\bff_s\frac{\beps_s^\top\bdelta_t}{N}
+ \frac{1}{T}\sum_{s=1}^T\tilde\bff_s\frac{\bdelta_s^\top\beps_t}{N} \right\} \\
& = \bV^{-1} \sum_{i=1}^{7} I_{i}.
\end{aligned}
\end{equation}
To bound $\tilde\bff_t - \bH\bff_t$, we need the following lemmas, along with Lemma 8 and 9 
in \cite{fan2013large}.

\begin{lemma} \label{thm:fac-error-avg-1}
Suppose that $\bz = \Op{\kappa_{NT}}$,
and Assumption \ref{assum:pervasive} -- \ref{assum:regul-covari} hold,
we have for all $i \le M$
\begin{enumerate}[label=(\alph*)]
\item $\frac{1}{T}\sum_{t=1}^T \paran{\frac{1}{T}\sum_{s=1}^T\tilde f_{s,i}\frac{\bdelta_s^\top\bdelta_t}{N}}^2 =
\Op{\kappa_{NT}^4 K^4}$.
\item $\frac{1}{T}\sum_{t=1}^T \paran{\frac{1}{T}\sum_{s=1}^T\tilde f_{s,i}\frac{\beps_s^\top\bdelta_t}{N}}^2
= \Op{\kappa_{NT}^2K^2}$.
\item $\frac{1}{T}\sum_{t=1}^T \paran{\frac{1}{T}\sum_{s=1}^T\tilde f_{s,i}\frac{\bdelta_s^\top\beps_t}{N}}^2
= \Op{\kappa_{NT}^2K^2}$.
\end{enumerate}
\end{lemma}

\begin{proof}
\begin{enumerate}[label=(\alph*)]
\item
Note that $\norm{\E{N^{-1}\bX_{s-1}^\top\bX_{t-1}}}^2 \le \norm{\E{N^{-1}\bX_{s-1}^\top\bX_{t-1}}}^2_F = \sum_{i=1}^{D} \sum_{j=1}^{D} \E{ N^{-1} \bx_{s-1, \cdot i}^\top \bx_{t-1, \cdot j} }$.
Since, the columns of $\bx_{s-1, \cdot i}$ and $\bx_{t-1, \cdot j}$ are either rotated VAR(1)~ $\by_{t-1}$ or uncorrelated $\bZ_{t-1}$, by Lemma \ref{thm:VAR-properties}, we have for any $i,j \in [D]$,
\begin{equation*}
\begin{split}
\sum_{s=1}^T \abs{ \E{ N^{-1} \bx_{s-1, \cdot i}^\top \bx_{t-1, \cdot j} } } \le \sum_{s=1}^T \abs{\E{ N^{-1} \bx_{s-1, \cdot i}^\top \bx_{t-1, \cdot j} } } =   \sum_{s=1}^T \abs{ \E{ x_{s-1, i}^\top x_{t-1, j} } } = \bigO{1},
\end{split}
\end{equation*}
Therefore, for any fix $t$, we have
\begin{equation*}
\begin{split}
\sum_{s=1}^T \norm{ \E{N^{-1}\bX_{s-1}^\top\bX_{t-1}}}^2 & \le  \sum_{i=1}^{D} \sum_{j=1}^{D} \sum_{s=1}^T \abs{ \E{ N^{-1} \bx_{s-1, \cdot i}^\top \bx_{t-1, \cdot j} } } = \bigO{D^2} = \bigO{K^4}.  \\
\end{split}
\end{equation*}
Then we have $\frac{1}{T}\sum_{t=1}^T \sum_{s=1}^T \norm{ \E{N^{-1}\bX_{s-1}^\top\bX_{t-1}}}^2 = \bigO{K^4}$.
Finally, we have
\begin{equation*}
\begin{aligned}
\frac{1}{T}\sum_{t=1}^T \paran{\frac{1}{T}\sum_{s=1}^T\tilde f_{s,i}\frac{\bdelta_s^\top\bdelta_t}{N}}^2
& \le \frac{1}{T}\sum_{t=1}^T \paran{\frac{1}{T}\sum_{s=1}^T \tilde f_{s,i}^2} \cdot \frac{1}{T}\sum_{s=1}^T \paran{\bz^\top\frac{\bX_{s-1}^\top\bX_{t-1}}{N}\bz}^2 \\
& {\le} \norm{\bz}^4 \frac{1}{T}\sum_{t=1}^T \frac{1}{T}\sum_{s=1}^T \norm{\frac{\bX_{s-1}^\top\bX_{t-1}}{N}}^2  \\
& = \norm{\bz}^4 \frac{1}{T}\sum_{t=1}^T \frac{1}{T}\sum_{s=1}^T \norm{\frac{\bX_{s-1}^\top\bX_{t-1}}{N} - \E{\frac{\bX_{s-1}^\top\bX_{t-1}}{N}}}^2 \\
& +  \norm{\bz}^4 \frac{1}{T} \cdot \frac{1}{T}\sum_{t=1}^T \sum_{s=1}^T \norm{ \E{\frac{\bX_{s-1}^\top\bX_{t-1}}{N}}}^2   \\
& = \Op{\kappa_{NT}^4 K^4}
\end{aligned}
\end{equation*}
by Lemma \ref{vSigx} and Lemma \ref{thm:VAR-properties}.

\item
\begin{equation*}
\begin{aligned}
\frac{1}{T}\sum_{t=1}^T \paran{\frac{1}{T}\sum_{s=1}^T\tilde f_{s,i}\frac{\beps_s^\top\bdelta_t}{N}}^2
& \le \frac{1}{T}\sum_{t=1}^T \paran{\frac{1}{T}\sum_{s=1}^T \tilde f_{s,i}^2} \cdot \frac{1}{T}\sum_{s=1}^T \paran{\frac{\beps_s^\top\bX_{t-1}}{N}\bz}^2 \\
& = \frac{1}{T}\sum_{t=1}^T \frac{1}{T}\sum_{s=1}^T \norm{\frac{\beps_s^\top\bX_{t-1}}{N}}^2 \norm{\bz}^2 \\
& = \frac{1}{T}\sum_{t=1}^T \frac{1}{T}\sum_{s=1}^T \norm{ \frac{\beps_s^\top\bX_{t-1}}{N} - \E{\frac{\beps_s^\top\bX_{t-1}}{N}} }^2 \norm{\bz}^2 \\
& + \frac{1}{T}\sum_{t=1}^T \frac{1}{T}\sum_{s=1}^T \norm{\E{\frac{\beps_s^\top\bX_{t-1}}{N}}}^2 \norm{\bz}^2 \\
& = \Op{\kappa_{NT}^2K^2},
\end{aligned}
\end{equation*}
where the last equation is derived the same as that in (a).

\item Same as that of (b).
\end{enumerate}
\end{proof}


\begin{corollary} \label{thm:factors-est}
Suppose that $\bz = \Op{\kappa_{NT}}$, $\kappa_{NT}K =\smlo{1}$,
and Assumption \ref{assum:pervasive} -- \ref{assum:regul-covari} hold,
\begin{enumerate}[label=(\alph*)]
\item $\underset{i \le M}{\max} \frac{1}{T} \sum_{t=1}^T \paran{\tilde\bff_t-\bH\bff_t}_i^2 = \Op{\frac{1}{T}+\frac{1}{N} + \kappa_{NT}^2 K^2}$.
\item $\frac{1}{T} \sum_{t=1}^T \norm{\tilde\bff_t - \bH\bff_t}^2 = \Op{\frac{1}{T}+\frac{1}{N} + \kappa_{NT}^2 K^2}$.
\end{enumerate}
\end{corollary}

\begin{proof}~
\begin{enumerate}[label=(\alph*)]
\item  Recall that $\bV$ be the $M\times M$ diagonal matrix of the first $M$ largest eigenvalues of $\frac{1}{TN}\tilde\calE\tilde\calE^\top$ in decreasing order.
Similar to Lemma 5 in \cite{fan2013large}, it is straightforward to prove that all the eigenvalues of $\bV$ are bounded away from 0.
Using inequality $\paran{\sum_{i=1}^{M}a_i}^2 \le M\sum_{i=1}^{M}a_i^2$ and identity \eqref{eqn:fac-err-decomp}, we have, for some constant $C>0$,
\begin{equation*}
\begin{aligned}
\underset{i\le M}{\max} \frac{1}{T}\sum_{t=1}^{T}\paran{\tilde\bff_t - \bH\bff_t}^2_i & \le C\;\underset{i\le M}{\max} \frac{1}{T}\sum_{t=1}^{T} \paran{\frac{1}{T}\sum_{s=1}^T\tilde f_{s,i}\frac{\E{\be_s^\top\be_t}}{N}}^2 \\
& + C\;\underset{i\le M}{\max} \frac{1}{T}\sum_{t=1}^{T} \paran{\frac{1}{T}\sum_{s=1}^T\tilde f_{s,i}\frac{\be_s^\top\be_t-\E{\be_s^\top\be_t}}{N}}^2 \\
& + C\;\underset{i\le M}{\max} \frac{1}{T}\sum_{t=1}^{T} \paran{\frac{1}{T}\sum_{s=1}^T\tilde\bff_s\frac{\be_s^\top\bLam\bff_t}{N}}^2
+ C\;\underset{i\le M}{\max} \frac{1}{T}\sum_{t=1}^{T} \paran{\frac{1}{T}\sum_{s=1}^T\tilde\bff_{ s,i}\frac{\bff_{s}^\top\bLam^\top\be_t}{N}}^2 \\
& + C\;\underset{i\le M}{\max} \frac{1}{T}\sum_{t=1}^{T} \paran{\frac{1}{T}\sum_{s=1}^T\tilde\bff_{ s,i}\frac{\bdelta_s^\top\beps_t}{N}}^2
+ C\;\underset{i\le M}{\max} \frac{1}{T}\sum_{t=1}^{T} \paran{\frac{1}{T}\sum_{s=1}^T\tilde\bff_{ s,i}\frac{\beps_s^\top\bdelta_t}{N}}^2 \\
& + C\;\underset{i\le M}{\max} \frac{1}{T}\sum_{t=1}^{T} \paran{\frac{1}{T}\sum_{s=1}^T\tilde\bff_{ s,i}\frac{\bdelta_s^\top\bdelta_t}{N}}^2.
\end{aligned}
\end{equation*}
The first four terms on the right-hand side are bounded in $\Op{T^{-1}+N^{-1}}$ in Lemma 8 in \cite{fan2013large}, while the last three terms are bounded in Lemma \ref{thm:fac-error-avg-1}which produce an additional term with order $\kappa_{NT}^2 K^2$.

\item Part (b) follows inequality $\frac{1}{T} \sum_{t=1}^T \norm{\tilde\bff_t - \bH\bff_t}^2\le M\;\underset{i\le M}{\max} \frac{1}{T}\sum_{t=1}^{T}\paran{\tilde\bff_t - \bH\bff_t}^2_i$ and part (a).
\end{enumerate}
\end{proof}

\begin{lemma} \label{thm:H}~
\begin{enumerate}[label=(\alph*)]
\item $\bH\bH^\top = \bI_K + \Op{T^{-1/2}+N^{-1/2} + \kappa_{NT} K}$.
\item $\bH^\top\bH = \bI_K + \Op{T^{-1/2}+N^{-1/2} + \kappa_{NT} K}$.
\end{enumerate}
\end{lemma}
\begin{proof}
Using Lemma 
\ref{thm:fac-error-avg-1}, and \ref{thm:factors-est}, the proof uses the argument as in Lemma 11 in \cite{fan2013large} and is thus omitted here.
\end{proof}

\begin{corollary} \label{thm:lam-f-max-bounds}
Suppose that $\bz = \Op{\kappa_{NT}}$, $\kappa_{NT}K =\smlo{1}$, 
and $\log(N) = \smlo{T^{\gamma/6}}$ where $\gamma$ is defined in Assumption \ref{assum:mixing}, and Assumption \ref{assum:pervasive} -- \ref{assum:regul-covari} hold,
 \begin{enumerate}[label=(\alph*)]
\item $\underset{i\le N}{\max} \norm{\tilde\blam_{i\cdot}-\bH\blam_{i\cdot}} = \Op{\frac{1}{\sqrt{N}} + \sqrt{\frac{\log(N)}{T}} + \kappa_{NT}K}$.
 \end{enumerate}
\end{corollary}

\begin{proof}~
\begin{enumerate}[label=(\alph*)]

\item Using the fact that $\tilde\blam_{i\cdot}=\frac{1}{T} \tilde\bF^\top\tilde\beps_{\cdot i}$ and that $\frac{1}{T} \tilde\bF^\top\tilde\bF=\bI_M$, we have
\begin{equation*}
\begin{aligned}
\tilde\blam_{i\cdot} - \bH\blam_{i\cdot} & = \frac{1}{T}\bH\bF^\top\be_{\cdot i}
+ \frac{1}{T}\paran{\wt\bF-\bF\bH^\top}^\top\beps_{\cdot i}
+ \bH\paran{\frac{1}{T}\bF^\top\bF - \bI_M}\blam_{i\cdot}\\
& + \frac{1}{T}\paran{\tilde\bF-\bF\bH^\top}^\top\bdelta_{\cdot i} + \frac{1}{T}\bH\bF^\top\bdelta_{\cdot i}.
\end{aligned}
\end{equation*}

We bound each term on the right-hand side.
It follows from Lemma 4 (iii) in \cite{fan2013large} and Lemma \ref{thm:H} that
\begin{equation*}
    \underset{i \le N}{\max} \norm{ \frac{1}{T}\bH\bF^\top\be_{\cdot i} } \le \norm{\bH} \underset{i \le N}{\max} \sqrt{ \sum_{j = 1}^{M} \paran{\frac{1}{T} \sum_{t=1}^T f_{t,j} e_{t,i} }^2 } = \Op{\sqrt{\frac{\log (N)}{T}}}.
\end{equation*}

For the second term, $\E{\eps_{t,i}^2} = \bigO{1}$.
Therefore, $\underset{i \le N}{\max} T^{-1} \sum_{t=1}^T \eps_{t,i}^2 = \Op{1}$.
The Cauchy-Schwarz inequality and (a) imply
\begin{equation*}
\underset{i \le N}{\max} \norm{ \frac{1}{T}\paran{\wt\bF-\bF\bH^\top}^\top\beps_{\cdot i}  } \le \underset{i \le N}{\max} \paran{ \frac{1}{T} \sum_{t=1}^T \eps_{t,i}^2 \frac{1}{T} \sum_{t=1}^T \norm{\tilde\bff_t - \bH\bff_t}^2   }^{1/2} = \Op{\frac{1}{\sqrt{T}} + \frac{1}{\sqrt{N}} + \kappa_{NT} K}.
\end{equation*}

For the third term, $\norm{ \frac{1}{T}\bF^\top\bF - \bI_M } = \Op{T^{-1/2}}$ and $\underset{i \le N}{\max} \norm{\blam_{i\cdot}} = \bigO{1}$ imply that the third term is $\Op{T^{-1/2}}$.

For the second last term, by Cauchy-Schwarz inequality,  Lemma \ref{thm:delta-e} (a) and Lemma \ref{thm:factors-est}, we have
\begin{equation*}
\underset{i\le N}{\max} \norm{\frac{1}{T}\paran{ \tilde\bF-\bF\bH^\top}^\top\bdelta_{\cdot i}}
\le\underset{i \le N}{\max} \paran{ \frac{1}{T} \sum_{t=1}^T \delta_{t,i}^2 \frac{1}{T} \sum_{t=1}^T \norm{\tilde\bff_t - \bH\bff_t}^2   }^{1/2} =\Op{\kappa_{NT}K\paran{ \frac{1}{\sqrt{N}}+\frac{1}{\sqrt{T}}  + \kappa_{NT} K }}.
\end{equation*}

For the last term, we have
\begin{equation*}
\underset{i\le N}{\max} \norm{\frac{1}{T}\bH\bF^\top\bdelta_{\cdot i}}
\le \norm{\bH} \underset{i\le N}{\max} \norm{\frac{1}{T}\sum_{t=1}^T\bff_{t}\bx_{t-1, i\cdot}^\top}\norm{\bz}
= \Op{\sqrt{\frac{\log N}{T}} \kappa_{NT}K},
\end{equation*}
where we use the result that $\underset{i\le N}{\max} \norm{\frac{1}{T}\sum_{t=1}^T\bff_{t}\bx_{t-1, i\cdot}^\top} = \Op{\sqrt{\frac{\log N}{T}} K}$ derived in Lemma \ref{thm:delta-e} (c).
The result follows by adding all term together.
\end{enumerate}
\end{proof}

\subsubsection{Idiosyncratic covariance bound}

\begin{lemma} \label{thm:idio-est-bound}
    Suppose that $\bz = \Op{\kappa_{NT}}$,  $\kappa_{NT} K = \smlo{1}$,
    $\log(N) = \smlo{T^{\gamma/6}}$ where $\gamma$ is defined in Assumption \ref{assum:mixing}, and Assumption \ref{assum:pervasive} -- \ref{assum:regul-covari} hold.
    \begin{enumerate}[label=(\roman*)]
        \item $\underset{i\le N}{\max} \frac{1}{T} \sum_{t=1}^T \abs{\tilde e_{t,i}-e_{t,i}}^2 = \Op{\frac{1}{N} + \frac{\log(N)}{T} +\kappa_{NT}^2 K^2}$.
        \item $\underset{i\le N}{\max} \frac{1}{T} \sum_{t=1}^T \abs{\tilde e_{t,i}-e_{t,i}}\abs{e_{t,i}} = \Op{\frac{1}{\sqrt{N}} + \sqrt{\frac{\log(N)}{T}} +\kappa_{NT} K}$.
    \end{enumerate}
\end{lemma}

\begin{proof}
    \begin{enumerate}[label=(\roman*)]
        \item We have
        \begin{equation*}
       \abs{\tilde e_{t,i} - e_{t,i}} = \abs{ \tilde\eps_{t,i} - \tilde\blambda_{i\cdot}\tilde\bff_t - e_{t,i}} = \abs{ \blambda_{i\cdot}\bff_t + e_{t,i}  + \delta_{t,i} - \tilde\blambda_{i\cdot}\tilde\bff_t - e_{t,i}} = \abs{\blambda_{i\cdot}^\top \bff_t - \tilde\blambda_{i\cdot}^\top \tilde\bff_t} + \abs{\delta_{t,i}},
        \end{equation*}
        then
        \begin{equation*}
        \begin{aligned}
        & \underset{i\le N}{\max} \frac{1}{T} \sum_{t=1}^T \abs{\tilde e_{t,i} - e_{t,i}}^2 \le \underset{i\le N}{\max} \frac{2}{T} \sum_{t=1}^T\abs{\blambda_{i\cdot}^\top \bff_t - \tilde\blambda_{i\cdot}^\top \hat\bff_t}^2
        + \underset{i\le N}{\max} \frac{2}{T} \sum_{t=1}^T \abs{\delta_{t,i}}^2\\
        & \le 8 \; \underset{i\le N}{\max} \norm{\blambda_{i\cdot}^\top\bH^\top}^2 \frac{1}{T} \sum_{t=1}^T \norm{ \hat \bff_t - \bH \bff_t }^2 + 8 \; \underset{i\le N}{\max} \norm{\tilde\blambda_{i\cdot}^\top - \blambda_{i\cdot}^\top\bH^\top}^2 \frac{1}{T} \sum_{t=1}^T \norm{ \hat \bff_t}^2   \\
        & + 8 \; \underset{i\le N}{\max} \norm{\blambda_{i\cdot}}^2 \frac{1}{T} \sum_{t=1}^T \norm{ \hat \bff_t}^2 \norm{\bH^\top\bH-\bI_K}_F^2 + \underset{i\le N}{\max} \frac{2}{T} \sum_{t=1}^T \abs{\delta_{t,i}}^2 \\
        & = \Op{\frac{1}{N} + \frac{\log(N)}{T}+\kappa_{NT}^2 K^2}, \\
        \end{aligned}
        \end{equation*}
        which follows from Corollary \ref{thm:factors-est}, \ref{thm:lam-f-max-bounds}, \ref{thm:H}, and Lemma \ref{thm:first-resid}.
        \item We have
        \begin{equation}
        \begin{split}
        \underset{i\le N}{\max} \frac{1}{T} \sum_{t=1}^T \abs{\tilde e_{t,i}-e_{t,i}}\abs{e_{t,i}} & \le   \underset{i\le N}{\max} \paran{\frac{1}{T} \sum_{t=1}^T \abs{\tilde e_{t,i}-e_{t,i}}^2 \paran{\frac{1}{T} \sum_{t=1}^T e_{t,i}^2 - \E{e_{t,i}^2} } }^{1/2} \\
        & + \underset{i\le N}{\max} \paran{\frac{1}{T} \sum_{t=1}^T \abs{\tilde e_{t,i}-e_{t,i}}^2 \E{e_{t,i}^2} }^{1/2} \\
        & = \Op{\frac{1}{\sqrt{N}} + \sqrt{\frac{\log(N)}{T}} +\kappa_{NT} K } \cdot \Op{\sqrt{\frac{\log(N)}{T}}} \\
        & + \Op{\frac{1}{\sqrt{N}} + \sqrt{\frac{\log(N)}{T}} +\kappa_{NT} K }
        \end{split}
        \end{equation}
    \end{enumerate}
\end{proof}

\begin{corollary} \label{thm:idio-cov-bound}
Suppose that $\bz = \Op{\kappa_{NT}}$,  $\kappa_{NT} K = \smlo{1}$,
$\log(N) = \smlo{T^{\gamma/6}}$ where $\gamma$ is defined in Assumption \ref{assum:mixing}, and Assumption \ref{assum:pervasive} -- \ref{assum:regul-covari} hold.
Then, the proposed estimator of idiosyncratic error based on first-step CNAR residual $\tilde{\bSigma}_e$ satisfies
\begin{enumerate}[label=(\roman*)]
\item $\norm{\tilde{\bSigma}_e - \bSigma_e} = \Op{\frac{1}{\sqrt{N}} + \sqrt{\frac{\log(N)}{T}} +\kappa_{NT} K}$.
\item $\norm{\tilde{\bSigma}_e^{-1} - \bSigma_e^{-1}} = \Op{\frac{1}{\sqrt{N}} + \sqrt{\frac{\log(N)}{T}} +\kappa_{NT} K}$.
\item $\norm{\tilde{\bSigma}_e - \bSigma_e}_{\Sigma_\eps}^2 = \Op{\frac{1}{N} + \frac{\log(N)}{T} +\kappa_{NT}^2 K^2}$.
\end{enumerate}
\end{corollary}

\begin{proof}
    \begin{enumerate}[label=(\roman*)]
        \item Note that $\tilde{\bSigma}_e = \diag\braces{ \frac{1}{T} \sum_{t=1}^T \tilde e_{t,i}^2 }$ and $\bSigma_e = \diag\braces{\E{e_{t,i}^2}}$
        \begin{equation*}
        \begin{split}
        \norm{\tilde{\bSigma}_e - \bSigma_e}  & = \underset{i\le N}{\max} \braces{\Big| \frac{1}{T} \sum_{t=1}^T\paran{ \tilde e_{t,i}^2 - e_{t,i}^2}\Big| - \Big| \frac{1}{T} \sum_{t=1}^T \paran{ e_{t,i}^2 - \E{e_{t,i}^2} } \Big|}\\
        & \le \underset{i\le N}{\max} \frac{1}{T} \sum_{t=1}^T \abs{\tilde e_{t,i}-e_{t,i}}^2
        + 2 \underset{i\le N}{\max} \frac{1}{T} \sum_{t=1}^T \abs{\tilde e_{t,i}-e_{t,i}}\abs{e_{t,i}}
        + \underset{i\le N}{\max} \Big|\frac{1}{T} \sum_{t=1}^T \paran{ e_{t,i}^2 - \E{e_{t,i}^2} }\Big| \\
        & = \Op{\frac{1}{\sqrt{N}} + \sqrt{\frac{\log(N)}{T}} +\kappa_{NT} K}.
        \end{split}
        \end{equation*}
        which follows from Lemma \ref{thm:idio-est-bound} and Assumption \ref{assum:mixing}.
        \item Note that $\max_i|\frac{1}{T} \sum_{t=1}^T e_{t,i}^2 - \E{e_{t,i}^2}| = \Op{\sqrt{\log N/T}}$.
        Further by (i) we have $\underset{i\le N}{\min}T^{-1}\sum_t e_{t,i}^2 > c$,
        $\underset{i\le N}{\min} T^{-1}\sum_t\tilde e_{t,i}^2$ is bounded away from 0 with probability approaching 1.
        This implies that
        \[\norm{\tilde{\bSigma}_e^{-1} - \bSigma_e^{-1}} = \Op{\frac{1}{\sqrt{N}} + \sqrt{\frac{\log(N)}{T}} +\kappa_{NT} K}.\]
        \item  By (i), we have
        \begin{equation*}
            \begin{split}
              \norm{\tilde{\bSigma}_e - \bSigma_e}_{\Sigma_\eps}^2 & = \Op{ N^{-1} \norm{\tilde{\bSigma}_e - \bSigma_e}_F^2 } = \Op{ \norm{\tilde{\bSigma}_e - \bSigma_e}^2 } = \Op{\frac{1}{N} + \frac{\log(N)}{T} + \kappa_{NT}^2 K^2}
            \end{split}
        \end{equation*}
    \end{enumerate}
\end{proof}

\subsubsection{Covariance bound of the CNAR residual $\tilde{\beps}_t$}

\begin{lemma} \label{thm:cov-signal-utils-1}
Denote $\omega_{NTK} = \frac{1}{\sqrt{N}} + \sqrt{\frac{\log(N)}{T}}+\kappa_{NT}K$
, we have the following results:
\begin{enumerate}[label=(\roman*)]
    \item $\norm{\tilde\bLam - \bLam\bH^\top}_F^2 = \Op{ N \omega_{NTK}^2 }$ and $\norm{\paran{\tilde\bLam - \bLam\bH^\top}^\top \paran{\tilde\bLam - \bLam\bH^\top}}_{\Sigma_\eps}^2 = \Op{ N \omega_{NTK}^4 }$.
    \item $\norm{\bLam\bH^\top \paran{\tilde\bLam - \bLam\bH^\top}^\top}_{\Sigma_\eps}^2 = \Op{\omega_{NTK}^2}$.
    \item $\norm{\bLam (\bH^\top\bH - \bI_K) \bLam^\top}_{\Sigma_\eps}^2 = \Op{ N^{-1} \paran{T^{-1}+N^{-1} + \kappa_{NT}^2 K^2} }$.
\end{enumerate}
\end{lemma}
\begin{proof}
    We denote $\bC_T = \tilde\bLam - \bLam\bH^\top$.
    \begin{enumerate}[label=(\roman*)]
        \item We have $\norm{\tilde\bLam - \bLam\bH^\top}_F^2 \le N \cdot \underset{i\le N}{\max} \norm{ \tilde\blam_{i\cdot} - \bH\blam_{i\cdot} } = \Op{N \omega_{NTK}^2}$ by Corollary \ref{thm:lam-f-max-bounds}.
        Moreover, since
               all the eigenvalues of $\bSigma_{\eps}$ are bounded away from $0$, for any matrix $\bA$, $\norm{\bA}_{\Sigma_\ve}^2 = \Op{N^{-1}} \norm{\bA}_F^2$.
        Hence, $\norm{\paran{\tilde\bLam - \bLam\bH^\top}^\top \paran{\tilde\bLam - \bLam\bH^\top}}_{\Sigma_\ve}^2 = \Op{N^{-1}}\cdot\Op{\norm{\tilde\bLam - \bLam\bH^\top}_F^4} =  \Op{N \omega_{NTK}^4}$.
        \item The same argument for the proof of Theorem 2 in \cite{fan2008high} implies that $\norm{\bLam^\top\bSigma_\eps^{-1}\bLam} = \bigO{1}$.
        Thus using the conclusion of (i)
        \begin{equation*}
            \begin{split}
            \norm{\bLam\bH^\top \paran{\tilde\bLam - \bLam\bH^\top}^\top}_{\Sigma_\ve}^2 & = \norm{\bLam\bH^\top \bC_T^\top}_{\Sigma_\ve}^2 = N^{-1} \Tr\paran{ \bH^\top \bC_T^\top \Sigma_{\eps}^{-1} \bC_T \bH \bLam^\top \Sigma_{\eps}^{-1} \bLam} \\
            & \le  N^{-1} \norm{ \bH }^2 \norm{\bLam^\top \Sigma_{\eps}^{-1} \bLam} \norm{\bC_T}_F^2 \\
            & = \Op{ N^{-1} \norm{\bC_T}_F^2} \\
            & = \Op{\omega_{NTK}^2}.
            \end{split}
        \end{equation*}

        \item Again, by $\norm{\bLam^\top\bSigma_\eps^{-1}\bLam} = \bigO{1}$ and Lemma \ref{thm:H},
        \begin{equation*}
        \begin{split}
        \norm{\bLam (\bH^\top\bH - \bI_K) \bLam^\top}_{\Sigma_\ve}^2 & = N^{-1} \Tr\paran{ (\bH^\top\bH - \bI_K) \bLam^\top \Sigma_{\eps}^{-1} \bLam (\bH^\top\bH - \bI_K) \bLam^\top \Sigma_{\eps}^{-1} \bLam} \\
        & \le  N^{-1} \norm{ \bH^\top\bH - \bI_K }_F^2 \norm{\bLam^\top \Sigma_{\eps}^{-1} \bLam}^2\\
        & = \Op{ N^{-1} \paran{T^{-1}+N^{-1} + \kappa_{NT}^2 K^2}} .
        \end{split}
        \end{equation*}

    \end{enumerate}
\end{proof}

\begin{lemma} \label{thm:cov-signal-utils-2}
    Denote $\omega_{NTK} = \frac{1}{\sqrt{N}} + \sqrt{\frac{\log(N)}{T}} + \kappa_{NT}K$,  we have
    \[
      \norm{ \tilde\bLam^\top \tilde\bSigma_{e}^{-1} \tilde\bLam - (\bLam\bH^\top)^\top \bSigma_{e}^{-1} (\bLam\bH^\top)} = \Op{N  \omega_{NTK} }.
    \]
\end{lemma}
\begin{proof}
    By the results of Corollary \ref{thm:lam-f-max-bounds} (ii) and Corollary \ref{thm:idio-cov-bound}, the proof is the same as Lemma 14 in \cite{fan2013large}.
    Our result corresponds to the special case where $q = 0$ and $m_p$ is a constant in Lemma 14 in \cite{fan2013large}.
\end{proof}

\begin{lemma} \label{thm:cov-signal-utils-3}
    Denote $\omega_{NTK} = \frac{1}{\sqrt{N}} + \sqrt{\frac{\log(N)}{T}} + \kappa_{NT}K$. If $\omega_{NTK} = \smlo{1}$, then with probability approaching 1, for some $c > 0$,
    \begin{enumerate}[label=(\roman*)]
        \item $\lambda_{\min} \braces{ \bI_K + (\bLam\bH^\top)^\top \bSigma_{e}^{-1} \bLam\bH^\top } \ge cN$.
        \item $\lambda_{\min} \braces{ \bI_K + \tilde\bLam^\top \tilde\bSigma_{e}^{-1} \tilde\bLam } \ge cN$
        \item $\lambda_{\min} \braces{ \bI_K + \bLam^\top \bSigma_{e}^{-1} \bLam } \ge cN$
        \item $\lambda_{\min} \braces{ (\bH\bH^\top)^{-1} + \bLam^\top \bSigma_{e}^{-1} \bLam } \ge cN$
    \end{enumerate}
\end{lemma}

\begin{proof}
    Our result corresponds to the special case where $q = 0$ and $m_p$ is a constant in Lemma 15 in \cite{fan2013large}.
    The proof is the same and thus omitted here.
\end{proof}

\subsubsection{Proof of Theorem \ref{thm:resid-cov-bound}}\label{appen_cov}


\begin{proof}
    Denote $\omega_{NTK} = \frac{1}{\sqrt{N}} + \sqrt{\frac{\log(N)}{T}}+\kappa_{NT}K$,
    \begin{enumerate}[label=(\roman*)]
        \item By Lemma \ref{thm:cov-signal-utils-1},
        \begin{equation*}
        \begin{split}
         & \norm{\bLam (\bH^\top\bH - \bI_K) \bLam^\top}_{\Sigma_\eps}^2 + \norm{\bLam\bH^\top \paran{\tilde\bLam - \bLam\bH^\top}^\top}_{\Sigma_\eps}^2 + \norm{\paran{\tilde\bLam - \bLam\bH^\top}^\top \paran{\tilde\bLam - \bLam\bH^\top}}_{\Sigma_\eps}^2 \\
         & = \Op{N^{-1} \paran{T^{-1}+N^{-1} + \kappa_{NT}^2 K^2} + \omega_{NTK}^2 + N \omega_{NTK}^4} \\
         & = \Op{ N^{-1} + \log(N)/T + N \log(N)^2/T^2 + \kappa_{NT}^2K^2 + N\kappa_{NT}^4K^4}.
        \end{split}
        \end{equation*}
        Hence for a generic constant $C>0$, by Corollary \ref{thm:idio-cov-bound}, we have
        \begin{equation*}
        \begin{split}
        \norm{\tilde{\bSigma}_{\eps} - \bSigma_{\eps}}_{\bSigma_{\eps}}^2 & \le C \norm{ \tilde\bLam\tilde\bLam^\top - \bLam\bLam^\top }_{\Sigma_\eps}^2 + C \norm{\tilde{\bSigma}_e - \bSigma_e}_{\Sigma_\eps}^2 \\
         & \le C \norm{\bLam (\bH^\top\bH - \bI_K) \bLam^\top}_{\Sigma_\eps}^2 + C \norm{\bLam\bH^\top \paran{\tilde\bLam - \bLam\bH^\top}^\top}_{\Sigma_\eps}^2 \\
         & + C \norm{\paran{\tilde\bLam - \bLam\bH^\top}^\top \paran{\tilde\bLam - \bLam\bH^\top}}_{\Sigma_\eps}^2 + C \norm{\tilde{\bSigma}_e - \bSigma_e}_{\Sigma_\eps}^2 \\
        & = \Op{N^{-1} + \log(N)/T + N \log(N)^2/T^2 + \kappa_{NT}^2K^2 + N\kappa_{NT}^4K^4}.
        \end{split}.
        \end{equation*}
         {Using the fact that $2ab \le a^2 + b^2$, we obtain the final results by throwing away cross terms  $\log(N)/T=\paran{ N^{-1} \cdot N \log(N)^2/T^2 }^{1/2}$ and $\kappa_{NT}^2K^2 = \paran{N^{-1}\cdot N\kappa_{NT}^4K^4}^{1/2}$. }
         \item Define $\breve\bSigma_\eps = \bLam\bH^\top\bH\bLam^\top + \bSigma_e$.
         Note that $\tilde\bSigma_\eps = \tilde\bLam\tilde\bLam^\top + \tilde\bSigma_{e}$ and $\bSigma_{\eps} = \bLam\bLam^\top + \bSigma_e$.
         The triangular inequality gives
         \[
           \norm{\tilde{\bSigma}_{\eps}^{-1} - \bSigma_{\eps}^{-1}} \le \norm{\tilde{\bSigma}_{\eps}^{-1} - \breve\bSigma_{\eps}^{-1}} + \norm{\breve{\bSigma}_{\eps}^{-1} - \bSigma_{\eps}^{-1}}.
         \]

         {\sc Step 1.} Firstly, we bound $\norm{\tilde{\bSigma}_{\eps}^{-1} - \breve\bSigma_{\eps}^{-1}}$.

         Let $\bG = \braces{\bI_K + \tilde\bLam^\top \tilde{\bSigma}_e^{-1} \tilde\bLam}^{-1}$ and
         $\bG_1 = \braces{\bI_K + (\bLam\bH^\top)^\top \tilde{\bSigma}_e^{-1} (\bLam\bH^\top) }^{-1}$.

         By Sherman-Morrison-Woodbury formula, we have $\tilde\Sigma_\ve^{-1} = \tilde\Sigma_e^{-1} - \tilde\Sigma_e^{-1}\tilde\bLam(\bI_K +
         \tilde\bLam^\top \tilde\Sigma_e^{-1}\tilde\bLam)^{-1}\tilde\bLam^\top\tilde\Sigma_e^{-1}$
         and $\breve \Sigma_\ve^{-1} = \Sigma_e^{-1} - \Sigma_e^{-1}\blam\bH^\top\{\bI_K +
         (\tilde\bLam\bH^\top)^\top \tilde\Sigma_e^{-1}(\tilde\bLam\bH^\top)\}^{-1}\tilde((\tilde\bLam\bH^\top))^\top\Sigma_e^{-1}$.
         Hence it leads to $\norm{\tilde{\bSigma}_{\eps}^{-1} - \breve\bSigma_{\eps}^{-1}} \le \sum_{i = 1}^6 L_i$, where
         \begin{equation*}
             \begin{split}
             L_1 & = \norm{\tilde{\bSigma}_e^{-1} - \bSigma_e^{-1}} \\
             L_2 & = \norm{ ( \tilde\bSigma_{e}^{-1} - \bSigma_{e}^{-1} ) \tilde\bLam \bG \tilde\bLam^\top \tilde\bSigma_{e}^{-1}} \\
             L_3 & = \norm{ ( \tilde\bSigma_{e}^{-1} - \bSigma_{e}^{-1} ) \tilde\bLam \bG \tilde\bLam^\top \bSigma_{e}^{-1}} \\
             L_4 & = \norm{ \bSigma_{e}^{-1} (\tilde\bLam - \bLam\bH^\top) \bG \tilde\bLam^\top \bSigma_{e}^{-1}} \\
             L_5 & = \norm{ \bSigma_{e}^{-1} (\tilde\bLam - \bLam\bH^\top) \bG \bH\bLam^\top \bSigma_{e}^{-1}} \\
             L_6 & = \norm{ \bSigma_{e}^{-1} \bLam\bH^\top (\bG - \bG_1) \bH\bLam^\top \bSigma_{e}^{-1}}.
             \end{split}
         \end{equation*}
         We bound each of the six terms.
         First, $L_1$ is bounded in Corollary \ref{thm:idio-cov-bound} (ii): $\norm{\tilde{\bSigma}_e^{-1} - \bSigma_e^{-1}} = \Op{\omega_{NTK}}$.
         Then,
         \[
         L_2 \le \norm{ \tilde\bSigma_{e}^{-1} - \bSigma_{e}^{-1}} \norm{ \tilde\bLam \bG \tilde\bLam^\top } \norm{\bSigma_{e}^{-1}} = \Op{L_1},
         \]
          which used the results that $\norm{\bSigma_{e}^{-1}} = O(1)$ and  $\norm{ \bG  } = \Op{N^{-1}} $ by Lemma \ref{thm:cov-signal-utils-3}.
          Similarly, we have $L_3 = \Op{L_1}$.
          Since $\norm{\tilde\bLam - \bLam\bH^\top}_F^2 = \Op{N\omega_{NTK}}$
          Then,
          \[
          L_4 \le \norm{ \bSigma_{e}^{-1} (\tilde\bLam - \bLam\bH^\top)} \norm{\bG} \norm{\tilde\bLam^\top \bSigma_{e}^{-1}} = \Op{\omega_{NTK}}.
          \]
          Similarly, $L_5 = \Op{L_4}$.
          Finally, by Lemma \ref{thm:cov-signal-utils-3}, $\norm{\bG_1} = \Op{N^{-1}}$.
          Then, by Lemma \ref{thm:cov-signal-utils-2},
          \begin{equation*}
              \begin{split}
               \norm{\bG - \bG_1} & = \norm{\bG(\bG^{-1} - \bG_1^{-1})\bG_1}  \le \Op{ N^{-2} } \norm{ (\bLam\bH^\top)^\top\bSigma_{e}^{-1}\bLam\bH^\top - \tilde\bLam^\top \tilde\bSigma_{e}^{-1} \tilde\bLam  } \\
               & = \Op{N^{-1} \omega_{NTK} }.
              \end{split}
          \end{equation*}
          Consequently, $L_6 \le \norm{ \bSigma_{e}^{-1} \bLam\bH^\top}^2 \norm{ \bG - \bG_1 }^2 = \Op{\omega_{NTK}}$.
          Adding up $L_1$ to $L_6$ gives,
          \[
          \norm{\tilde{\bSigma}_{\eps}^{-1} - \breve\bSigma_{\eps}^{-1}} = \Op{\omega_{NTK} }.
          \]
          {\sc Step 2.} Now we turn to $\norm{\breve{\bSigma}_{\eps}^{-1} - \bSigma_{\eps}^{-1}}$.
          Again, using the Sherman-Morrison-Woodbury formula, we have
          \begin{equation*}
              \begin{split}
              \norm{\breve{\bSigma}_{\eps}^{-1} - \bSigma_{\eps}^{-1}} & \le \norm{ \bSigma_{e}^{-1}\bLam\paran{ \paran{(\bH^\top\bH)^{-1} + \bLam^\top\bSigma_{e}^{-1} \bLam }^{-1} - \paran{\bI_k + \bLam^\top\bSigma_{e}^{-1} \bLam}^{-1} } \bLam^\top \bSigma_{e}^{-1}} \\
              & \le \bigO{N} \norm{ \paran{(\bH^\top\bH)^{-1} + \bLam^\top\bSigma_{e}^{-1} \bLam }^{-1} - \paran{\bI_k + \bLam^\top\bSigma_{e}^{-1} \bLam}^{-1} } \\
              & = \Op{N^{-1}} \norm{ (\bH^\top\bH)^{-1} - \bI_k } \\
              & = \Op{N^{-1}\paran{T^{-1/2}+N^{-1/2} + \kappa_{NT} K}} \\
              & = \op{\omega_{NTK}  }
              \end{split}
          \end{equation*}
        where the first equation is due to that $\paran{(\bH^\top\bH)^{-1} + \bLam^\top\bSigma_{e}^{-1} \bLam}^{-1} - \paran{\bI_k + \bLam^\top\bSigma_{e}^{-1} \bLam}^{-1}  =
        \paran{(\bH^\top\bH)^{-1} + \bLam^\top\bSigma_{e}^{-1} \bLam}^{-1}\{(\bH^\top\bH - \bI_K)\}\paran{\bI_k + \bLam^\top\bSigma_{e}^{-1} \bLam}^{-1}$
        and $\|\paran{(\bH^\top\bH)^{-1} + \bLam^\top\bSigma_{e}^{-1} \bLam}^{-1}\| = \Op{N^{-1}}$,
        $\|\paran{\bI_k + \bLam^\top\bSigma_{e}^{-1} \bLam}^{-1}\| = \Op{N^{-1}}$ by Lemma \ref{thm:cov-signal-utils-3}.

    \end{enumerate}
\end{proof}

\section{Additional simulation results}

\subsection{Different network models}

\begin{enumerate}[label={\bf \sc Example \arabic*}, align=left, wide, labelwidth=!, labelindent=0pt]

    \setcounter{enumi}{3}

    \item ({\bf Clusters of Powerlaw graphs and CNAR})
    In this section, we consider a network generated by Holme and Kim algorithm \citep{holme2002growing} for growing graphs with powerlaw degree distribution.
    Specifically, in this type of network, $K$ communities are generated the degrees of the nodes follow a powerlaw distribution.
    As a result, it assumes that the majority of nodes
    have few connections but a small proportion have a large amount of connections \citep{barabasi1999emergence,Clauset:2009}.
    Numerically, the graph can be generated by the Holme and Kim algorithm implemented using the power law graph function \textit{powerlaw\_cluster\_graph} in the Python NetworkX.
    For each node of the $k$-th community, we set the expected degree as $m=0.8\cdot n_k$, where $n_k$ is the number of nodes in the $k$-th community.
    Subsequently, we randomly add $n_{ij}$ edges between the $i$-th and $j$-th communities for all $i,j\in [K]$ with $n_{ij} \sim {\rm Uniform}(0,10)$.
    The membership matrix $\bTheta$ is generated according to the original $K$ disjoint community assignment.
    Other model parameters, such as $\bB$, $\bbeta_2$ and $\gamma$, are set in the same way as in Example 1.
    Under this setting we focus on the sensitivity analysis when the network is not generated from the SBM model.
    We set $N=400$ and vary $T\in\brackets{50,100,\cdots,450}$.
    For CNAR, we set $K=2$.

    \begin{figure}[htpb!]
        \centering
        \begin{subfigure}[b]{0.35\textwidth}
            \centering
            \includegraphics[width=\linewidth]{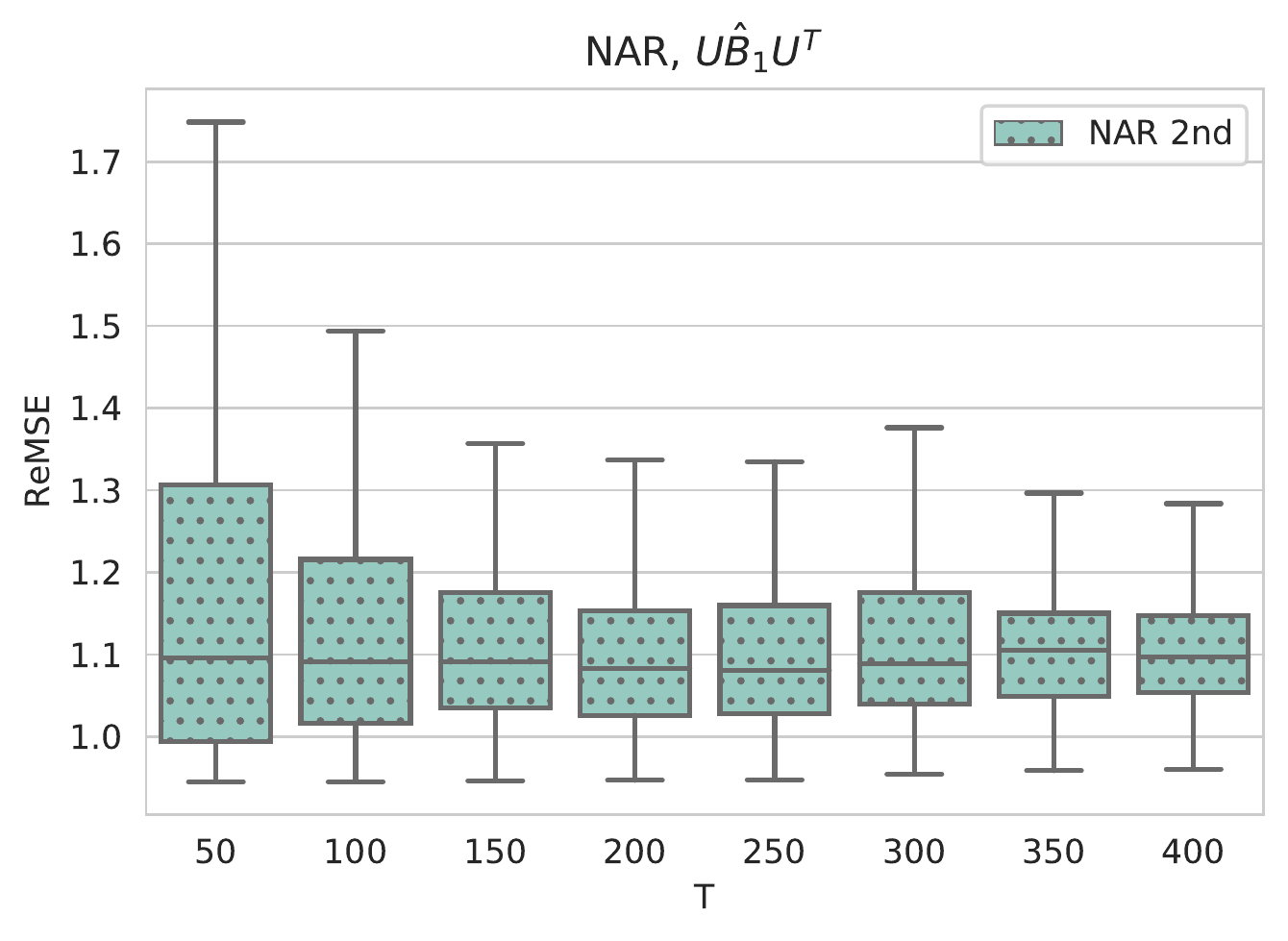}
        \end{subfigure}
        \hspace{3ex}
        \begin{subfigure}[b]{0.35\textwidth}
            \centering
            \includegraphics[width=\linewidth]{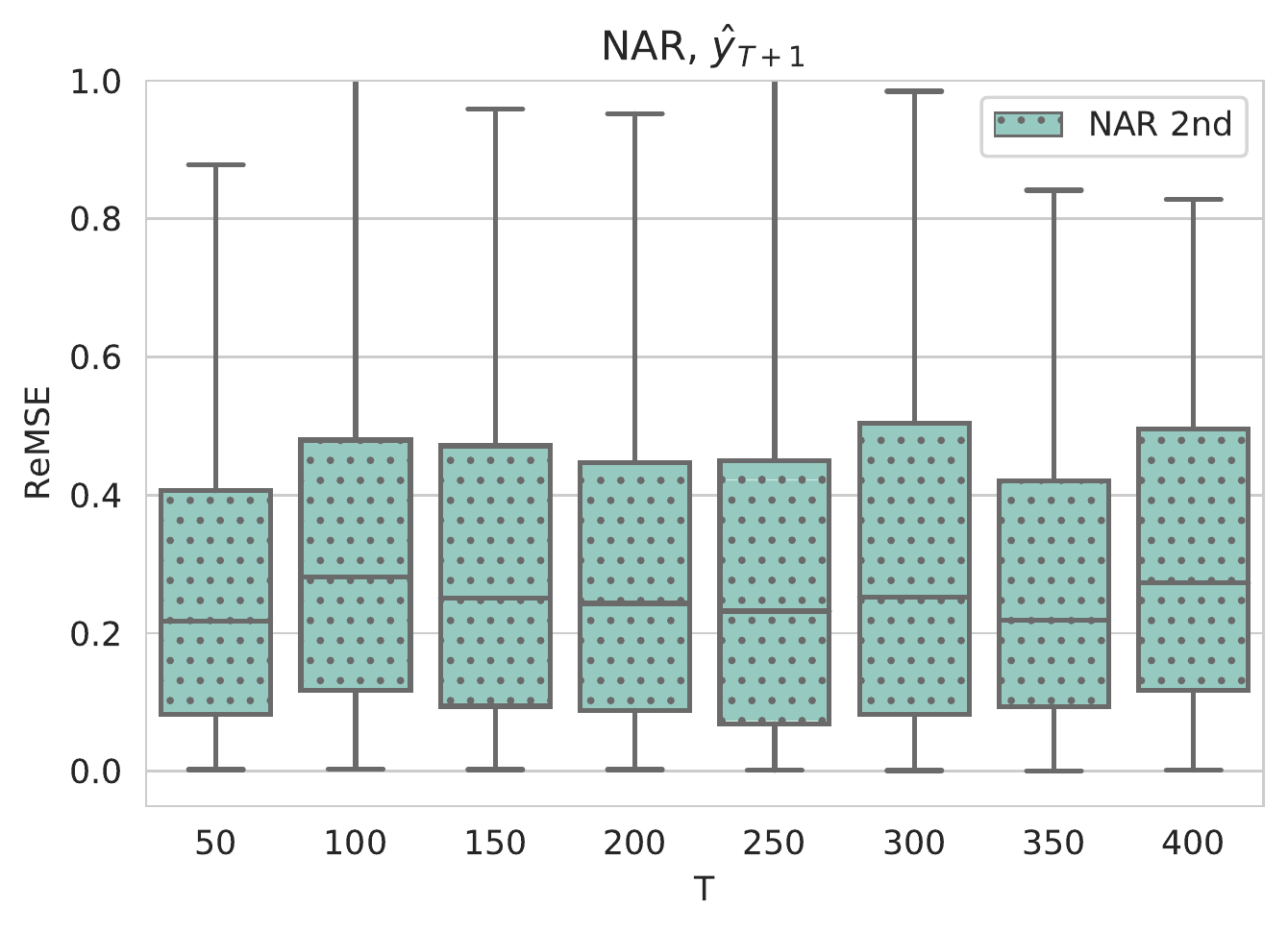}
        \end{subfigure}

        \begin{subfigure}[b]{0.35\textwidth}
            \centering
            \includegraphics[width=\linewidth]{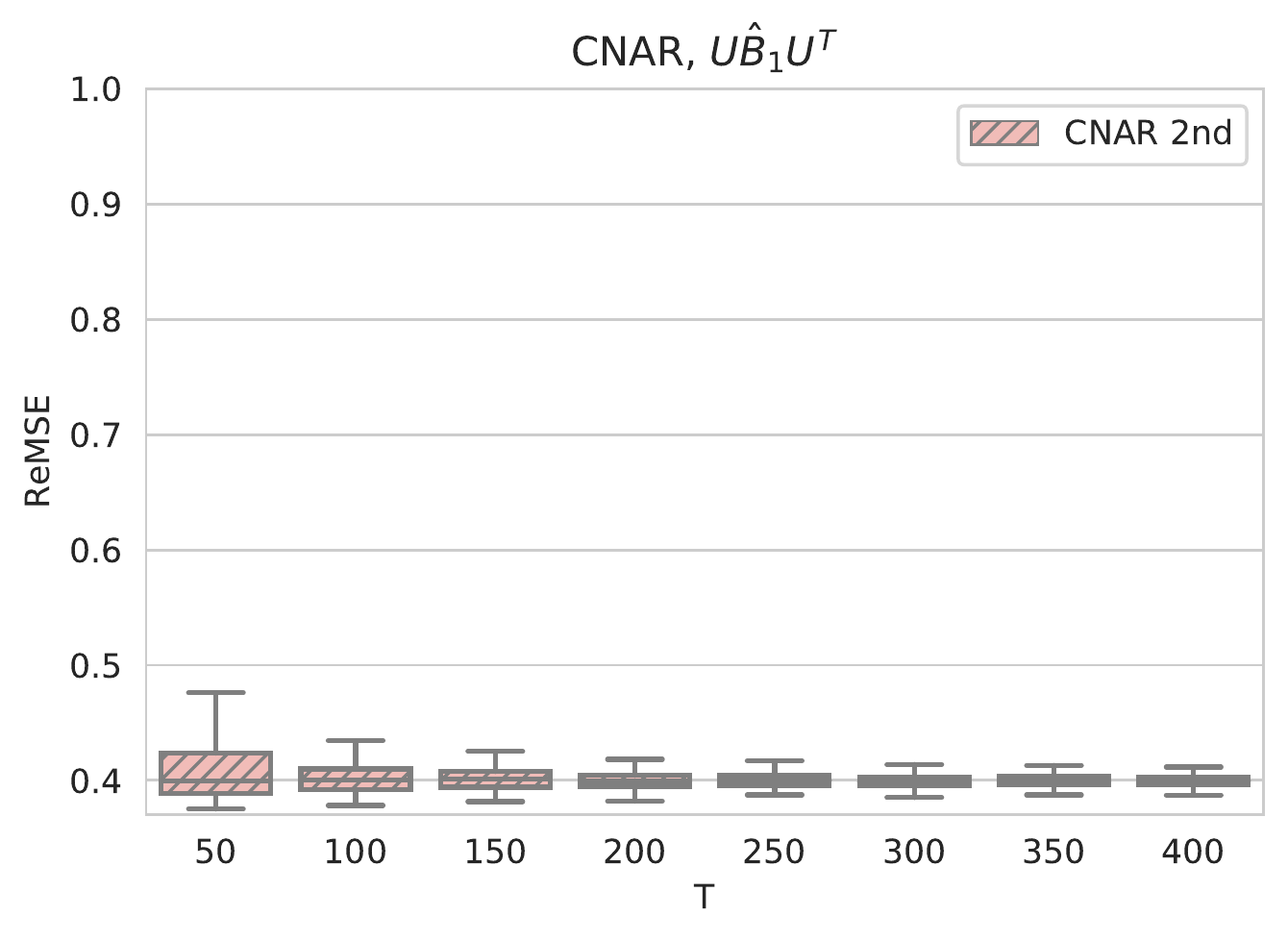}
        \end{subfigure}
        \hspace{3ex}
        \begin{subfigure}[b]{0.35\textwidth}
            \centering
            \includegraphics[width=\linewidth]{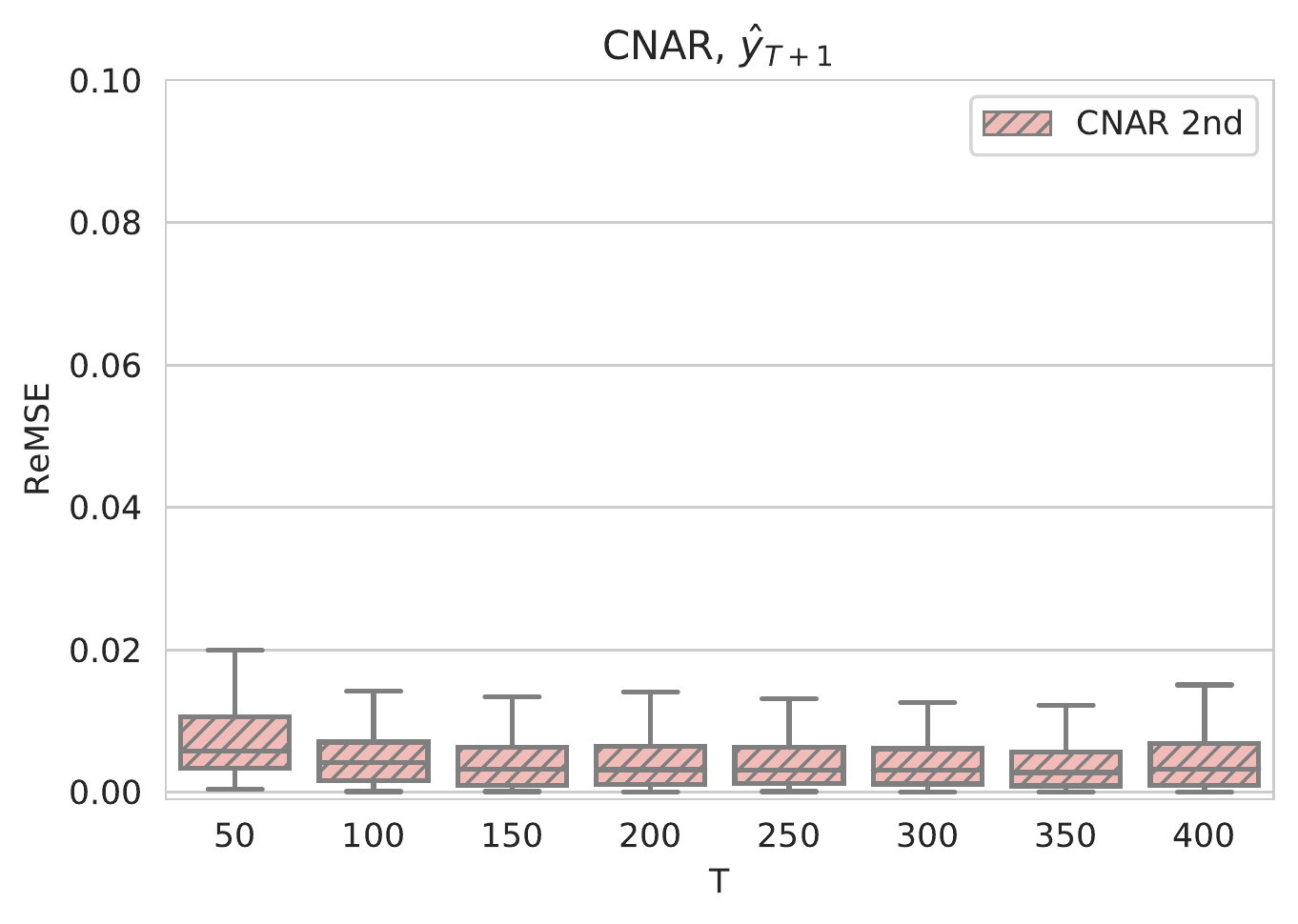}
        \end{subfigure}
        \caption{Example 4 (Power law and CNAR model with $K=2$ and $N=400$) box plots of estimated network AR coefficient $\bU\bB_1\bU^\top$ and 1-step prediction $y_{T+1}$ by NAR (first row) and CNAR (second row), respectively.
        }
        \label{fig:example-3-boxplot}
    \end{figure}

    Figure \ref{fig:example-3-boxplot} shows, for Example 4, the box plots of estimated network AR coefficient $\bU\bB_1\bU^\top$ and 1-step prediction $y_{T+1}$ by NAR (first row) and CNAR (second row), respectively.
    Under this setting with a different network generative model, CNAR still has much smaller  estimation and prediction error compared to the NAR model.

    \item ({\bf Random Partition Graph and CNAR})
    In this example, we consider a network generated by a random partition graph.
    A random partition graph is a graph of communities with different sizes.
    The nodes within the same group are connected with probability $\Pr_{in}$ and nodes of different groups are connected with probability $\Pr_{out}$.
    The network is a generalization of the planted-l-partition described in \cite{fortunato2010community}.
    Under this setting we focus on the sensitivity analysis when the network is not generated from the SBM model.
    We fix $\Pr_{in}=0.9$, $\Pr_{out}=0.1$, set $N=400$ and vary $T\in\brackets{50,100,\cdots,450}$.
    For CNAR, we set $K=2$.

    \begin{figure}[htpb!]
        \centering
        \begin{subfigure}[b]{0.35\textwidth}
            \centering
            \includegraphics[width=\linewidth]{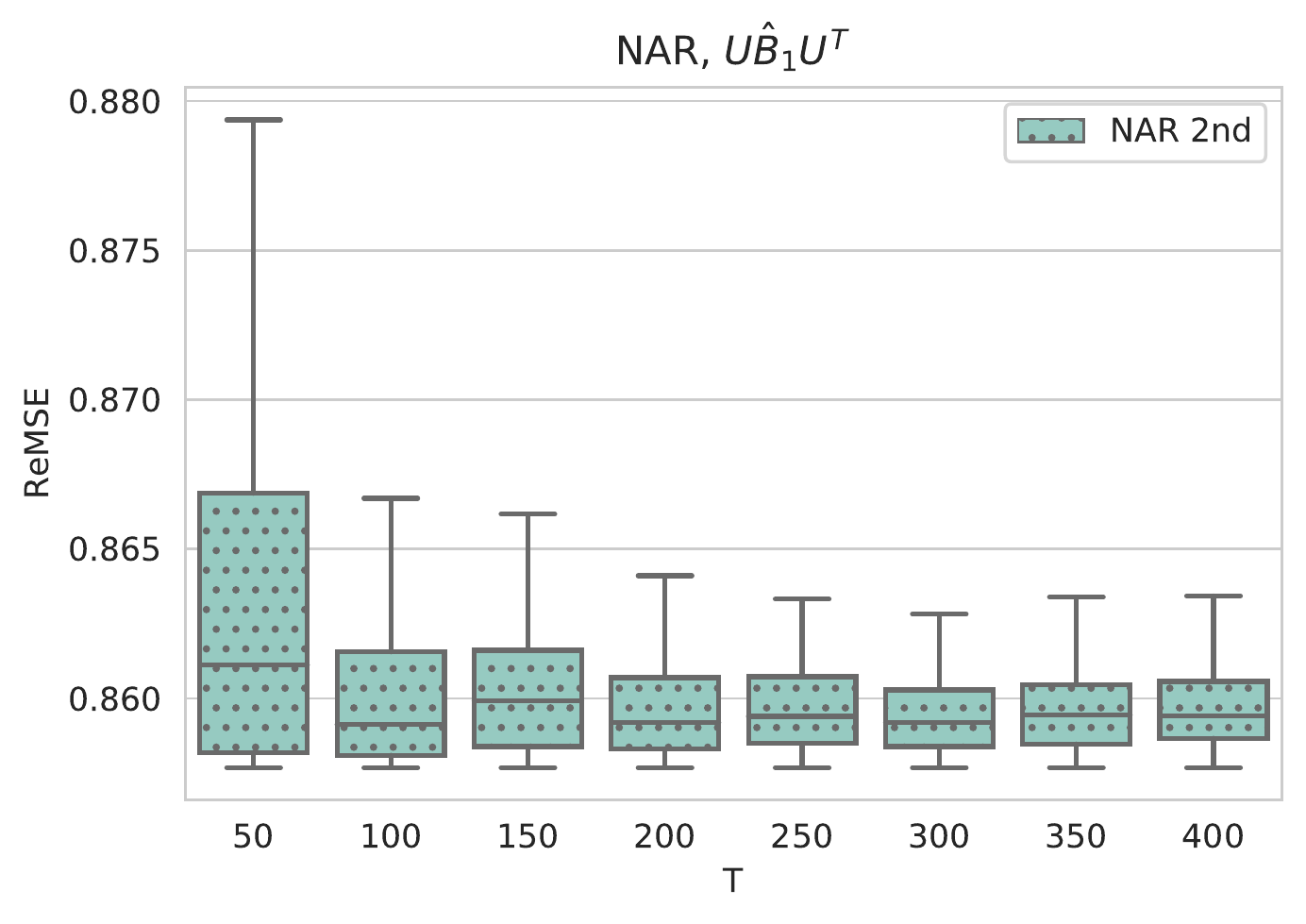}
        \end{subfigure}
        \hspace{3ex}
        \begin{subfigure}[b]{0.35\textwidth}
            \centering
            \includegraphics[width=\linewidth]{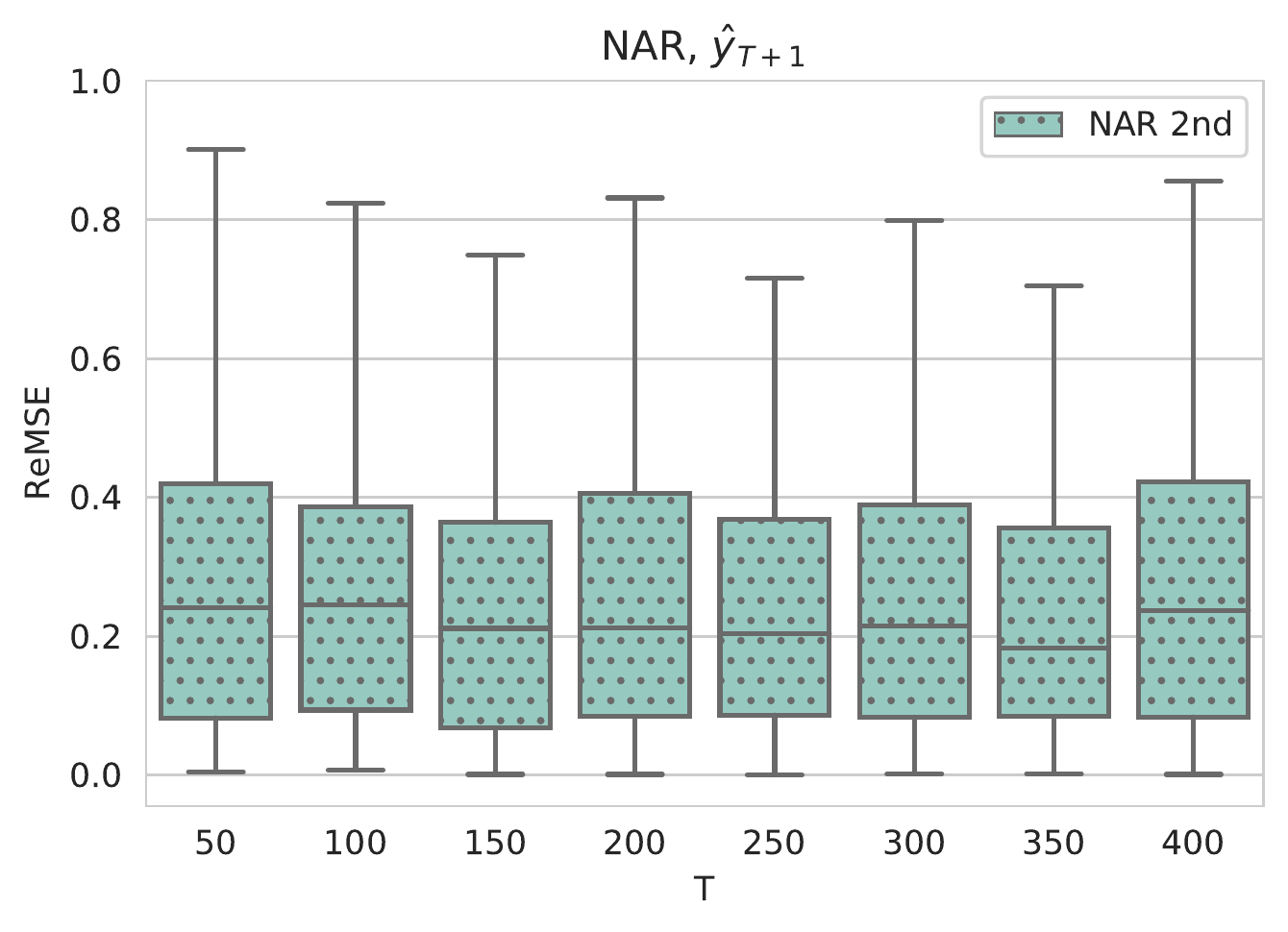}
        \end{subfigure}

        \begin{subfigure}[b]{0.35\textwidth}
            \centering
            \includegraphics[width=\linewidth]{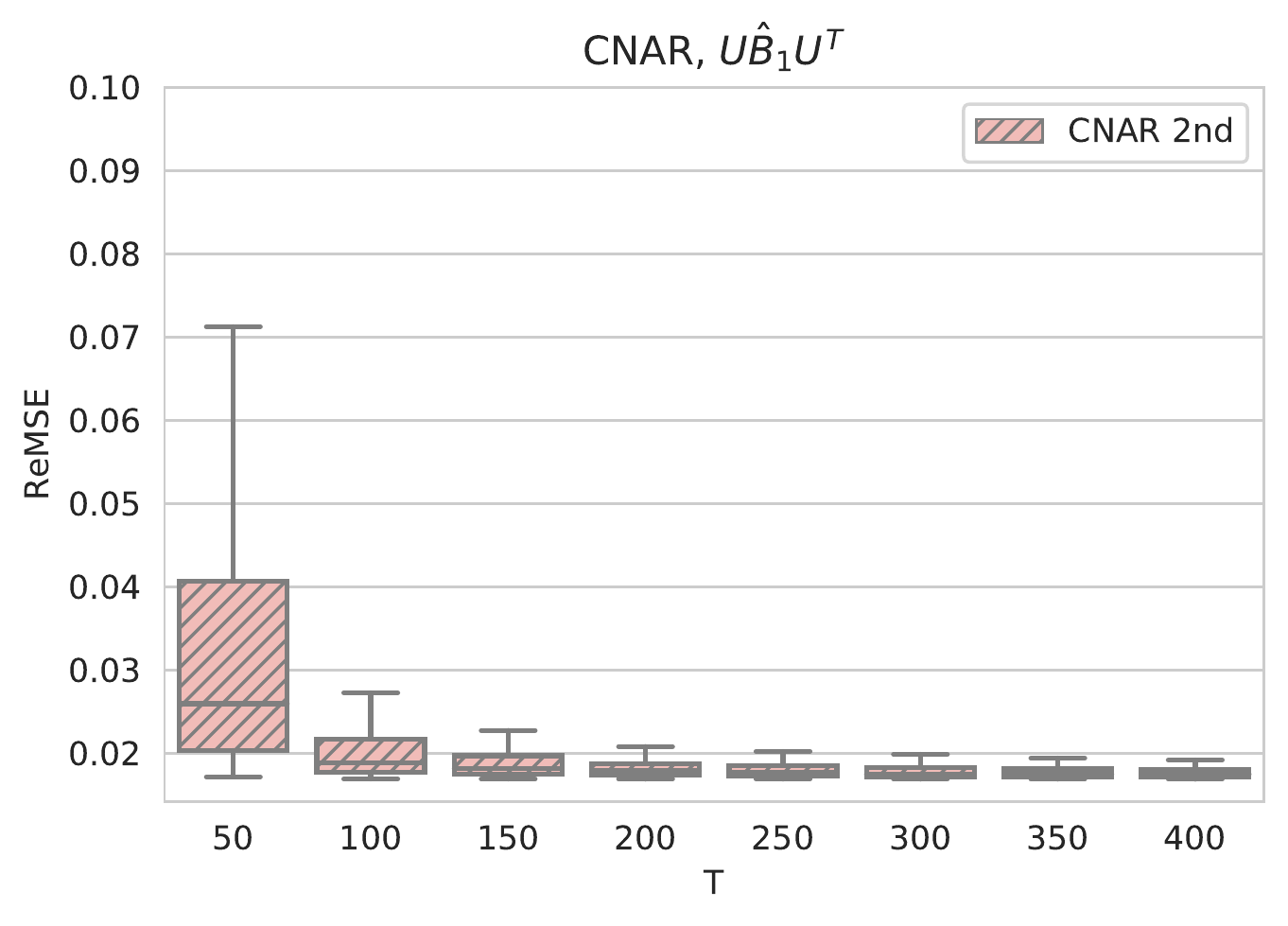}
        \end{subfigure}
        \hspace{3ex}
        \begin{subfigure}[b]{0.35\textwidth}
            \centering
            \includegraphics[width=\linewidth]{figs/cnar-other-cnar-N-400-K-2-pred-0-v2.pdf}
        \end{subfigure}
        \caption{Example 3 ($K=2$ and $N=400$) box plots of estimated network AR coefficient $\bU\bB_1\bU^\top$ and 1-step prediction $y_{T+1}$ by NAR (first row) and CNAR (second row), respectively.
        }
        \label{fig:example-3-boxplot}
    \end{figure}

    Figure \ref{fig:example-3-boxplot} shows, for Example 3, the box plots of estimated network AR coefficient $\bU\bB_1\bU^\top$ and 1-step prediction $y_{T+1}$ by NAR (first row) and CNAR (second row), respectively.
    Under this setting with a different network generative model, CNAR still has much smaller  estimation and prediction error compared to the NAR model.
\end{enumerate}

\subsection{More Tables}

Table \ref{table:sbm-cnar} compares the performance of CNAR and NAR in Example 1 with different numbers of network communities $K\in\braces{2,4,8}$, network sizes $N\in\braces{200,400,800}$ and time series length $T=100,200,400$.
The performance of CNAR is better than that of NAR under all settings.
For the CNAR model, it can be observed that the errors reduce when $T$ or $N$ with given number of
communities $K$, which corroborates with the theoretical findings.

\begin{table}[htpb!]
    \caption{Means and standard deviations in parantheses (all values $\times 100$) of the ReMSEs of coefficient estimators and time series predictors of CNAR and NAR for $K\in\braces{2,4,8}$, $N\in\braces{200,400,800}$ and $T=100,200,400$.} \label{table:sbm-cnar}
    \centering

    \begin{subtable}[t]{\linewidth}
        \caption{Network community number $K = 2$.}
        \label{table:sbm-cnar-K-2}
        \centering
        \resizebox{\textwidth}{!}{%
            \begin{tabular}{c|c|ccc|ccc|ccc}
                \hline
                \multicolumn{2}{c|}{N} & \multicolumn{3}{c|}{200} & \multicolumn{3}{c|}{400} & \multicolumn{3}{c}{800} \\ \hline
                \multicolumn{2}{c|}{T} & 100 & 200 & 400 & 100 & 200 & 400 & 100 & 200 & 400 \\ \hline
                \multirow{5}{*}{CNAR} & $\bU\bB_1\bU$ & 1.028(0.826) & 0.692(0.235) & 0.573(0.109) & 0.671(0.64) & 0.378(0.175) & 0.282(0.064) & 0.423(0.574) & 0.188(0.111) & 0.132(0.038) \\
                & $\beta_2$ & 0.087(0.115) & 0.047(0.066) & 0.028(0.037) & 0.062(0.086) & 0.022(0.028) & 0.011(0.016) & 0.028(0.039) & 0.012(0.019) & 0.005(0.008) \\
                & $\gamma$ & 0.075(0.049) & 0.033(0.019) & 0.017(0.011) & 0.039(0.027) & 0.02(0.013) & 0.008(0.005) & 0.023(0.019) & 0.009(0.006) & 0.004(0.003) \\
                & $\bLam$ & 0.393(0.186) & 0.16(0.051) & 0.067(0.013) & 0.436(0.183) & 0.168(0.042) & 0.073(0.012) & 0.436(0.233) & 0.168(0.052) & 0.075(0.013) \\ \cline{2-11}
                & $\hat \by_{T+1}$ & 0.222(0.199) & 0.158(0.125) & 0.123(0.093) & 0.133(0.143) & 0.068(0.054) & 0.055(0.042) & 0.068(0.065) & 0.036(0.033) & 0.024(0.018) \\ \hline
                \multirow{5}{*}{NAR} & $\bU\bB_1\bU$ & 79.304(0.95) & 78.978(0.442) & 78.84(0.267) & 79.44(1.129) & 78.988(0.391) & 78.858(0.206) & 79.263(0.945) & 78.901(0.358) & 78.768(0.17) \\
                & $\beta_2$ & 1.18(1.617) & 0.609(0.87) & 0.323(0.417) & 1.012(1.341) & 0.408(0.534) & 0.197(0.238) & 0.84(1.137) & 0.427(0.562) & 0.29(0.392) \\
                & $\gamma$ & 0.128(0.09) & 0.063(0.039) & 0.034(0.024) & 0.062(0.045) & 0.029(0.019) & 0.013(0.008) & 0.049(0.043) & 0.026(0.019) & 0.016(0.01) \\
                & $\bLam$ & 131.969(33.141) & 135.776(32.599) & 132.931(29.923) & 147.165(47.26) & 162.518(46.868) & 149.743(47.375) & 148.61(43.089) & 152.632(43.682) & 148.676(42.989) \\ \cline{2-11}
                & $\hat \by_{T+1}$ & 29.393(22.326) & 29.081(22.5) & 26.765(20.776) & 25.45(18.84) & 24.625(19.555) & 24.977(19.662) & 23.822(19.763) & 26.169(19.523) & 22.901(18.158) \\ \hline
            \end{tabular}%
        }
    \end{subtable}

    \rule{0pt}{1ex}

    \begin{subtable}[t]{\linewidth}
        \caption{Network community number $K = 4$.}
        \label{table:sbm-cnar-K-4}
        \centering
        \resizebox{\textwidth}{!}{%
            \begin{tabular}{c|c|ccc|ccc|ccc}
                \hline
                \multicolumn{2}{c|}{N} & \multicolumn{3}{c|}{200} & \multicolumn{3}{c|}{400} & \multicolumn{3}{c}{800} \\ \hline
                \multicolumn{2}{c|}{T} & 100 & 200 & 400 & 100 & 200 & 400 & 100 & 200 & 400 \\ \hline
                \multirow{5}{*}{CNAR} & $\bU\bB_1\bU$ & 6.438(4.782) & 3.699(1.463) & 2.781(0.67) & 12.215(13.866) & 3.751(3.269) & 2.287(1.191) & 6.715(7.052) & 2.808(3.018) & 1.44(0.751) \\
                & $\beta_2$ & 0.501(0.527) & 0.459(0.348) & 0.458(0.276) & 0.242(0.286) & 0.166(0.153) & 0.148(0.099) & 0.058(0.066) & 0.036(0.042) & 0.033(0.029) \\
                & $\gamma$ & 0.083(0.053) & 0.034(0.021) & 0.018(0.012) & 0.047(0.037) & 0.021(0.013) & 0.009(0.006) & 0.027(0.024) & 0.01(0.006) & 0.004(0.003) \\
                & $\bLam$ & 0.579(0.299) & 0.221(0.093) & 0.087(0.022) & 0.584(0.274) & 0.203(0.073) & 0.081(0.019) & 0.547(0.296) & 0.204(0.087) & 0.086(0.022) \\ \cline{2-11}
                & $\hat y_{T+1}$ & 0.55(0.321) & 0.466(0.205) & 0.435(0.194) & 0.328(0.363) & 0.24(0.088) & 0.225(0.089) & 0.146(0.118) & 0.11(0.037) & 0.098(0.038)) \\ \hline
                NAR & $\bU\bB_1\bU$ & 118.235(11.852) & 120.618(7.064) & 120.764(5.263) & 117.957(10.267) & 118.787(6.804) & 118.434(3.579) & 117.439(11.364) & 118.106(7.028) & 119.381(4.248) \\
                & $\beta_2$ & 125.762(62.419) & 163.337(48.24) & 177.707(36.655) & 146.252(69.925) & 167.32(48.436) & 182.017(35.401) & 133.572(62.174) & 152.431(44.752) & 168.377(33.508) \\
                & $\gamma$ & 0.626(0.529) & 0.672(0.366) & 0.7(0.342) & 0.743(0.687) & 0.75(0.462) & 0.77(0.32) & 0.476(0.493) & 0.498(0.413) & 0.549(0.288) \\
                & $\bLam$ & 154.012(44.686) & 165.392(40.281) & 173.069(37.539) & 140.849(44.641) & 141.83(42.436) & 147.844(43.091) & 149.635(45.276) & 148.763(41.65) & 146.915(42.255) \\ \cline{2-11}
                & $\hat \by_{T+1}$ & 70.517(79.946) & 78.527(77.728) & 81.192(93.785) & 74.014(82.83) & 76.315(80.286) & 96.753(105.348) & 61.656(61.865) & 64.156(62.587) & 80.163(84.228) \\ \hline
            \end{tabular}%
        }
    \end{subtable}%

    \rule{0pt}{1ex}

    \begin{subtable}[t]{\linewidth}
        \caption{Network community number $K = 8$.}
        \label{table:sbm-cnar-K-4}
        \centering
        \resizebox{\textwidth}{!}{%
            \begin{tabular}{cc|ccc|ccc|ccc}
                \hline
                \multicolumn{2}{c|}{N} & \multicolumn{3}{c|}{200} & \multicolumn{3}{c|}{400} & \multicolumn{3}{c}{800} \\ \hline
                \multicolumn{2}{c|}{T} & 100 & 200 & 400 & 100 & 200 & 400 & 100 & 200 & 400 \\ \hline
                \multicolumn{1}{c|}{\multirow{5}{*}{CNAR}} & $\bU\bB_1\bU$ & 34.058(12.94) & 20.223(4.597) & 14.838(1.848) & 47.738(52.773) & 16.928(5.854) & 10.103(1.622) & 67.223(72.426) & 16.054(8.291) & 7.117(1.736) \\
                \multicolumn{1}{c|}{} & $\beta_2$ & 7.882(5.952) & 8.632(4.437) & 8.414(3.188) & 1.624(1.528) & 1.314(0.876) & 1.27(0.564) & 0.569(0.526) & 0.522(0.368) & 0.549(0.243) \\
                \multicolumn{1}{c|}{} & $\gamma$ & 0.096(0.068) & 0.042(0.025) & 0.021(0.013) & 0.052(0.042) & 0.023(0.014) & 0.009(0.006) & 0.033(0.029) & 0.011(0.007) & 0.005(0.003) \\
                \multicolumn{1}{c|}{} & $\bLam$ & 0.97(0.55) & 0.338(0.16) & 0.119(0.039) & 1.032(0.587) & 0.315(0.14) & 0.108(0.039) & 0.998(0.543) & 0.307(0.152) & 0.113(0.039) \\ \cline{2-11}
                \multicolumn{1}{c|}{} & $\hat y_{T+1}$ & 3.302(2.062) & 3.285(1.791) & 3.086(1.99) & 1.432(1.539) & 1.102(0.525) & 1.131(0.625) & 0.679(0.355) & 0.531(0.198) & 0.513(0.206) \\ \hline
                \multicolumn{1}{c|}{\multirow{5}{*}{NAR}} & $\bU\bB_1\bU$ & 112.158(6.815) & 115.272(5.37) & 115.143(3.672) & 111.769(7.856) & 113.367(5.904) & 113.878(3.362) & 112.965(9.235) & 113.947(5.687) & 114.929(3.282) \\
                \multicolumn{1}{c|}{} & $\beta_2$ & 162.078(78.495) & 194.58(60.394) & 213.942(47.352) & 128.507(84.004) & 170.411(73.947) & 194.52(49.671) & 148.284(89.314) & 185.514(68.186) & 228.376(49.852) \\
                \multicolumn{1}{c|}{} & $\gamma$ & 0.571(0.471) & 0.576(0.335) & 0.614(0.314) & 0.386(0.376) & 0.329(0.264) & 0.328(0.178) & 0.443(0.447) & 0.49(0.361) & 0.637(0.301) \\
                \multicolumn{1}{c|}{} & $\bLam$ & 173.495(39.55) & 183.232(29.815) & 186.471(25.941) & 159.843(41.978) & 163.451(41.729) & 165.72(40.3) & 158.264(41.08) & 157.368(44.235) & 153.193(44.345) \\ \cline{2-11}
                \multicolumn{1}{c|}{} & $\hat \by_{T+1}$ & 78.506(80.632) & 89.195(88.442) & 85.424(90.699) & 79.499(81.137) & 86.888(89.941) & 113.873(114.174) & 63.753(64.996) & 73.894(74.438) & 95.918(101.764) \\ \hline
            \end{tabular}%
        }
    \end{subtable}%
\end{table}
\end{appendices}

\end{document}